\documentclass[prx,twocolumn,english,superscriptaddress,floatfix,longbibliography]{revtex4-2}

\usepackage[utf8]{inputenc}
\usepackage[T1]{fontenc}
\usepackage{lmodern}
\usepackage{amsmath,amssymb,amsthm}
\usepackage{bm}
\usepackage{physics}
\usepackage{mathtools}
\usepackage{graphicx}
\usepackage{booktabs}
\usepackage{array}
\usepackage{placeins}
\usepackage{xcolor}
\usepackage{hyperref}
\hypersetup{
  pdfdisplaydoctitle = true,
  pdftitle           = {Crystallographic Symmetry Generates Phononic Holonomic Gates with Biased-Erasure Channels},
  pdfauthor          = {El Mustapha Mansouri; Keigo Arai},
  pdfsubject         = {Symmetry-generated phononic Lambda control, holonomic gates, biased-erasure channels, and decoder-aware solid-state quantum information processing},
  pdfkeywords        = {phononic quantum control, holonomic quantum gates, biased-erasure noise, XZZX surface code, quantum error correction, nitrogen-vacancy center, silicon-vacancy center, strain-mediated control, crystallographic symmetry, Lambda systems, superadiabatic transitionless driving, very strong paper, solid-state quantum processors, NV Centers, SiV Centers, quantum error correction, quantum decoding, quantum channels, quantum noise bias, quantum gate design, quantum control, quantum information processing, quantum sensing, quantum metrology},
  colorlinks  = true,
  linkcolor   = blue!70!black,
  citecolor   = green!50!black,
  urlcolor    = blue!80!black,
  bookmarks   = true
}
\usepackage{siunitx}
\DeclareSIUnit{\strain}{strain}
\DeclareSIUnit{\gauss}{G}
\AtBeginDocument{\RenewCommandCopy\qty\SI}
\usepackage{tikz}
\usetikzlibrary{decorations.pathmorphing,arrows.meta}
\graphicspath{{figures/}}

\definecolor{CBblue}{HTML}{648FFF}
\definecolor{CBpink}{HTML}{DC267F}
\definecolor{CBamber}{HTML}{FFB000}
\definecolor{CBorange}{HTML}{FE6100}
\definecolor{CBviolet}{HTML}{785EF0}

\newcommand{\htwofive}{h_{25}}
\newcommand{\htwosix}{h_{26}}
\newcommand{\honesix}{h_{16}}

\newcommand{\Sx}{S_x}
\newcommand{\Sy}{S_y}
\newcommand{\Sz}{S_z}
\newcommand{\exx}{\varepsilon_{xx}}
\newcommand{\eyy}{\varepsilon_{yy}}
\newcommand{\exy}{\varepsilon_{xy}}
\newcommand{\exz}{\varepsilon_{xz}}
\newcommand{\eyz}{\varepsilon_{yz}}

\newcommand{\omD}{\omega_{\mathrm{d}}}
\newcommand{\Omm}{\Omega_{\mathrm{m}}}
\newcommand{\Dbend}{D_{\mathrm{bend}}}
\newcommand{\kmech}{k_{\mathrm{mech}}}
\newcommand{\tauc}{\tau_{\mathrm{c}}}
\newcommand{\ghop}{\gamma_{\mathrm{hop}}}
\newcommand{\tgate}{T_{\mathrm{gate}}}

\newcommand{\aCD}{\alpha_{\mathrm{CD}}}
\newcommand{\Favg}{F_{\mathrm{avg}}}
\newcommand{\Geff}{G_{\mathrm{eff}}}

\newtheorem{theorem}{Theorem}
\newtheorem{proposition}{Proposition}
\newtheorem{corollary}{Corollary}
\theoremstyle{remark}
\newtheorem{remark}{Remark}

\begin{document}

\title{Crystallographic Symmetry Generates Phononic Holonomic Gates with Biased-Erasure Channels}

\author{El Mustapha Mansouri}
\email{mansouri.e.2224@m.isct.ac.jp}
\affiliation{School of Engineering, Institute of Science Tokyo, Yokohama, Kanagawa, 226-8501, Japan}

\author{Keigo Arai}
\email{arai.k.835f@m.isct.ac.jp}
\affiliation{School of Engineering, Institute of Science Tokyo, Yokohama, Kanagawa, 226-8501, Japan}
\date{\today}

\begin{abstract}
Fault-tolerant solid-state processors require control layers whose
errors are not merely small, but legible to quantum-error-correction
decoders.  We show that crystallographic symmetry can provide such a
layer in strain-active solid-state $\Lambda$ manifolds.  When the
projected strain tensor and $\Lambda$-transition operators share a
multiplicity-one two-dimensional irreducible representation, symmetry
fixes the linear strain interaction to a scalar dot product.  Driving
two phase-locked mechanical modes in quadrature then synthesizes a
circular strain field, enabling complex phononic $\Lambda$-leg control
without local microwave near fields.  On this symmetry-generated
manifold we construct a superadiabatic echo-lune holonomic gate using
$\Lambda$-leg control and a resonant double-quantum counterdiabatic
tone.  Rotating-frame open-system simulations of a nitrogen-vacancy
center give $99.88\%$ conditional average fidelity in
\SI{1.833}{\micro\second}, or $99.40\%$ when leaked population is
counted as error; a separate resonant gigahertz high-overtone bulk
acoustic resonator analysis translates the same Hamiltonian into
Rabi-rate, linewidth, and envelope-tracking requirements.  The same
bright-state structure organizes the residual noise: $A_2$-sector
perturbations are parity-filtered into an optically distinguishable
auxiliary state, whereas transverse $E$-sector faults are echo
suppressed and retained as an explicit decoder stress axis.  The
extracted channel has $0.47\%$ erasure probability and $0.168\%$
residual $Z$ error, with transverse error at the extraction floor.  In
XZZX code-capacity simulations, this biased-erasure model yields a
nominal $64\%$ fit-extrapolated reduction in data-qubit count relative
to an unstructured Rabi baseline, while repeated-round detector-model
diagnostics preserve the nominal distance-$9$ proxy and identify
missed erasures, finite transverse floors, leakage/flag timing, and
strong crosstalk as circuit-level validation limits.  Extensions to
orbital $\Lambda$ systems, single-shot non-Abelian control, and
bright-projector phonon-bus diagnostics identify crystallographic
symmetry as a hardware-level principle for co-designing phononic
actuation, leakage, noise bias, and quantum decoding.
\end{abstract}

\maketitle

\begin{figure*}[t!]
  \centering
  \includegraphics[width=\textwidth]{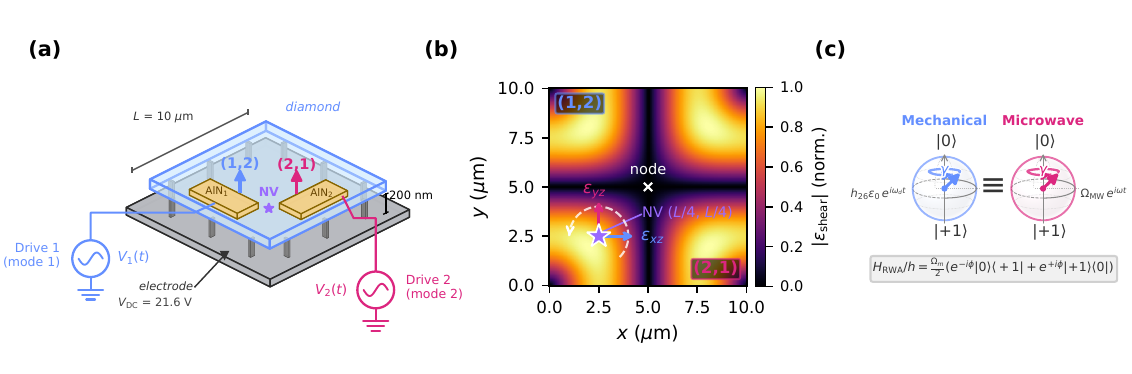}
  \caption{\label{fig:architecture}%
  Synthetic rotation isomorphism and hardware blueprint.
  (a)~Three-dimensional schematic: a simply-supported diamond membrane
  ($10\times 10\times \SI{0.2}{\micro\metre}$) suspended above a
  full-area electrode with a \SI{200}{\nano\metre} vacuum gap.
  Two AlN piezoelectric transducers independently excite the degenerate
  $(1{,}2)$ and $(2{,}1)$ flexural modes.
  (b)~Top-down strain topology showing orthogonal shear components
  $\exz$ and $\eyz$ combining in quadrature at the NV site
  ($L/4,\,L/4$) to produce a circularly rotating strain field.
  (c)~Bloch-sphere comparison: the mechanical rotating strain field
  drives the NV $|0\rangle\!\leftrightarrow\!|{+1}\rangle$ transition
  identically to a circularly polarized microwave field.
  The single-quantum transition shown here illustrates the fundamental
  strain--spin coupling mechanism; the holonomic compilers of
  Sec.~\ref{sec:gate} operate in the
  $\{|{-1}\rangle,|{+1}\rangle\}$ computational subspace via
  $\Lambda$-system traversal through~$|0\rangle$.
  The shear strain amplitude plotted in panel~(b) corresponds
  to $\varepsilon_0$ defined in Eq.~\eqref{eq:rotating_strain}.
  This figure illustrates one quadrature-mode topology; the resonant
  gigahertz high-overtone bulk acoustic resonator (GHz-HBAR)
  implementation is treated separately in Regime~B and
  Appendix~\ref{app:ghz}.  Regime~A is the rotating-frame
  $\Lambda$ + DQ-SATD channel benchmark; Regime~B is the resonant
  GHz-HBAR implementation route.
  }
\end{figure*}

\section{Introduction}
\label{sec:intro}

Fault-tolerant quantum processors require more than high-fidelity
physical gates.  They require control mechanisms whose locality,
native error channels, and leakage signatures are compatible with
decoding.  In most solid-state platforms, these layers are engineered
separately: the Hamiltonian is designed at the device level, while
erasure structure and error bias are imposed later through encoding,
measurement, and decoder choice.  Here we show that, in a broad class
of strain-active crystalline $\Lambda$ manifolds, crystallographic
symmetry can unify these layers.

The concrete motivation is the microwave delivery problem in dense
solid-state registers~\cite{degen2017_rmp,childress2013_review}.
State-of-the-art dynamical gates on
individually addressed qubits have achieved extraordinary fidelities,
including $99.999\%$ for nitrogen-vacancy (NV$^-$) centers in
diamond~\cite{bartling2025_nv}, $99.99\%$+ for trapped
ions~\cite{harty2014_trapped_ion}, and sub-$10^{-4}$ single-qubit
gate error in superconducting circuits~\cite{li2023_superconducting_gate},
alongside continuing superconducting-qubit architecture
development~\cite{hyyppa2022_unimon}; yet local microwave delivery in
dense layouts introduces cryogenic heat-load and fan-out constraints
in control stacks~\cite{hornibrook2015_cryocontrol,krinner2019_cryogenic,lecocq2021_photoniclink},
as well as spectral-crowding and crosstalk concerns for dense spin
control~\cite{kunne2024_spinbus} and microwave-field
delivery/characterization in NV platforms~\cite{mariani2020_mwfields}.
A phononic control layer replaces local near-field microwave delivery
by strain coupling at the qubit site.  GHz electronics may still drive
the piezoelectric transducer, but the strain field itself is local,
frequency-selective, and confined by the acoustic wavelength.

Mechanical control is already an experimental resource:
Barfuss~\textit{et~al.}~\cite{barfuss2015_natphys} demonstrated
coherent mechanical driving of a single NV spin,
MacQuarrie~\textit{et~al.}~\cite{macquarrie2013_prl,macquarrie2015_optica} and
Chen~\textit{et~al.}~\cite{chen2020_pra,chen2018_prl,chen2019_nanolett}
engineered acoustic control of spin and orbital transitions,
Lee~\textit{et~al.}~\cite{lee2025_aplquantum} showed multimode
longitudinal strain tuning, and
Cornell~\textit{et~al.}~\cite{cornell2025} achieved mechanical
coherence protection of silicon-vacancy (SiV$^-$) centers with Rabi
frequencies reaching $\SI{800}{\mega\hertz}$.
The missing ingredient for geometric quantum information processing is
complex $\Lambda$ control: a single acoustic mode produces linearly
polarized strain and cannot by itself select one arm of a
three-level $\Lambda$ manifold.

We show that this obstacle is removed whenever the defect's point
group possesses a two-dimensional irreducible representation
$\Gamma_E$: group theory \emph{mandates} a dot-product coupling
$H = g(\varepsilon_1 \mathcal{O}_1 + \varepsilon_2
\mathcal{O}_2)$, so that quadrature driving of two degenerate
modes produces a circularly rotating strain field
mathematically identical to a circularly polarized electromagnetic
drive (Theorem~\ref{thm:synrot}).
This \textit{synthetic rotation} mechanism is a shared-irrep
criterion on the effective device symmetry: it applies to any
strain-active solid-state $\Lambda$ manifold whose projected strain
pair and transition operators share a common two-dimensional irrep.
Defect centers provide the experimentally mature examples, including
$C_{3v}$ defects (NV$^-$, divacancy in SiC) via spin-strain and
$D_{3d}$ defects (SiV$^-$, GeV$^-$, SnV$^-$) via orbit-strain;
representative platforms are summarized in Sec.~\ref{sec:selrule}.

The microwave analogy is deliberately limited.  At the rotating-frame
Hamiltonian level the strain drive reproduces a circular microwave
drive, but its chirality is not imposed by electromagnetic
polarization.  It is selected by a symmetry-protected strain/operator
scalar of the crystal group.  The same representation structure that
fixes the control Hamiltonian also classifies perturbations into
irreducible sectors, linking Hamiltonian design to the effective error
channel rather than merely replacing one actuator by another.
We test that link explicitly: a sector-injection diagnostic shows
that the $\Lambda$ bright manifold routes $A_2$ perturbations toward
the detectable auxiliary level $|0\rangle$, whereas $E$-sector
perturbations remain in the logical bright direction and are removed
to leading order by the echo topology.

The protection statement is local and representation-theoretic:
point-group symmetry protects the allowed strain-defect tensor against
$\Geff$-preserving perturbations, but the mechanical mode degeneracy,
quadrature balance, bath spectral weights, and erasure-detection
efficiency remain device-level assumptions.  This distinction lets us
separate what is symmetry-protected from what is benchmarked
numerically.

The symmetry theorem supplies the essential control resource: a
protected complex $\Lambda$-system drive generated mechanically from a
$\Gamma_E$ strain doublet.  Once this $\Lambda$ manifold is available,
there are two complementary holonomic compilers.  The first is the
superadiabatic transitionless driving (SATD) echo-lune compiler used
below for the Regime-A NV channel benchmark and the quantum error
correction (QEC)-channel extraction.  The second is a
single-shot bright-state compiler in which the two $\Lambda$-leg
amplitudes and scalar detuning remain proportional to one common
envelope,
$u(t)=(\Omega_0,\Omega_1,\Delta)^T=\Omega(t)u_0$.  The SATD compiler
is therefore the channel-engineering route, while the single-shot
compiler is the compact universal-control route.

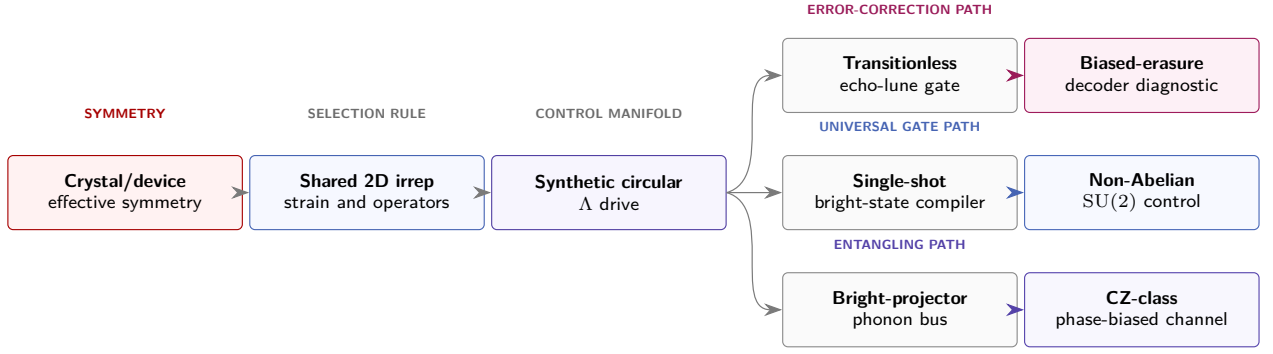
\begin{figure*}[t!]
  \centering
  \begin{tikzpicture}[
    box/.style={
      line width=0.42pt,
      rounded corners=2pt,
      minimum height=0.98cm,
      text width=2.75cm,
      align=center,
      font=\sffamily\scriptsize,
      inner xsep=5pt,
      inner ysep=4pt
    },
    tag/.style={
      font=\sffamily\tiny\bfseries,
      text=black!58,
      align=center
    },
    qectag/.style={
      tag,
      text=CBpink!70!black
    },
    symtag/.style={
      tag,
      text=red!65!black
    },
    unittag/.style={
      tag,
      text=CBblue!70!black
    },
    enttag/.style={
      tag,
      text=CBviolet!70!black
    },
    symbox/.style={
      box,
      draw=red!65!black,
      fill=red!5
    },
    rulebox/.style={
      box,
      draw=CBblue!70!black,
      fill=CBblue!6
    },
    manifoldbox/.style={
      box,
      draw=CBviolet!68!black,
      fill=CBviolet!6
    },
    compilerbox/.style={
      box,
      draw=black!46,
      fill=black!2
    },
    qecbox/.style={
      box,
      draw=CBpink!72!black,
      fill=CBpink!7
    },
    gatebox/.style={
      box,
      draw=CBblue!72!black,
      fill=CBblue!5
    },
    entbox/.style={
      box,
      draw=CBviolet!72!black,
      fill=CBviolet!5
    },
    arr/.style={
      -{Stealth[length=2.4mm, width=1.7mm]},
      line width=0.45pt,
      draw=black!52
    },
    qecarr/.style={
      arr,
      draw=CBpink!70!black
    },
    unitarr/.style={
      arr,
      draw=CBblue!70!black
    },
    entarr/.style={
      arr,
      draw=CBviolet!70!black
    }
  ]
    \node[symtag] at (0,1.05) {SYMMETRY};
    \node[symbox] (sym) at (0,0)
      {\textbf{Crystal/device}\\effective symmetry};

    \node[tag] at (3.20,1.05) {SELECTION RULE};
    \node[rulebox] (irrep) at (3.20,0)
      {\textbf{Shared 2D irrep}\\strain and operators};

    \node[tag] at (6.40,1.05) {CONTROL MANIFOLD};
    \node[manifoldbox] (drive) at (6.40,0)
      {\textbf{Synthetic circular}\\$\Lambda$ drive};

    \node[qectag] at (10.25,2.42) {ERROR-CORRECTION PATH};
    \node[compilerbox] (satd) at (10.25,1.55)
      {\textbf{Transitionless}\\echo-lune gate};
    \node[qecbox] (channel) at (13.45,1.55)
      {\textbf{Biased-erasure}\\decoder diagnostic};

    \node[unittag] at (10.25,0.87) {UNIVERSAL GATE PATH};
    \node[compilerbox] (single) at (10.25,0)
      {\textbf{Single-shot}\\bright-state compiler};
    \node[gatebox] (su2) at (13.45,0)
      {\textbf{Non-Abelian}\\$\mathrm{SU}(2)$ control};

    \node[enttag] at (10.25,-0.68) {ENTANGLING PATH};
    \node[compilerbox] (bus) at (10.25,-1.55)
      {\textbf{Bright-projector}\\phonon bus};
    \node[entbox] (cz) at (13.45,-1.55)
      {\textbf{CZ-class}\\phase-biased channel};

    \draw[arr] (sym) -- (irrep);
    \draw[arr] (irrep) -- (drive);
    \draw[arr] (drive.east) to[out=16,in=180] (satd.west);
    \draw[arr] (drive.east) -- (single.west);
    \draw[arr] (drive.east) to[out=-16,in=180] (bus.west);
    \draw[qecarr] (satd) -- (channel);
    \draw[unitarr] (single) -- (su2);
    \draw[entarr] (bus) -- (cz);
  \end{tikzpicture}
  \caption{\label{fig:principle}%
  Symmetry-to-decoder logic of the proposal.
  A multiplicity-one shared two-dimensional irrep fixes the linear
  strain interaction to a scalar dot product.  Quadrature mechanical
  driving then realizes a protected complex $\Lambda$ manifold.  On
  that manifold, the superadiabatic transitionless driving (SATD)
  echo-lune compiler engineers the
  biased-erasure channel used for the quantum error correction (QEC)
  analysis, while the
  single-shot bright-state compiler supplies compact non-Abelian
  $\mathrm{SU}(2)$ control.  The sector-injection diagnostic in
  Sec.~\ref{sec:erasure} and Appendix~\ref{app:noise_sectors} verifies the intermediate
  sector-to-channel map on the SATD/QEC path.  A Regime-F effective
  bright-projector bus extension, discussed in
  Sec.~\ref{sec:universality} and Appendix~\ref{app:2q_projector_extension},
  tests the corresponding CZ-class entangling-channel diagnostic but
  is not part of the single-qubit compiler stack.}
\end{figure*}

\begin{table*}[t!]
\caption{\label{tab:regimes}Validation hierarchy.
Each row names the model layer and its role in the paper; numerical
performance is reported in the corresponding Results sections.}
\begin{ruledtabular}
\begin{tabular}{llll}
Regime & Platform / model & Control resources & What this establishes \\
\colrule
A & NV rotating-frame channel & $\Lambda$ legs + DQ-SATD &
Biased-erasure channel and sector map \\
B & NV GHz-HBAR projection & Same protocol at resonant $\Omm$ &
Resonant route and envelope tracking \\
C & SiV orbital $\Lambda$ benchmark & $D_{3d}$ orbit-strain $\Lambda$ legs &
Synthetic-rotation transfer \\
D & Generic/NV single-shot compiler & $\Lambda$ legs + scalar detuning &
Compact non-Abelian $\mathrm{SU}(2)$ control \\
E & SiV single-shot compiler & SiV $\Lambda$ legs + scalar detuning &
Fast orbital universal-control route \\
F & Bright-projector phonon bus & Bright projectors + phonon loop &
CZ-class phase-biased diagnostic \\
\end{tabular}
\end{ruledtabular}
\end{table*}

A matched-environment 3C-SiC $C_{3v}$ comparison is reported in
Sec.~\ref{sec:c3v_platform_rule} and Appendix~\ref{app:extended_sweeps_sm},
Sec.~7, as a
platform-design check, not as a separate measured-device benchmark.

One NV membrane mode topology that can instantiate the rotating-frame
channel benchmark is summarized in Fig.~\ref{fig:architecture}.
The erasure and decoder analysis below use the extracted
$\Lambda$ + DQ-SATD channel, not the membrane strain model by itself.
The resulting symmetry-to-decoder logic and validation hierarchy are
summarized in Fig.~\ref{fig:principle} and Table~\ref{tab:regimes}.

For the SATD compiler, we trace the fidelity ceiling of the adiabatic
Orange Slice holonomic protocol to a geometric singularity at the
Bloch-sphere South Pole and eliminate it with a composite
nonadiabatic geometric quantum computation (NGQC) two-loop topology
combined with superadiabatic transitionless driving
(SATD, $\lambda \equiv 1.0$)~\cite{zhu2002_prl,berry2009_transitionless}.
This echo-lune protocol removes the adiabatic speed limit and shapes
the residual channel into the biased-erasure form used for the QEC
analysis.
For the single-shot compiler, the same complex $\Lambda$ legs are held
in a fixed bright-state direction while a common envelope closes one
cyclic excursion; this gives a simpler universal gate layer, but its
residual channel is not merged with the SATD QEC channel.

We validate the framework through a controlled hierarchy of models.
Regime~A establishes the symmetry-generated rotating-frame NV channel
and decoder-facing biased-erasure noise model; Regime~B separately
maps the resonant GHz-HBAR implementation route, where Rabi rate,
linewidth, and envelope tracking become resonator-design parameters.
Regime~A is a rotating-frame NV channel benchmark.  The membrane mode
topology of Fig.~\ref{fig:architecture} supplies an illustrative
quadrature strain source, while the benchmark itself tests the
complete $\Lambda$ + DQ-SATD Hamiltonian at
$\Omm=\SI{2.22}{\mega\hertz}$ and supplies the biased-erasure channel
used for the decoder analysis.
This benchmark reaches conditional $\Favg=99.88\%$ at
$\tgate=\SI{1.833}{\micro\second}$, with
$\Favg^{\mathrm{eff}}=99.40\%$ after leakage is counted as error.
Regime~B evaluates the corresponding resonant GHz-HBAR implementation
route at the experimentally projected $\Omm=2.83$, $28.3$, and
$\SI{141.5}{\kilo\hertz}$ scale.
Regime~C is a conservative SiV $\Lambda$-leg benchmark that tests
$D_{3d}$ synthetic-rotation control without an additional
lower-doublet SATD channel, giving $\Favg=96.32\%$ at
$\tgate=\SI{46.2}{\nano\second}$.
Regime~D is the single-shot universal suite on the same
symmetry-generated $\Lambda$ manifold, using a
$\tgate=\SI{0.9009}{\micro\second}$ bright-state pulse with the
two $\Lambda$ legs plus scalar detuning and reaching
$F_{\mathrm{eff}}\simeq99.86\%$ across representative
noncommuting single-qubit gates.
Regime~E applies the same single-shot bright-state compiler directly
to the SiV orbital $\Lambda$ manifold using only the two orbit-strain
$\Lambda$ legs and a synchronized scalar detuning.
No lower-doublet SATD matrix element is assumed.
At the primary millikelvin point, $\Omega_{\mathrm{peak}}
= \SI{300}{\mega\hertz}$ gives
$\tgate=\SI{6.667}{\nano\second}$ and
$F_{\mathrm{eff}}=99.52\%$--$99.67\%$ across the representative
four-gate suite.
The dominant residual error is orbital-$T_1$ depolarization during
auxiliary occupation, not erasure-convertible leakage.
Regime~F is an effective $\Lambda$-level architecture diagnostic:
the same logical bright projectors couple to a detuned phonon loop to
test whether a CZ-class two-qubit channel can remain phase-biased and
heralded-leakage compatible.  Regime~F is evaluated at the effective
projector-force level and supplies the entangling-channel diagnostic.

The same $\Gamma_E$-enabled $\Lambda$ structure also organizes the
leading noise operators into $A_1$, $A_2$, and $E$ sectors.
For weak noise with comparable irrep-sector bath weights, $A_1$
perturbations are common-mode, $A_2$ perturbations preferentially
leak toward the optically distinguishable auxiliary state, and
$E$-sector bit-flip errors are echo-suppressed.
In the Regime-A rotating-frame NV channel this gives a strongly
biased extracted model with $p_{\mathrm{era}}=0.47\%$ and
$p_Z=0.168\%$, while the transverse component lies at the numerical
floor of the channel extraction.  In a code-capacity model of this
extracted Regime-A channel, XZZX and Calderbank-Shor-Steane (CSS)
codes lie on opposite sides
of the finite-size scaling boundary, yielding a substantial nominal
overhead reduction.  We treat this number as a logical-channel
diagnostic rather than a complete architecture threshold, and promote
finite $p_{XY}$ floors and reduced erasure-detection efficiency to
explicit decoder stress axes (Appendix~\ref{app:qec_montecarlo}).

This hierarchy also defines the scope of the simulations.
Regime~A is the rotating-frame NV $\Lambda$ + DQ-SATD channel
extraction, with the membrane mode topology serving as one
strain-source realization; Regime~B separately evaluates the resonant
GHz-HBAR implementation route and its transfer-function constraints.
Regime~C tests the $D_{3d}$ $\Lambda$-control mechanism using only
the $\Lambda$ legs established by the orbit-strain selection rule.
Regime~D establishes compact universal non-Abelian control on the
same symmetry-generated $\Lambda$ manifold.
Regime~E establishes the corresponding SiV-compatible single-shot
universal-control route, again using only $\Lambda$ legs plus scalar
detuning.
Regime~F closes the architecture only at the effective
projector-force level, as an entangling-channel diagnostic.
The primary QEC simulations isolate the code-capacity consequence of
the Regime~A channel and then sweep finite transverse-error floors as
a decoder validation envelope.  We additionally run a lightweight
scheduled two-sector XZZX detector-model stress diagnostic with
repeated syndrome rounds, erasure-aware weights, measurement/reset
faults, explicit transverse $X/Y$ faults, leakage persistence,
delayed flags, and local crosstalk.  This diagnostic preserves the
nominal $d=9$ biased-erasure proxy under the extracted channel
assumptions and identifies strong local crosstalk as the dominant
architecture-specific validation target.  A hardware-calibrated
detector-error model remains the next step toward a device-level
threshold.

The paper is organised as follows.
Section~\ref{sec:synrot_theory} derives the synthetic rotation
selection rule and the platform-independent $\Lambda$ Hamiltonian.
Section~\ref{sec:gate} presents the holonomic compilers:
the adiabatic baseline and its South Pole limitation, the composite
NGQC + SATD channel-engineering route, the single-shot bright-state
universal-control route, and the gate-time compression mechanism.
Sections~\ref{sec:nv} and~\ref{sec:siv} provide full
platform-specific models for NV$^-$ and SiV$^-$, respectively.
Section~\ref{sec:results} reports the simulation results, error
budget, biased-erasure noise analysis, decoder estimates, and scheduled
stress diagnostic.
Section~\ref{sec:discussion} discusses the cross-platform scaling
law, universal and entangling extensions, and experimental outlook.

\section{Symmetry-to-\texorpdfstring{$\Lambda$}{Lambda} theory}
\label{sec:synrot_theory}

\subsection{Selection rule for strain-mediated $\Lambda$-system control}
\label{sec:selrule}

We begin with the central theoretical result: a shared-irrep selection
rule that identifies when an effective solid-state manifold supports
complex $\Lambda$-leg control via mechanical strain.

\begin{theorem}[Shared-irrep synthetic circular-strain interface]
\label{thm:synrot}
Let $\mathcal H_\Lambda$ be a projected low-energy solid-state
$\Lambda$ manifold with effective device symmetry
$\Geff$ and a real two-dimensional irreducible representation
$\Gamma_E$.  For an ideal point defect, $\Geff$ reduces to the defect
point group; in a device it is the little group left by crystal
symmetry, confinement, static fields, and the acoustic mode basis.
For $D_{3d}$ centers, replace $E,A_1,A_2$ below by
$E_g,A_{1g},A_{2g}$.  Suppose:
\begin{enumerate}
  \item The projected symmetric strain tensor contains a pair
  $\boldsymbol{\epsilon}=(\epsilon_1,\epsilon_2)$ transforming as
  $\Gamma_E$.
  \item The transition-operator space of the encoded manifold contains
  a pair
  $\mathbf{O}=(O_1,O_2)$, spin or orbital, transforming under the
  same $\Gamma_E$.
  \item The mechanical resonator supplies two degenerate phase-locked
  modes whose local strain projections are $\epsilon_1$ and
  $\epsilon_2$ at the defect site.
  \item The selected shared-irrep channel is multiplicity-one in the
  projected linear strain coupling.
\end{enumerate}
Then:
\begin{enumerate}
  \item The only $\Geff$-invariant linear strain-defect coupling in
  this channel is
  \[
    H_{\mathrm{int}}=g(\epsilon_1O_1+\epsilon_2O_2).
  \]

  \item The non-scalar bilinears transform as
  \[
    B_{A_2}=\epsilon_1O_2-\epsilon_2O_1,
  \]
  and
  \[
    \mathbf{B}_E=
    \left(
    \epsilon_1O_1-\epsilon_2O_2,\;
    \epsilon_1O_2+\epsilon_2O_1
    \right).
  \]
  Thus they cannot appear in the Hamiltonian unless accompanied by
  symmetry-breaking spurions transforming as $A_2$ or $E$.

  \item A $\Geff$-preserving perturbation can only renormalize the scalar
  coefficient $g\rightarrow g+\delta g$ or generate $A_1$ detuning
  shifts.  It cannot generate ellipticity, handedness mixing, or
  anisotropic strain coupling at first order.  In this precise sense,
  the circular strain selection rule is symmetry-protected.

  \item For quadrature driving,
  \[
    \epsilon_1(t)=\epsilon_0\cos\omega_d t,\qquad
    \epsilon_2(t)=\epsilon_0\sin\omega_d t,
  \]
  and circular operators $O_\pm=O_1\pm iO_2$, the interaction becomes
  \[
    H_{\mathrm{int}}(t)=\frac{g\epsilon_0}{2}
    \left(
    O_+e^{-i\omega_d t}+O_-e^{+i\omega_d t}
    \right).
  \]
  Under the rotating-wave approximation, tuning $\omega_d$ to one arm
  of the $\Lambda$ system retains a single circular component, yielding
  the single-arm rotating Hamiltonian used below.
\end{enumerate}
\end{theorem}

\begin{proof}[Proof sketch]
Write the most general linear coupling in the selected doublet channel
as $H_\epsilon=\boldsymbol{\epsilon}^{T}C\mathbf{O}$.  Invariance
under $g\in\Geff$ requires
$R_E(g)^T C R_E(g)=C$.  Because the shared irrep is irreducible and
appears once in the selected channel, Schur's lemma gives
$C=g_0I_2$, yielding the dot product.  For the defect point groups used
below, the familiar decompositions
$\Gamma_E\otimes\Gamma_E=A_1\oplus A_2\oplus E$ for $C_{3v}$ and
$E_g\otimes E_g=A_{1g}\oplus A_{2g}\oplus E_g$ for $D_{3d}$ give the
listed non-scalar bilinears explicitly.  They can enter only when
multiplied by symmetry-breaking spurions in the corresponding irreps.
Substitution of the quadrature drive gives the circular decomposition
above, and the RWA selects one component on resonance.  A complete
projector derivation and spurion classification are given in
Appendix~\ref{app:noise_sectors}.
The representation-theory notation and selection-rule logic follow
standard condensed-matter group-theory conventions~\cite{dresselhaus2008_group}.
\end{proof}

\begin{remark}[Effective symmetry beyond point defects]
Theorem~\ref{thm:synrot} is not a statement about defect chemistry.
For a non-defect solid-state subsystem, such as an acceptor,
valley-orbital manifold, hole-spin manifold, or gate-defined quantum
dot, the relevant group is the effective little group
\[
  \Geff
  =
  G_{\mathrm{crystal}}
  \cap G_{\mathrm{confinement}}
  \cap G_{\mathrm{device}}
  \cap G_{\mathrm{field}} .
\]
The criterion is the same: the projected strain tensor and the
transition-operator space of the encoded $\Lambda$ manifold must share
a multiplicity-one two-dimensional irrep.  Thus the mechanism
generalizes to strain-active crystalline $\Lambda$ manifolds
satisfying this shared-irrep criterion.  The present work focuses on
defect centers because they provide natural high-symmetry realizations
and experimentally established strain couplings.
\end{remark}

The theorem protects the local strain-defect tensor, not the mechanical
device by itself.  The resonator must either supply a
symmetry-degenerate doublet or be tuned and phase-locked so that
residual splitting and quadrature imbalance are perturbative on the
gate timescale.  Point groups with only one-dimensional irreps, such
as $C_{2v}$, $C_{2h}$, or $C_2$, cannot support this synthetic circular
mechanism because no strain pair transforms as a single
two-dimensional irrep.

The NV-center platform is reviewed in
Refs.~\cite{childress2013_review,doherty2013_nvreview}, while
group-IV split-vacancy centers in diamond are reviewed in
Ref.~\cite{bradac2019_groupiv}; the microscopic group-IV electronic
structure is treated in Ref.~\cite{thiering2018_groupiv}.
Table~\ref{tab:platforms} summarizes the verified platforms.

\begin{table}[b]
\caption{\label{tab:platforms}%
Solid-state platforms satisfying the synthetic rotation selection
rule (Theorem~\protect\ref{thm:synrot}).  All entries have a
verified $A_1$($A_{1g}$)-type dot-product coupling.}
\begin{ruledtabular}
\begin{tabular}{lcccc}
\textrm{Platform} & \textrm{Group} &
  \textrm{$\Lambda$ type} &
  \textrm{Coupling} &
  \textrm{Splitting} \\
\colrule
NV$^-$ (diamond) & $C_{3v}$ & Spin & 2.8~GHz/str. & 2.87~GHz \\
VV$^0$ (3C-SiC) & $C_{3v}$ & Spin & 1.8~GHz/str. & 1.33~GHz \\
SiV$^-$ (diamond) & $D_{3d}$ & Orbital & 1.3~PHz/str. & 48~GHz \\
GeV$^-$ (diamond) & $D_{3d}$ & Orbital & $\sim$1~PHz/str. & 170~GHz \\
SnV$^-$ (diamond) & $D_{3d}$ & Orbital & $\sim$1~PHz/str. & 850~GHz \\
\end{tabular}
\end{ruledtabular}
\end{table}

Among $C_{3v}$ platforms, the neutral divacancy in 3C-SiC shares the
NV spin-strain tensor structure
identically~\cite{falk2014_sic,udvarhelyi2018_sic_pra}, differing
only in coupling magnitudes and zero-field splitting;
Density-functional theory (DFT) calculations yield
$|h_{16}/h_{26}| \approx 0.75$ for 3C-SiC
(compared with $6.95$ for NV$^-$), so the parasitic AC~Stark
channel that arises in NV (Sec.~\ref{sec:acstark}) is intrinsically
weaker in this host.
This ratio reappears below as a $C_{3v}$ platform-selection
parameter: at fixed $\Omm$, the parasitic DQ Stark scale is
proportional to $|h_{16}/h_{26}|^2/D$.
Among $D_{3d}$ platforms, SiV$^-$
(spin-orbit splitting $\Delta_{\mathrm{SO}} = 48$~GHz) is within reach of demonstrated
acoustic resonator frequencies~\cite{cornell2025}, while
GeV$^-$ (170~GHz~\cite{wan2018_gev}) and
SnV$^-$ (850~GHz~\cite{trusheim2020_snv}) require nanophotonic
or optomechanical transduction.

\subsection{Synthetic rotation via degenerate mechanical modes}
\label{sec:synrot_mechanism}

The selection rule guarantees the \emph{form} of the coupling; here
we show how a concrete resonator geometry realizes it.
Consider a mechanical resonator supporting two degenerate modes
whose strain fields project onto orthogonal components
$\varepsilon_1$ and $\varepsilon_2$ at the defect site.
Driving these modes in quadrature with equal amplitude produces a
circularly rotating strain field:
\begin{equation}
  \varepsilon_1(t) = \varepsilon_0 \cos(\omD t),
  \quad
  \varepsilon_2(t) = \varepsilon_0 \sin(\omD t),
  \label{eq:rotating_strain}
\end{equation}
where $\varepsilon_0$ is the peak strain amplitude and
$\omega_d = 2\pi f_0$ the mechanical drive angular frequency,
tuned to the $\Lambda$-system splitting.
Substituting into the dot-product coupling of Theorem~\ref{thm:synrot}
and applying the rotating-wave approximation yields an effective
$\sigma_+$-type Hamiltonian
\begin{equation}
  \frac{H_{\mathrm{RWA}}}{h} = \frac{\Omm}{2}
  \left(e^{-i\phi}|0\rangle\langle{+1}|
  + e^{+i\phi}|{+1}\rangle\langle 0|\right),
  \label{eq:HRWA}
\end{equation}
mathematically identical to a circularly polarized microwave drive,
where $\Omm = g\,\varepsilon_0$ is the mechanical Rabi frequency.
The physical origin of $\Omm$ differs by platform
(spin-strain coupling $h_{26}$ for $C_{3v}$; orbit-strain
coupling $f_\perp$ for $D_{3d}$), but the rotating-frame
Hamiltonian is universal.

\subsection{Effective $\Lambda$-system Hamiltonian}
\label{sec:lambda_ham}

The holonomic gate protocols of Sec.~\ref{sec:gate} require two
independently controlled circular drives of opposite chirality:
$\sigma_+$ near $D_+$ and $\sigma_-$ near $D_-$
(where $D_\pm \equiv D \pm \gamma_e B_z$ are the
Zeeman-split transition frequencies), to
simultaneously address both legs of the $\Lambda$ manifold.
The identical spatial mode topology provides the required
orthogonal coupling for each chirality with independent amplitude
and phase control.
The complete control Hamiltonian decomposes as
\begin{multline}
  \frac{H_{\Lambda}(t)}{h} = \frac{\Omm}{2}\!
  \Biggl[\sin\!\frac{\theta}{2}\,
    \bigl(|0\rangle\langle{-1}| + \text{h.c.}\bigr) \\
  - \cos\!\frac{\theta}{2}\,
    \bigl(e^{-i\phi}|0\rangle\langle{+1}| + \text{h.c.}\bigr)
  \Biggr],
  \label{eq:Hlambda}
\end{multline}
where $\theta(t)$ and $\phi(t)$ are the time-dependent mixing angle
and azimuthal phase. The zero eigenstate of $H_\Lambda$ is
the dark state
\[
  |D\rangle =
  \cos\frac{\theta}{2}|{-1}\rangle
  +
  e^{i\phi}\sin\frac{\theta}{2}|{+1}\rangle ,
\]
while the coupled bright state is
\[
  |B\rangle =
  \sin\frac{\theta}{2}|{-1}\rangle
  -
  e^{i\phi}\cos\frac{\theta}{2}|{+1}\rangle .
\]
With this convention the Hamiltonian reduces to
\[
  H_\Lambda/h = \frac{\Omm}{2}
  \left(|0\rangle\langle B|+|B\rangle\langle 0|\right),
\]
so that $|D\rangle$ is exactly dark and $|B\rangle$ couples to
$|0\rangle$ with strength $\Omm/2$. Cyclic traversal of
$|D\rangle$ on the Bloch sphere generates the geometric holonomy.

The subsequent gate protocols (Sec.~\ref{sec:gate}) operate entirely
on this Hamiltonian, parameterised solely by $\Omm$, $\theta(t)$,
and $\phi(t)$.
The physical realization of $\Omm$ differs by platform
(Secs.~\ref{sec:nv} and~\ref{sec:siv}), but the gate construction
is platform-independent.

\section{Holonomic compilers on the symmetry-generated \texorpdfstring{$\Lambda$}{Lambda} manifold}
\label{sec:gate}

The symmetry-to-$\Lambda$ reduction supplies the common control
manifold.  The remaining choice is the compiler: how the available
complex $\Lambda$ legs, and when needed a scalar detuning or auxiliary
counterdiabatic actuator, are turned into a logical gate.  We use two
compilers with different purposes.  The SATD echo-lune compiler is the
error-channel engineering route used for the Regime-A NV channel benchmark and
the QEC extraction.  The single-shot bright-state compiler is the
compact universal-control route on the same manifold.
This holonomic-control language builds on non-Abelian geometric
phases and holonomic quantum computation~\cite{wilczek1984_gauge,zanardi1999_holonomic},
with the nonadiabatic $\Lambda$-system formulation closely related to
Ref.~\cite{sjoqvist2012_nhqci}.

\subsection{Adiabatic baseline and the South Pole singularity}
\label{sec:orangeslice}

The adiabatic Orange Slice protocol is the useful foil: it sweeps the
dark state $|D(\theta,\phi)\rangle$ of Eq.~\eqref{eq:Hlambda} around a
two-leg loop, but the path crosses the South-Pole coordinate
singularity where the azimuthal phase jumps.  In our benchmark this
leaves a $99.07\%$ noiseless ceiling even after first-order
derivative-removal-by-adiabatic-gate (DRAG)-type correction, showing
that the bottleneck is geometric rather
than a simple drive-amplitude limit.

The relevant counter-diabatic operator basis is
\begin{equation}
  H_{\mathrm{CD}} = -\lambda\,\frac{\dot{\theta}}{2}
  \left[\sin\phi\; \mathrm{Op}_{\mathrm{re}}
  + \cos\phi\; \mathrm{Op}_{\mathrm{im}}\right],
  \label{eq:HCD}
\end{equation}
Here $\lambda$ is the counter-diabatic correction strength, and
$\mathrm{Op}_{\mathrm{re}} = |{-1}\rangle\langle{+1}| + |{+1}\rangle\langle{-1}|$
and
$\mathrm{Op}_{\mathrm{im}} = i(|{+1}\rangle\langle{-1}| - |{-1}\rangle\langle{+1}|)$
are lower-manifold Pauli-like operators.
This term suppresses first-order non-adiabatic leakage from the dark
state.
At $\lambda=1.0$, it improves the bare Orange-Slice fidelity from
$98.67\%$ to $99.07\%$, but it does not remove the longitudinal phase
artifact generated at the pole~\cite{barnes2021_geometric}.  This
motivates changing the trajectory topology before applying the exact
SATD correction.

\subsection{Composite NGQC + SATD protocol}
\label{sec:ngqc}

The Zhu--Wang two-loop construction~\cite{zhu2002_prl} provides the
required topological change.
The trajectory consists of two symmetric sub-loops
(Fig.~\ref{fig:composite}),
each tracing a thin orange-slice lune on the parameter sphere,
where $T \equiv \tgate$ denotes the total gate time.

\noindent\textbf{Loop~1} ($0 \leq t \leq T/2$, starting azimuth $\phi_1 = 0$):
\begin{equation}
  \theta(t) = \theta_{\max}\sin\!\left(\frac{\pi t}{T/2}\right),
  \label{eq:loop1}
\end{equation}
evolving from the North Pole through the South Pole
($\theta_{\max} = \pi$) and back.
The azimuth shifts by $\delta\phi = \pi/4$ via a narrow
$\tanh$ step centered at the South Pole
($\theta = \pi$, $\sin\theta = 0$), closing the lune with
finite enclosed area.

\noindent\textbf{Loop~2} ($T/2 \leq t \leq T$, starting azimuth $\phi_2 = \pi$):
An identical lune displaced to the complementary azimuth, with
$\phi$ sweeping from $\pi$ to $\pi + \delta\phi$ at the South Pole.

Each lune encloses a solid angle
$\Omega_0 = \delta\phi\,(1 - \cos\theta_{\max}) = \pi/2$.
The dark-state Berry phase per lune is
$\gamma_D = -\Omega_0/2 = -\pi/4$, while the bright state
acquires $\gamma_B = +\Omega_0/2 = +\pi/4$
(their Berry curvatures carry opposite signs).
Over the full two-loop path, the dark state accumulates
$\gamma_D^{\mathrm{tot}} = -\pi/2$ and the bright state
$\gamma_B^{\mathrm{tot}} = +\pi/2$; the relative geometric phase
$\gamma_{\mathrm{geo}} = \gamma_D^{\mathrm{tot}} - \gamma_B^{\mathrm{tot}} = -\pi$
implements the target $Z(\pi)$ rotation.

This topology provides three simultaneous advantages:
\begin{enumerate}
  \item \textbf{Singularity neutralization.}
  Although each sub-loop traverses the South Pole
  ($\theta = \pi$), the azimuthal transition is confined to
  this point where $\sin\theta = 0$.  Consequently the
  $\dot{\phi}$-induced non-adiabatic coupling
  $\langle B|\partial_t|D\rangle \propto \dot{\phi}\sin\theta$
  vanishes identically at the pole, and the $\dot{\theta}$-only
  SATD counter-diabatic Hamiltonian remains exact.

  \item \textbf{Phase echo.}
  The $\pi$ separation between loop starting azimuths causes
  deterministic dynamical phases to cancel by destructive interference:
  $+\gamma_d$ in loop~1 and $-\gamma_d$ in loop~2.  The same echo
  symmetry cancels the first- and second-order Berry-connection
  corrections generated by quasi-static $E$-sector perturbations.  The
  leading surviving geometric correction from this sector is therefore
  $O(\sigma_E^3/\Omm^3)$, while the $A_2$ sector is suppressed by the
  parity filter discussed in Sec.~\ref{sec:erasure} and proved in
  Appendix~\ref{app:noise_sectors}.

  \item \textbf{Gate-time compression.}
  When combined with counter-diabatic driving, the adiabatic speed
  limit is removed, enabling aggressive temporal compression.
\end{enumerate}

\begin{figure}[t]
  \centering
  \includegraphics[width=\columnwidth]{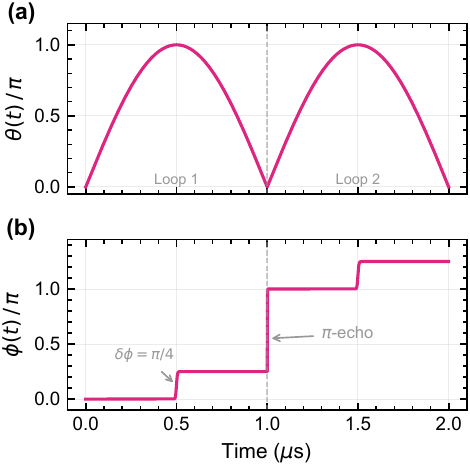}
  \caption{\label{fig:composite}%
  Composite NGQC trajectory.
  (a)~Mixing angle $\theta(t)$ executes two symmetric sub-loops,
  each traversing the full $\theta_{\max} = \pi$
  (North Pole $\to$ South Pole $\to$ North Pole).
  (b)~Azimuthal phase $\phi(t)$: within each sub-loop $\phi$ shifts
  by $\delta\phi = \pi/4$ at the South Pole, and the starting
  azimuth advances by $\pi$ between loops.
  All $\phi$ transitions occur where $\sin\theta = 0$,
  neutralising the coordinate singularity.
  }
\end{figure}

\subsection{Superadiabatic transitionless driving (SATD)}
\label{sec:satd}

Executing two sub-loops within $\tgate \leq \SI{2}{\micro\second}$
demands extreme angular velocities, far exceeding the adiabatic
limit.
Without counter-diabatic correction, the bare composite protocol
yields $F = 22.92\%$, confirming that the compressed path is not
usable without the counterdiabatic generator.

The SATD framework~\cite{berry2009_transitionless,baksic2016_dressed,zhou2017_superadiabatic}
constructs an exact
counter-diabatic Hamiltonian (Appendix~\ref{app:drag}):
\begin{equation}
  H_{\mathrm{CD}} = -\frac{\dot{\theta}}{2}
  \left[\sin\phi\;\mathrm{Op}_{\mathrm{re}}
  + \cos\phi\;\mathrm{Op}_{\mathrm{im}}\right],
  \label{eq:SATD}
\end{equation}
with operators $\mathrm{Op}_{\mathrm{re}}$ and
$\mathrm{Op}_{\mathrm{im}}$ as defined in Eq.~\eqref{eq:HCD}.
Equation~\eqref{eq:SATD} has the same operator structure as the
DRAG Hamiltonian (Eq.~\eqref{eq:HCD}), but with the critical
distinction that the correction strength $\lambda \equiv 1.0$ is
now derived \textit{analytically} from Berry's
theorem~\cite{berry2009_transitionless} rather than determined
empirically, requiring no tuning.
This exactly cancels all non-adiabatic transitions, enforcing strict
parallel transport of the instantaneous eigenstates and completely
removing the adiabatic speed limit.
For NV$^-$ the operator in Eq.~\eqref{eq:SATD} is physically
synthesized by a resonant double-quantum strain tone at
$2\gamma_e B_z$, with amplitude and phase proportional to
$|\dot{\theta}|$ and $\phi+\pi/2$, respectively
(Appendix~\ref{app:satd_hardware}).
In the DQ rotating frame this tone supplies
$H_{\mathrm{DQ}}^{\mathrm{res}}/\hbar =
[\Omega_{\mathrm{re}}\mathrm{Op}_{\mathrm{re}}+
\Omega_{\mathrm{im}}\mathrm{Op}_{\mathrm{im}}]/2$, and the SATD
Hamiltonian is obtained by choosing
$\Omega_{\mathrm{re}}=-\dot{\theta}\sin\phi$ and
$\Omega_{\mathrm{im}}=-\dot{\theta}\cos\phi$.

With SATD active, the noiseless fidelity of the composite protocol
reaches $F = 99.76\%$, breaking the $99.07\%$ Orange Slice ceiling.
The residual $0.24\%$ error is dominated by higher-order coherent
imperfections in the composite loop geometry.

\subsection{Single-shot bright-state compiler and non-Abelian SU(2) gates}
\label{sec:single_shot_compiler}

The same symmetry-generated $\Lambda$ manifold also supports a
nonadiabatic bright-state compiler that does not use adiabatic
following, the SATD counterdiabatic actuator, or direct logical-state
coupling.  In a rotating frame and using cycle-frequency units
$K=H/h$, write
\begin{align}
K(t)=&\;
\Delta(t)\ket{a}\bra{a}
\nonumber\\
&+\frac{1}{2}
\left[
\Omega_0(t)\ket{a}\bra{0_L}
+\Omega_1(t)\ket{a}\bra{1_L}
+{\rm h.c.}
\right].
\label{eq:ss_lambda_hamiltonian}
\end{align}
The single-shot condition is proportional control:
\begin{align}
\Omega_0(t)&=\Omega(t)\cos\alpha\cos\frac{\vartheta}{2},\\
\Omega_1(t)&=\Omega(t)\cos\alpha e^{-i\phi}
              \sin\frac{\vartheta}{2},\\
\Delta(t)&=\Omega(t)\sin\alpha .
\label{eq:ss_controls}
\end{align}
In the Regime-D validation, $\Delta(t)$ is treated as a synchronized
rotating-frame detuning command; Stark or DC detuning can be
substituted only after calibration as the third waveform channel.
Equivalently,
\begin{equation}
u(t)=(\Omega_0,\Omega_1,\Delta)^T=\Omega(t)u_0 .
\end{equation}
Defining
\begin{align*}
\ket{b_{\mathbf n}}&=
\cos\frac{\vartheta}{2}\ket{0_L}
\!+e^{i\phi}\sin\frac{\vartheta}{2}\ket{1_L},
\\
\ket{d_{\mathbf n}}&=
\sin\frac{\vartheta}{2}\ket{0_L}
\!-e^{i\phi}\cos\frac{\vartheta}{2}\ket{1_L},
\end{align*}
with
$\mathbf n=(\sin\vartheta\cos\phi,\sin\vartheta\sin\phi,
\cos\vartheta)$, the Hamiltonian factorizes as
\begin{equation}
K(t)=\Omega(t)M_{\alpha,\mathbf n},
\label{eq:ss_factorization}
\end{equation}
where
\begin{equation}
M_{\alpha,\mathbf n}
=\sin\alpha\ket{a}\bra{a}
+\frac{\cos\alpha}{2}
\left(\ket{a}\bra{b_{\mathbf n}}+
\ket{b_{\mathbf n}}\bra{a}\right).
\end{equation}
Thus $[K(t),K(t')]=0$ and the pulse shape enters only through its
area.  The cyclic area condition is
\begin{equation}
\int_0^T\Omega(t)\,dt=1 .
\label{eq:ss_area}
\end{equation}
It returns the logical subspace to itself and induces
\begin{equation}
U_L(\mathbf n,\gamma)=
\ket{d_{\mathbf n}}\bra{d_{\mathbf n}}
+e^{-i\gamma}\ket{b_{\mathbf n}}\bra{b_{\mathbf n}},
\qquad
\gamma=\pi(1+\sin\alpha).
\label{eq:ss_gate}
\end{equation}
Up to a global phase, this is
$\exp[-i\gamma\,\mathbf n\cdot\boldsymbol\sigma/2]$.
The family is non-Abelian because changing $\mathbf n$ changes the
bright-state projector:
\begin{equation}
[U(\mathbf n_1,\gamma_1),U(\mathbf n_2,\gamma_2)]
\propto
(\mathbf n_1\times\mathbf n_2)\cdot\boldsymbol\sigma .
\label{eq:ss_commutator}
\end{equation}
The operational diagnostic used in Regime~D is the pair
$X_{\pi/2}Z_{\pi/2}$ and $Z_{\pi/2}X_{\pi/2}$: each sequence matches
its own target, while the two composed unitaries have average fidelity
$0.5$ relative to one another.  The full proof, transfer-matrix
condition $G(\omega)u_0\simeq g(\omega)u_0$, detuning sensitivity, and
protocol-stack comparison are collected in Appendix~\ref{app:single_shot}.

\begin{table}[tb]
\footnotesize
\caption{\label{tab:compiler_comparison}
Two compilers on the same symmetry-generated $\Lambda$ manifold.  The
SATD compiler is the channel-engineering route used for the Regime-A
QEC analysis; the single-shot compiler is the compact universal-control
route used for Regime~D.}
\setlength{\tabcolsep}{3pt}
\begin{tabular}{@{}ll@{}}
\toprule
Compiler & Controls, role, and output \\
\midrule
SATD echo-lune &
\begin{tabular}[t]{@{}l@{}}
$\Lambda$ legs + resonant DQ SATD \\
Error-channel shaping \\
Regime-A biased-erasure channel
\end{tabular} \\
\addlinespace
Single-shot bright-state &
\begin{tabular}[t]{@{}l@{}}
$\Lambda$ legs + scalar detuning \\
Universal SU(2) gates \\
Regime-D non-Abelian gate suite
\end{tabular} \\
\bottomrule
\end{tabular}
\end{table}

\subsection{Gate-time compression}
\label{sec:compression}

Because SATD eliminates coherent leakage to a negligible
$\sim\!0.003\%$ residual, the error budget is overwhelmingly
dominated by dissipative processes that scale linearly with gate time.
Compressing $\tgate$ from $\SI{5.0}{\micro\second}$ to the
${\sim}\SI{1}{\micro\second}$ scale truncates the system's temporal
exposure to all destructive channels ($T_1$ lattice relaxation,
$T_{1\rho}$ rotating-frame depolarization, and quasi-static surface
noise), providing the primary engineering lever for
maximizing gate fidelity.

The NGQC construction eliminates dark-state dynamical phase
by design; the bright-state dynamical phase
$\gamma_B \approx \Omm \tgate$ persists and produces
oscillations in the process-level fidelity $\Favg$.
The fidelity peaks sharply at ``magic'' gate times where
$\gamma_B = 2\pi n$ ($n$ integer), and drops as low as $90\%$
between them.
The optimal operating point must be selected from these discrete
magic-time peaks (Sec.~\ref{sec:sweep7}).

The gate-time compression mechanism predicts that performance is
governed by the ratio $T_1/\tgate$ rather than by absolute $T_1$:
any platform with sufficient strain coupling to compress $\tgate$
well below $T_1$ should achieve near-unity fidelity.
This scaling is illustrated by the NV channel benchmark and the conservative
SiV benchmark within their respective modeled control assumptions.

\section{NV$^-$ spin-qubit implementation and channel benchmark}
\label{sec:nv}

\subsection{Ground-state Hamiltonian and computational encoding}
\label{sec:hamiltonian}

The ground-state spin-1 triplet (${}^3A_2$) of the NV center in
diamond is described by the static Hamiltonian
\begin{equation}
  H_0 / h = D\, \Sz^2 + \gamma_e B_z \Sz \,,
  \label{eq:H0}
\end{equation}
where all frequencies are ordinary (cycle) frequencies consistent
with the $H/h$ convention;
$D \approx \SI{2.87}{\giga\hertz}$ is the zero-field splitting,
$\Sz$ is the $z$-component of the spin-1 angular momentum operator,
$\gamma_e \approx \SI{2.80}{\mega\hertz\per\gauss}$ is the electron
gyromagnetic ratio, and $B_z = \SI{50}{\gauss}$ lifts the
$|\pm 1\rangle$ degeneracy via the electronic Zeeman effect.
The architecture encodes the computational qubit within the
$\{|-1\rangle,\,|+1\rangle\}$ doublet, utilizing the $|0\rangle$
state as the necessary auxiliary level to facilitate
$\Lambda$-system holonomic traversal.

\subsection{$C_{3v}$ spin-strain coupling}
\label{sec:spinstrain}

Six independent coupling parameters characterize the spin-strain
interaction of the NV center, which operates under the constraints
of $C_{3v}$ point-group symmetry~\cite{udvarhelyi2018_prb}.
The single-quantum ($\Delta m_s = \pm 1$) interaction Hamiltonian,
driving transitions between $|0\rangle$ and $|\pm 1\rangle$, is
\begin{multline}
  \frac{H_{\varepsilon 1}}{h} =
  \frac{1}{2}\!\left[\htwosix\,\exz
    - \tfrac{1}{2}\htwofive(\exx - \eyy)\right]\!\{\Sx,\Sz\}  \\
  + \frac{1}{2}\!\left[\htwosix\,\eyz
    + \htwofive\,\exy\right]\!\{\Sy,\Sz\} ,
  \label{eq:Heps1}
\end{multline}
where the shear coupling $\htwosix$ and in-plane normal coupling
$\htwofive$ are summarized with their DFT uncertainties in
Table~\ref{tab:strain_tensor}.

The same dynamic strain field simultaneously drives the
double-quantum ($\Delta m_s = \pm 2$) interaction, directly coupling
$|-1\rangle$ and $|+1\rangle$ via the transverse operators
$(S_x^2 - S_y^2)$ and $\{S_x, S_y\}$:
\begin{equation}
  \frac{H_{\varepsilon 2}}{h} = \honesix
  \bigl[(\exx - \eyy)(S_x^2 - S_y^2) + 2\exy\{S_x, S_y\}\bigr],
  \label{eq:Heps2}
\end{equation}
with $\honesix = \Xi_\perp = \SI{19660\pm90}{\mega\hertz\per\strain}$,
approximately seven times larger than the primary single-quantum
coupling.
This large double-quantum interaction produces a parasitic
AC Stark shift (Sec.~\ref{sec:acstark}).

\begin{table}[b]
\caption{\label{tab:strain_tensor}$C_{3v}$ spin-strain coupling tensor
components and their roles in the full open-system model.
All values from DFT calculations of Ref.~\cite{udvarhelyi2018_prb}.}
\begin{ruledtabular}
\begin{tabular}{lcl}
  Component & Magnitude & Transition \\
  & (MHz/strain) & \\
  \hline
  $\htwosix$\footnote{Mechanical Rabi drive} & $-2830(70)$ & $\Delta m_s = \pm 1$ \\
  $\htwofive$\footnote{Minor ellipticity ($\eta \lesssim 4.3\%$)} & $-2600(80)$ & $\Delta m_s = \pm 1$ \\
  $\honesix$ ($\Xi_\perp$)\footnote{AC Stark shift ($\Sz|_\mathcal{Q} = \sigma_z$): canceled by dynamic compensation; see Theorem~\ref{thm:Sz2}} & $19\,660(90)$ & $\Delta m_s = \pm 2$ \\
  $h_{43}$\footnote{Longitudinal detuning} & $2300$ & $\Delta m_s = 0$ \\
\end{tabular}
\end{ruledtabular}
\end{table}

\subsection{Mechanical implementation and rotating-frame parametrization}
\label{sec:membrane}

A (100)-oriented single-crystal diamond membrane~\cite{burek2014_nanoletter}
($L = \SI{10}{\micro\meter}$, $h_m = \SI{200}{\nano\meter}$)
supports nearly degenerate $(1,2)$ and $(2,1)$ flexural modes at
$f_0 = \SI{78.8}{\mega\hertz}$.  At an off-center NV site
$(L/4,L/4)$, the ideal thin-plate mode topology gives two orthogonal
local strain-channel directions of equal symmetry weight, so that
quadrature driving realizes the rotating strain field of
Eq.~\eqref{eq:rotating_strain}.
We use this membrane geometry as a hardware blueprint for the
synthetic-rotation topology, while treating
$\Omm = \SI{2.22}{\mega\hertz}$ as an effective rotating-frame
Rabi-rate target for the Regime-A channel benchmark rather than as a
calibrated three-dimensional membrane prediction.  Reaching the NV
resonance at $D\approx\SI{2.87}{\giga\hertz}$ requires a resonant
GHz-HBAR implementation, treated separately in Regime~B and
Appendix~\ref{app:ghz}.

\begin{table*}[t]
\caption{\label{tab:nv_implementations}NV operating regimes used in
the simulations.}
\begin{ruledtabular}
\begin{tabular*}{\textwidth}{@{\extracolsep{\fill}}llll}
\begin{tabular}[c]{@{}l@{}}Implementation\\ / role\end{tabular} &
\begin{tabular}[c]{@{}l@{}}Topology or\\ resonant route\end{tabular} &
$\Omm$ &
\begin{tabular}[c]{@{}l@{}}Role in\\ paper\end{tabular} \\
\colrule
Regime A channel package &
\begin{tabular}[c]{@{}l@{}}Membrane topology illustration;\\ rotating-frame $\Omm$ target\end{tabular} &
$\SI{2.22}{\mega\hertz}$ &
\begin{tabular}[c]{@{}l@{}}Rotating-frame\\ channel extraction\end{tabular} \\
Regime B GHz-HBAR &
\begin{tabular}[c]{@{}l@{}}Resonant near $\SI{2.87}{\giga\hertz}$\\ NV transition\end{tabular} &
$2.83$--$\SI{141.5}{\kilo\hertz}$ &
\begin{tabular}[c]{@{}l@{}}Resonant route\\ + envelope tracking\end{tabular} \\
\end{tabular*}
\end{ruledtabular}
\end{table*}

The same rotating-frame Hamiltonian applies once $\Omm$, drive
phases, and envelopes are specified; the achievable $\Omm$,
bandwidth, heating, phase noise, and multitone drive constraints are
implementation-specific and are therefore reported separately
(Appendix~\ref{app:ghz}).  Consequently the biased-erasure channel is
not claimed as a calibrated prediction of the flexural membrane strain
model.  It is extracted from the full Regime-A open-system control
package, while HBAR transfer-function residuals are treated as
explicit transverse-floor stress axes in the decoder analysis.

Under the RWA, the opposite-chirality matrix element oscillates at
${\sim}2D$ and is far off-resonant; a single circular drive
addresses only one arm of the $\Lambda$ system.
The holonomic protocol requires two independent drives of opposite
chirality ($\sigma_+$ near $D_+ \equiv D + \gamma_e B_z$ and
$\sigma_-$ near $D_- \equiv D - \gamma_e B_z
\approx \SI{2.73}{\giga\hertz}$), realising the full
$H_\Lambda(t)$ of Eq.~\eqref{eq:Hlambda}.

\subsection{AC Stark shift and dynamic compensation}
\label{sec:acstark}

The macroscopic strain $\varepsilon_0 \approx 7.84 \times 10^{-4}$
required for $\Omm = \SI{2.22}{\mega\hertz}$ acts simultaneously on
the large $\honesix$ tensor, producing an effective double-quantum
coupling
\begin{equation}
  \Omega_{\mathrm{DQ}} = \Xi_\perp\,\varepsilon_0
  \approx \SI{15.41}{\mega\hertz}.
  \label{eq:OmegaDQ}
\end{equation}
Although detuned by
$D \approx \SI{2.87}{\giga\hertz}$, this off-resonant drive generates
a deterministic AC Stark shift via second-order perturbation
theory~\cite{cohen_tannoudji1998_api}:
\begin{align}
  \delta_{\mathrm{AC}} &= \frac{\Omega_{\mathrm{DQ}}^2}{4D}
  = \frac{(15.41)^2}{4 \times 2870}\;\text{MHz}
  \notag\\
  &\approx \SI{20.72}{\kilo\hertz}.
  \label{eq:DeltaBS}
\end{align}
More generally, for a $C_{3v}$ spin-1 platform at fixed target
single-quantum rate,
\begin{equation}
  \varepsilon_0=\frac{\Omm}{|\htwosix|},\qquad
  \Omega_{\mathrm{DQ}}=|\honesix|\varepsilon_0
  =\left|\frac{\honesix}{\htwosix}\right|\Omm ,
\end{equation}
and therefore
\begin{equation}
  \delta_{\mathrm{AC}}
  =
  \frac{1}{4D}
  \left|\frac{\honesix}{\htwosix}\right|^2
  \Omm^2 .
  \label{eq:c3v_stark_scaling}
\end{equation}
Thus the same $C_{3v}$ symmetry class can have very different
uncompensated Stark burden depending on the tensor ratio
$|\honesix/\htwosix|$; this motivates the matched 3C-SiC comparison
in Sec.~\ref{sec:c3v_platform_rule} and Appendix~\ref{app:extended_sweeps_sm},
Sec.~7.

The double-quantum operators connect only
$|{+1}\rangle \leftrightarrow |{-1}\rangle$ and have identically
zero matrix elements with $|0\rangle$.
The shift acts exclusively within the computational manifold:

\begin{theorem}[AC Stark structure and cancellation]
\label{thm:Sz2}
The effective DQ AC Stark Hamiltonian in the rotating frame is
\begin{equation}
  H_{\mathrm{Stark}}(t) = -2\pi\,\delta_{\mathrm{AC}}
  \!\left(\cos\theta(t) + \frac{\gamma_e B_z}{D}\right)\Sz\,,
  \label{eq:Hstark_Sz}
\end{equation}
where $\theta(t)$ is the polar angle of the Bloch-sphere trajectory
and $\Sz|_{\{|{+1}\rangle,|{-1}\rangle\}} = \sigma_z$.
For the sinusoidal NGQC trajectory
$\theta(t) = \pi\sin(\pi t/\tau)$, the accumulated differential
phase over the full gate is
\begin{equation}
  \Phi_z = -2\pi\,\delta_{\mathrm{AC}}
  \!\left[J_0(\pi) + \frac{\gamma_e B_z}{D}\right]\tgate\,,
  \label{eq:stark_residual}
\end{equation}
where $J_0$ is the zeroth Bessel function
($J_0(\pi) \approx -0.304$).
Because $\theta(t)$ is prescribed by the control protocol,
a dynamic compensation field
$H_{\mathrm{comp}}(t) = -H_{\mathrm{Stark}}(t)$ cancels this
shift at every instant, removing all residual differential phase.
\end{theorem}

\begin{proof}
\textit{Part~1 (operator structure).}---Each physical strain drive
(at frequency $\omega_\pm = 2\pi(D \pm \gamma_e B_z)$) simultaneously
excites the DQ channel with amplitude
$g_\pm = (\Omega_{\mathrm{DQ}}/2)\{\cos(\theta/2),\,\sin(\theta/2)\}$.
In the doubly-rotating frame, these DQ couplings oscillate at
detunings $\Omega_\mp = 2\pi(D \mp \gamma_e B_z)$, respectively.
Standard dispersive perturbation
theory~\cite{blais2004_cqed,cohen_tannoudji1998_api,grimm2000_dipole} gives the
energy shifts
\begin{equation}
  \delta E_{\pm 1} =
  \mp\frac{g_+^2}{\hbar\Omega_-}
  \pm\frac{g_-^2}{\hbar\Omega_+}\,,
  \qquad
  \delta E_0 = 0\,,
  \label{eq:stark_shifts}
\end{equation}
where $\delta E_0 = 0$ follows from the vanishing DQ matrix elements
with $|0\rangle$.
Since $\delta E_{+1} = -\delta E_{-1}$, the shift is purely
antisymmetric, the signature of an $\Sz$ operator.
Combining over a common denominator and
using $\sin^2(\theta/2) - \cos^2(\theta/2) = -\cos\theta$ yields
Eq.~\eqref{eq:Hstark_Sz}.
The complete derivation is given in Appendix~\ref{app:dqstark}.

\textit{Part~2 (trajectory averaging).}---The accumulated differential
phase is
$\Phi_z = -2\pi\delta_{\mathrm{AC}}\!\int_0^{\tgate}\!
[\cos\theta(t') + \gamma_e B_z/D]\,dt'$.
Each NGQC sub-loop uses $\theta(t) = \pi\sin(\pi t/\tau)$ with
$\tau = \tgate/2$.  Substituting $u = \pi t/\tau$ and applying
the integral representation
$J_0(z) = \pi^{-1}\!\int_0^\pi\!\cos(z\sin u)\,du$ gives
$\int_0^\tau\!\cos\theta\,dt = \tau\,J_0(\pi)$.
Both sub-loops have identical $\theta(t)$ profiles, so they add,
yielding Eq.~\eqref{eq:stark_residual}.

\textit{Part~3 (dynamic compensation).}---Since $\theta(t)$ is
prescribed by the control protocol, the instantaneous Stark
coefficient $\cos\theta(t) + \gamma_e B_z/D$ is analytically known.
Applying $H_{\mathrm{comp}}(t) =
+2\pi\,\delta_{\mathrm{AC}}\,[\cos\theta(t) + \gamma_e B_z/D]\,\Sz$
via a DC electrode cancels the shift at every instant.
\end{proof}

Over a $\tgate = \SI{5.0}{\micro\second}$ gate, the uncompensated
bare Stark scale $\delta_{\mathrm{AC}} = \SI{20.72}{\kilo\hertz}$
could accumulate up to ${\sim}0.65\;\text{rad}$ of differential
$\sigma_z$ phase if $\cos\theta \equiv 1$.
Trajectory averaging reduces the effective coupling by
$|J_0(\pi)| \approx 0.30$, and partial cancellation between
$J_0(\pi) < 0$ and $\gamma_e B_z/D > 0$ further lowers the net
coefficient to $|J_0(\pi) + \gamma_e B_z/D| \approx 0.26$,
yielding an uncompensated residual of ${\sim}0.17\;\text{rad}$
at $\tgate = \SI{5.0}{\micro\second}$ and
${\sim}0.06\;\text{rad}$ at the optimal
$\tgate = \SI{1.833}{\micro\second}$.
Dynamic compensation eliminates this residual entirely.

\paragraph{Eigenstate protection.}
For computational basis inputs ($|{\pm 1}\rangle$), a static
$\Sz$ phase is exactly unobservable:
$|\langle{\pm 1}|e^{-i\alpha\sigma_z}|{\pm 1}\rangle|^2 = 1$
for all $\alpha$.
During driven evolution, however,
$[\Sz, H_{\mathrm{drive}}(t)] \neq 0$, so the time-dependent
AC Stark interaction does not reduce to a pure $\Sz$ rotation;
fidelity therefore degrades at sufficiently large Stark scales.

Note that this $S_z$-type AC Stark structure is specific to $S = 1$
defects with $\Delta m_s = \pm 2$ selection rules and does not hold
for $D_{3d}$ orbital systems (Sec.~\ref{sec:siv}).

Device design constraints including boundary conditions,
electrostatic degeneracy tuning, and the voltage budget are
detailed in Appendix~\ref{app:design_sm}.

\subsection{Open-system noise model}
\label{sec:noise}

The simulations assume a single NV center at depth
$d_{\mathrm{NV}} = \SI{20}{\nano\meter}$
(standard ion implantation~\cite{pezzagna2010_implant}) with
$T_1 = \SI{1}{\milli\second}$~\cite{jarmola2012_t1},
$T_{1\rho} = \SI{500}{\micro\second}$~\cite{bar-gill2013_t1rho}
(isotopically purified ${}^{12}$C,
${>}99.9\%$~\cite{balasubramanian2009_isotope,ishikawa2012_isotope}),
and $T_2^* \sim \SI{10}{\micro\second}$.

\paragraph{ME-CCE surface noise.}
Shallow NV centers at $d_{\mathrm{NV}} = \SI{20}{\nano\meter}$ are
dominantly decohered by in-sequence hopping of surface electron
spins~\cite{nagura2026,myers2014_surfacenoise,onizhuk2024_prl}.
We employ the many-body expanded cluster correlation expansion
(ME-CCE) framework (Appendix~\ref{app:decoherence}),
modeling $N = 20$ surface spins with density
$\rho_s = \SI{0.004}{\per\nano\meter\squared}$ and spatial
correlation length $r_c = \SI{5}{\nano\meter}$.
Pairwise hopping rates follow
$\Gamma_{ij} = (1/\tauc)\exp(-r_{ij}/r_c)$
with $\tauc = \SI{10}{\nano\second}$ (worst-case fast
noise~\cite{bauch2020_decoherence}),
propagated via a Gillespie kinetic Monte Carlo (KMC)
algorithm~\cite{nagura2026}.
The instantaneous noise field at the NV site is the dipolar sum
\begin{equation}
  B_z^{\mathrm{surf}}(t) = \frac{\mu_0 \mu_B}{4\pi}\sum_{i=1}^{N}
  s_i(t)\,\frac{3\cos^2\!\alpha_i - 1}{r_i^3} \,.
  \label{eq:dipolar}
\end{equation}
In the Hamiltonian this field enters as
$H_{\mathrm{surf}}/h \propto B_z^{\mathrm{surf}}(t)S_z$.
Thus the ME-CCE/KMC bath is not assumed to respect the NV point
group as a full environment; it is an explicitly modeled
$A_2$-sector system-operator perturbation.  The extracted leakage is
therefore consistent with the $A_2$ parity-filter mechanism, while
its absolute rate is obtained only from the Regime-A open-system
simulation.
A sector-projection diagnostic confirms this assignment:
Hilbert--Schmidt projection of the simulated surface-noise operator
onto equal-norm representatives
$(S_z^2,S_z,S_x,S_y)$ gives weights
$(0,1,0,0)$, and propagation of 96 ME-CCE/KMC traces through the
Regime-A SATD control with the surface term alone gives
$L_{\mathrm{surf}}=0.456\%$ (Appendix~\ref{app:noise_sectors}, Sec.~4).
In the instantaneous bright manifold, the same diagnostic directs
the first-order $A_2$ response entirely toward $|0\rangle$, whereas
an equal-norm counterfactual $E$-sector injection points toward
the logical bright state.

\paragraph{Variance normalization.}
The $N = 20$ discrete bath under-samples the macroscopic surface
integral, requiring a finite-size correction.
We employ a static per-lattice geometric scaling factor that
preserves the critical inter-trajectory DC offsets responsible for
inhomogeneous broadening ($T_2^*$), rather than per-trajectory
normalisation which would destroy quasi-static noise structure.

\paragraph{Lindbladian dissipation.}
Three collapse operators are integrated into the master equation:
(i)~spin-lattice relaxation ($T_1 = \SI{1}{\milli\second}$) via
$|\pm 1\rangle \to |0\rangle$ decay channels;
(ii)~rotating-frame depolarization
($T_{1\rho} = \SI{500}{\micro\second}$), contributing
${\sim}0.06\%$ at the optimal gate time; and
(iii)~pure dephasing from the ME-CCE stochastic field.

\subsection{Numerical implementation}
\label{sec:numerics}

The QuTiP~5.0.4 library~\cite{qutip,qutip5}, leveraging the
\texttt{mesolve} density-matrix propagator, executes the
open-system simulations across a 384-core computing cluster.
Each sweep point averages $500$ independent Monte Carlo noise
trajectories with $2000$ time steps and solver tolerances
$\mathrm{atol} = 10^{-10}$, $\mathrm{rtol} = 10^{-8}$.
The core validation suite contains ten parametric sweeps, spanning
$11$--$121$ grid points each and yielding $\sim\!100{,}000$ total
trajectory-sweep configurations.
The supplemental validation ledger also includes an additional
Floquet/lab-frame multitone check, reported separately as Sweep~12,
and the Regime-E SiV single-shot benchmark as a separate
universal-control model calculation.

Gate fidelity is assessed at the process level using three
complementary metrics.
The \emph{Uhlmann--Jozsa state fidelity} for a single input
$\rho_0 = |{-1}\rangle\langle{-1}|$ serves as a legacy baseline:
\begin{equation}
  F_{\mathrm{state}}
  = \!\left(\mathrm{Tr}\!\sqrt{\smash[b]{\sqrt{\rho_{\mathrm{id}}}\;
  \rho_{\mathrm{act}}\;\sqrt{\rho_{\mathrm{id}}}}}\right)^{\!2}.
  \label{eq:fidelity}
\end{equation}
We propagate an informationally complete (IC) set of six
input states $\{|{\pm x}\rangle, |{\pm y}\rangle,
|{\pm z}\rangle\}$ through the same noise realization to
reconstruct the $4\times 4$ Choi matrix
(the process-level representation of the quantum channel)
$\Lambda$ of the two-dimensional channel.
The \emph{entanglement (process) fidelity} is
\begin{equation}
  F_e = \langle\Phi_U| \Lambda |\Phi_U\rangle, \qquad
  |\Phi_U\rangle = (U_{\mathrm{target}}\otimes I)|\Phi^+\rangle,
  \label{eq:Fe}
\end{equation}
where $|\Phi^+\rangle = (|00\rangle + |11\rangle)/\sqrt{d}$
is the maximally entangled state, and the \emph{average gate fidelity} follows from the Horodecki
formula~\cite{horodecki1999,nielsen2002_gatefidelity}:
\begin{equation}
  \Favg = \frac{d\,F_e + 1}{d + 1}, \qquad d = 2.
  \label{eq:Favg}
\end{equation}

\paragraph{Post-selection and leakage-aware metrics.}
Projecting each $3\times 3$ qutrit output onto
the computational subspace and renormalizing generates a
$2\times 2$ Choi matrix; the derived $\Favg$ and $F_e$ are
\emph{conditional} gate fidelities.
The mean survival probability is
$p_{\mathrm{surv}} = 1 - \overline{\mathcal{L}}$
(where $\overline{\mathcal{L}}$ is the trajectory-averaged
leakage),
and the \emph{effective (unconditional) average gate fidelity} is
\begin{equation}
  \Favg^{\mathrm{eff}}
    = p_{\mathrm{surv}}\,\Favg^{\mathrm{cond}},
  \label{eq:Feff}
\end{equation}
which treats all leaked population as completely erroneous.
All ensemble statistics are quoted with $95\%$ bootstrap confidence
intervals ($2000$ resamples).

\section{SiV$^-$ orbital $\Lambda$-control benchmark}
\label{sec:siv}

The group-IV vacancy centers in diamond
(SiV$^-$, GeV$^-$, SnV$^-$, and PbV$^-$) possess $D_{3d}$
point-group symmetry and an
\emph{orbital} $E_g$ ground-state doublet split by spin-orbit
coupling into two Kramers doublets separated by
$\Delta_{\mathrm{SO}}$.
Because the strain tensor decomposes as
$\Gamma_\varepsilon = 2A_{1g} \oplus 2E_g$ and the inter-branch
transition operators carry $E_g$ symmetry~\cite{meesala2018_siv},
Theorem~\ref{thm:synrot} applies directly:
$E_g \otimes E_g \supset A_{1g}$ with multiplicity~1 guarantees
the dot-product coupling
$H = f_\perp(\varepsilon_1\,L_1 + \varepsilon_2\,L_2)$,
where $L_{1,2}$ are orbital pseudo-spin operators within the
$E_g$ doublet.

\subsection{$D_{3d}$ orbital-strain coupling}
\label{sec:d3d}

The three-level Hilbert space comprises two computational states
within the lower Kramers doublet
(a time-reversal-protected degenerate pair)
($|{+1}\rangle \equiv |g_-{\downarrow}\rangle$,
$|{-1}\rangle \equiv |g_-{\uparrow}\rangle$)
and the auxiliary upper-branch state
$|0\rangle \equiv |g_+\rangle$,
with splitting
$\Delta_{\mathrm{SO}} = \SI{48}{\giga\hertz}$~\cite{hepp2014_siv}.
The transverse orbit--strain coupling constant
$f_\perp = \SI{1.3}{\peta\hertz\per\strain}$~\cite{meesala2018_siv}
is approximately $5\times 10^5$ times larger than the NV
spin--strain coupling $h_{26} \approx 3$~GHz/strain,
reflecting the fundamental difference between direct
electrostatic crystal-field modulation (orbital) and
relativistic spin-orbit-mediated coupling (spin).
This yields mechanical Rabi frequencies
$\Omm = f_\perp \varepsilon_0$ reaching hundreds of MHz
at experimentally demonstrated strain
levels~\cite{cornell2025}, enabling gate times
in the low-nanosecond regime.

Two critical differences from NV emerge:
(i)~the AC Stark shift arising from the $\honesix$-type
double-quantum tensor is absent, since $D_{3d}$ orbital systems
lack the $\Delta m_s = \pm 2$ selection rules that generate the
NV Stark structure (Sec.~\ref{sec:acstark}), eliminating the
need for DC compensation; and
(ii)~the dominant dissipation channel is reversed: orbital
relaxation drives population from the auxiliary state \emph{into}
the computational subspace (depolarization), rather than out of it
(erasure).

\subsection{Open-system parameters}
\label{sec:siv_noise}

The simulations use Lindbladian orbital $T_1$ decay with
$T_1(\mathrm{mK}) = \SI{143}{\nano\second}$
(derived from the projected single-phonon-limited rate at
$\Delta_{\mathrm{SO}} = \SI{48}{\giga\hertz}$ via the $T^5$
phonon scaling law~\cite{jahnke2015_siv}) and
$T_1(\SI{4}{\kelvin}) = \SI{40}{\nano\second}$~\cite{jahnke2015_siv}.
Stochastic surface noise is not included: the orbital
$\Lambda$ system is insensitive to the magnetic-noise channels
that dominate NV decoherence, and the relevant orbital dephasing
mechanisms require separate characterization.
The reported fidelities are therefore $T_1$-limited upper bounds.
The scale $2/\Omm$ is the bright-state commensurability time used
for the ideal lower-manifold/SATD resource estimate.  In the
conservative Regime~C benchmark used below, the lower-manifold SATD
Hamiltonian is removed; the reported $\tgate^*$ values are selected by
a direct gate-time sweep and are lengthened relative to $2/\Omm$ to
recover adiabaticity.

For the SiV single-shot benchmark, we use the same orbital
$\Lambda$ manifold but switch from the echo-lune trajectory to the
bright-state compiler of Sec.~\ref{sec:single_shot_compiler}.
The control resources are the two orbit-strain $\Lambda$ legs and a
synchronized scalar detuning; no lower-doublet SATD actuator is
introduced.
The orbital relaxation model sends auxiliary-state population equally
into the two logical states.
Orbital dephasing, charge noise, and actuator transfer-function
errors are not included, so the reported fidelities should be
interpreted as orbital-$T_1$-limited control benchmarks.

\paragraph{Control transfer across platforms.}
The dot-product form, synthetic circular polarization via quadrature
drive, and reduction to a single-arm $\Lambda$ Hamiltonian transfer
to all five platforms.
The holonomic gate trajectory operates on the abstract $H_\Lambda$
and is platform-independent.
The SATD Hamiltonian requires a physical lower-doublet control: NV
supplies it through the resonant DQ strain tone derived in
Appendix~\ref{app:satd_hardware}, while Regime~C uses a conservative SiV
$\Lambda$-leg-only benchmark and reserves SATD-enhanced SiV
performance as a hardware-contingent upper bound.
However, the noise environment differs qualitatively: whereas NV
$T_1$ relaxation drives population from the computational
subspace into the optically distinguishable $|0\rangle$
(erasure-convertible leakage), SiV orbital relaxation drives
the auxiliary state \emph{into} the computational subspace,
producing depolarization (in-subspace Pauli errors) rather than
erasure (detectable leakage out of the computational subspace).

\section{Results}
\label{sec:results}

We label each parametric study as a numbered sweep
(Sweeps~1--10) for cross-referencing with the simulation code
and the SM.

\subsection{NV gate-time optimization (Sweep~7)}
\label{sec:sweep7}
\label{sec:nv_results}

Sweep~7 is the central result: process-level fidelity versus
$\tgate = 1.0$--$\SI{5.0}{\micro\second}$ [Fig.~\ref{fig:gatetime},
Table~\ref{tab:gatetime}].
All six informationally complete input states yield $\Favg$, $F_e$,
leakage, and bootstrap confidence intervals at every grid point.

The legacy (single-state) fidelity $F_{\mathrm{leg}}$, computed
from the input $|{-1}\rangle$ alone
(Eq.~\eqref{eq:fidelity}), decreases monotonically:
\begin{equation}
  \varepsilon_{\mathrm{leg}}(T)
  \approx 0.003\% + 0.14\%/\mu\text{s} \times T,
  \label{eq:error_leg}
\end{equation}
where the near-zero intercept confirms SATD eliminates coherent
leakage.
In contrast, $\Favg$ oscillates with period
$\Delta T \approx \SI{0.83}{\micro\second}$
($1/\Delta T \approx \Omm/2$), peaking at ``magic'' times
when the bright-state dynamical phase is commensurate with $2\pi n$
and dropping as low as $90\%$ between them, a discrepancy of up to
$9.5\%$ that underscores why process tomography is essential.

At the optimal Regime~A magic time
$\tgate = \SI{1.833}{\micro\second}$:
\begin{equation}
  \boxed{\Favg = 99.88\%,\quad F_e = 99.82\%
  \quad (\tgate = \SI{1.833}{\micro\second})},
  \label{eq:bestF}
\end{equation}
This operating point uses $\Omm=\SI{2.22}{\mega\hertz}$ and defines
the rotating-frame NV channel benchmark; the GHz-HBAR projection is reported separately in
Sweep~9.
The 95\% bootstrap confidence interval (CI) is
$\Favg \in [99.882\%,\, 99.884\%]$,
mean leakage $\overline{\mathcal{L}} = 0.484\%$, and
\begin{equation}
  \Favg^{\mathrm{eff}}
    = 0.995158 \times 0.998830
    = 99.40\%.
  \label{eq:Feff_val}
\end{equation}
The conditional $\Favg = 99.88\%$ corresponds to per-gate
depolarizing error $r \approx 0.18\%$.
The leakage is entirely detectable and convertible to a
heralded erasure because $m_s = 0$ exhibits high
photoluminescence (PL) contrast.

\begin{table}[b]
\caption{\label{tab:gatetime}Gate-time compression results for the
composite NGQC + SATD protocol (200 trajectories/point).
$F_{\mathrm{leg}}$ is the single-state $|{-1}\rangle$ fidelity;
$\Favg$ and $F_e$ are the process-level average gate and
entanglement fidelities from 6-state IC tomography.
Rows marked ($\star$) are ``magic'' gate times where the bright-state
dynamical phase is commensurate with $2\pi n$.}
\begin{ruledtabular}
\begin{tabular}{ccccc}
  $\tgate$ ($\mu$s) & $F_{\mathrm{leg}}$ (\%) &
  $\Favg$ (\%) & $F_e$ (\%) & Leak.\ (\%) \\
  \hline
  $1.000^{\star}$ & $99.89$ & $99.80$ & $99.70$ & $3.07$ \\
  $1.333$         & $99.86$ & $90.22$ & $85.33$ & $0.18$ \\
  $1.833^{\star}$ & $99.81$ & $\mathbf{99.88}$ & $\mathbf{99.82}$ & $0.48$ \\
  $2.000$         & $99.79$ & $98.10$ & $97.15$ & $7.21$ \\
  $2.667^{\star}$ & $99.72$ & $99.83$ & $99.75$ & $0.55$ \\
  $3.500^{\star}$ & $99.63$ & $99.63$ & $99.45$ & $3.24$ \\
  $4.500^{\star}$ & $99.53$ & $99.71$ & $99.57$ & $0.27$ \\
  $5.000$         & $99.48$ & $90.46$ & $85.69$ & $1.12$ \\
\end{tabular}
\end{ruledtabular}
\end{table}

\subsection{Control robustness and implementation projections (Sweeps~1--6, 9--10)}
\label{sec:sweep1}
\label{sec:sweep2}
\label{sec:sweep4}
\label{sec:sweep5}
\label{sec:sweep6}
\label{sec:sweep9}
\label{sec:sweep10}

\paragraph{Geometric robustness (Sweep~2).}
Figure~\ref{fig:robustness} compares the gate error of the
holonomic protocol at full drive against a reduced-drive
dynamical baseline ($\Omm/2$, no counter-diabatic correction)
across four decades of surface hopping rate
$\ghop = 10^6$--$\SI{e10}{\hertz}$.
The holonomic gate error forms a noise-immune plateau at
${\sim}0.50\%$ (variation $\pm 0.006\%$),
consistently below the dynamical gate
($0.54\%$--$0.63\%$, variation $\pm 0.04\%$).
The dynamical gate error decreases monotonically with increasing
$\ghop$ as motional narrowing sets in beyond the crossover rate
$\gamma_A = 1/\tauc^A$ (dashed line), whereas the holonomic
plateau is insensitive to this crossover.
This dual advantage, a lower error floor \emph{and} $7\times$
smaller variation, is intrinsic to the geometric phase mechanism:
the holonomy depends on the enclosed solid angle, not on the
instantaneous noise realization.

\begin{figure}[t]
  \centering
  \includegraphics[width=\columnwidth]{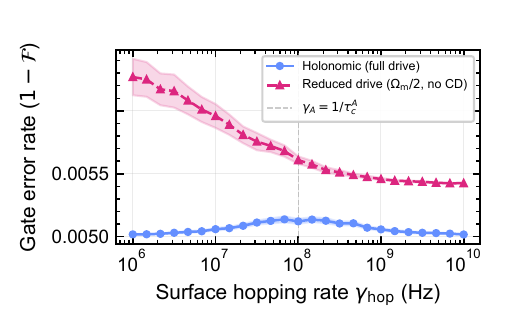}
  \caption{\label{fig:robustness}%
  Geometric robustness of the holonomic gate against surface noise
  (Sweep~2).
  Gate error rate $(1 - \Favg)$ versus surface hopping rate
  $\ghop$ for the composite NGQC + SATD holonomic gate at full
  drive strength (blue circles) and a reduced-drive dynamical
  baseline ($\Omm/2$, no counter-diabatic correction; pink
  triangles).
  Shaded bands indicate $2\sigma$ bootstrap confidence intervals.
  The holonomic gate maintains a flat plateau
  ($\pm 0.006\%$ variation) across four decades of $\ghop$,
  whereas the dynamical gate exhibits $7\times$ larger variation
  and converges only partially under motional narrowing.
  The dashed vertical line marks $\gamma_A = 1/\tauc^A$,
  the crossover between quasi-static and motionally narrowed
  noise regimes.}
\end{figure}

\paragraph{Counter-diabatic validation (Sweeps~5--6).}
Sweep~5 scans the DRAG parameter $\lambda = 0$--$2.0$
[Fig.~\ref{fig:optimization}]:
fidelity peaks sharply at $\lambda = 1.0$
($F_{\mathrm{noisy}} = 99.49\%$), confirming exact cancellation
of first-order non-adiabatic coupling.
Sweep~6 validates the SATD correction strength $\aCD = 0$--$2.0$
[Fig.~\ref{fig:optimization}(b)]:
a razor-sharp symmetric peak at $\aCD = 1.0$
($F = 99.76\%$); at $\aCD = 0$ and $2.0$, fidelity collapses to
$\sim\!26\%$.
Achieving $F > 99\%$ requires $|\aCD - 1.0| < 0.01$,
confirming that the analytically exact counter-diabatic strength
predicted by Berry's theorem~\cite{berry2009_transitionless}
is verified numerically.

\begin{figure*}[t]
  \centering
  \includegraphics[width=\textwidth]{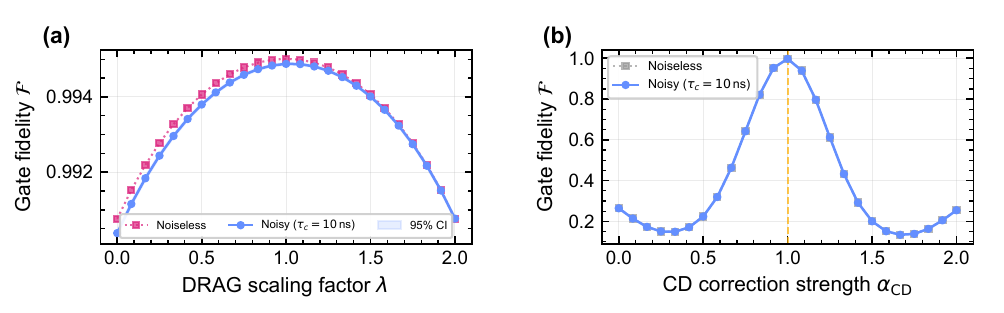}
  \caption{\label{fig:optimization}%
  Counter-diabatic validation.
  (a)~Sweep~5: noiseless (solid) and noisy (dashed) fidelity versus
  counter-diabatic parameter~$\lambda$; the peak at $\lambda = 1.0$
  confirms the $99.49\%$ composite ceiling.
  (b)~Sweep~6: SATD correction strength $\aCD$; the razor-sharp
  symmetric peak at $\aCD = 1.0$ validates exact counter-diabatic
  correction ($F = 99.76\%$).}
\end{figure*}

\paragraph{Extended parametric studies.}
Additional sweeps confirm robustness against hardware and
implementation variations (Appendix~\ref{app:extended_sweeps_sm}):
fidelity versus Rabi frequency $\Omm$ (Sweep~1, variation
$\leq 0.09\%$), frequency detuning tolerance (Sweep~3,
$< 0.04\%$ variation over $\SI{500}{\kilo\hertz}$), boundary-condition
sensitivity (Sweep~4), GHz HBAR extension (Sweep~9,
$\Favg=99.50$--$99.85\%$ at $\tgate=\SI{2}{\micro\second}$ across
the projected HBAR scenarios, assuming the same ideal rotating-frame
control Hamiltonian while bounding HBAR-specific phase noise, heating,
and transducer filtering separately), and
quadrature drive imbalance (Sweep~10,
$\Favg \geq 99.5\%$ on the central plateau, e.g.\
$|r - 1| \lesssim 4\%$ and
$|\delta\varphi| \lesssim 4^\circ$; over the broader
$6\%$--$6^\circ$ box the minimum sampled value is $99.32\%$).
A Floquet multitone
validation retaining leading counter-rotating terms from the two
$\Lambda$ tones and the DQ SATD tone gives only
$2.4\times10^{-5}$ additional infidelity in the ideal-bandwidth
limit and $1.0\times10^{-4}$ for a \SI{200}{\mega\hertz}
bandwidth with $0.2^\circ$ phase noise and $0.2\%$ amplitude
imbalance (Appendix~\ref{app:extended_sweeps_sm}).  The same test shows that an
unpredistorted $Q=10^4$ GHz-HBAR transfer function is narrower than
the \SI{1.833}{\micro\second} Regime-A benchmark envelope.  Regime~B
therefore converts the benchmark into a resonator-design target:
predistortion, lower effective~$Q$, or a slower HBAR-specific pulse
sets the envelope-tracking strategy.

\begin{figure}[t]
  \centering
  \includegraphics[width=\columnwidth]{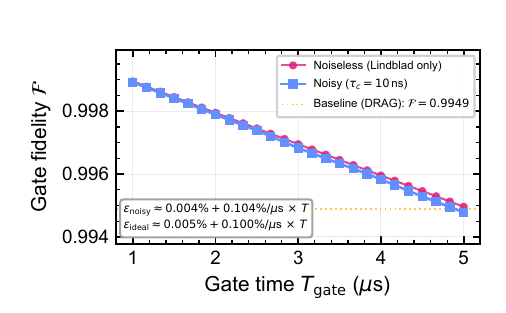}
  \caption{\label{fig:gatetime}%
  Gate-time compression for Regime~A
  (Sweep~7, $\Omm = \SI{2.22}{\mega\hertz}$):
  legacy single-state fidelity (ideal: red circles; noisy: blue
  squares) versus $\tgate$, with the smooth linear envelope
  (Eq.~\eqref{eq:error_leg}).  Process-level $\Favg$
  (Table~\ref{tab:gatetime}) oscillates due to bright-state
  dynamical phase, peaking at magic gate times; at the optimum
  $\tgate = \SI{1.833}{\micro\second}$,
  $\Favg^{\mathrm{cond}} = 99.88\%$
  ($\Favg^{\mathrm{eff}} = 99.40\%$).
  GHz HBAR extension results are in Appendix~\ref{app:extended_sweeps_sm}.
  }
\end{figure}

\subsection{NV error budget}
\label{sec:errorbudget}

The error budget at the optimal Regime~A operating point
($\tgate = \SI{1.833}{\micro\second}$, composite NGQC + SATD) is
presented in Table~\ref{tab:errorbudget}
(full parameter listing in Appendix~\ref{app:params},
Table~\ref{tab:params}).
All entries marked
``Lindblad,'' ``ME-CCE,'' ``coherent,'' ``AC Stark,'' and
``South Pole'' are \emph{demonstrated in simulation}.
Within the conditional process-level error
($\varepsilon = 0.12\%$), the largest contributor is
Lindbladian $T_1/T_{1\rho}$ dissipation;
however, the overall error budget is dominated by
leakage ($\overline{\mathcal{L}} = 0.48\%$),
which exceeds the in-subspace error by a factor of~$4$.

\begin{table}[h!]
\caption{\label{tab:errorbudget}Error budget at the optimal Regime~A
operating point ($\tgate = \SI{1.833}{\micro\second}$,
composite NGQC + SATD, $d_{\mathrm{NV}} = \SI{20}{\nano\meter}$,
$T_{1\rho} = \SI{500}{\micro\second}$).  The upper rows decompose the conditional
process-level error ($\varepsilon = 0.12\%$); the lower rows show
the impact of dissipative leakage on the unconditional (effective)
fidelity.}
\begin{ruledtabular}
\begin{tabular}{lc}
  Error source & Contribution \\
  \hline
  Lindblad $T_1/T_{1\rho}$\footnote{Reduced by isotopic enrichment} & ${\sim}0.12\%$ \\
  \quad of which $T_{1\rho}$\footnote{$5\times$ lower than at $T_{1\rho}=\SI{200}{\micro\second}$} & ${\sim}0.06\%$ \\
  \quad of which $T_1$\footnote{Sub-dominant} & ${\sim}0.06\%$ \\
  ME-CCE surface noise\footnote{Reduced by deeper implant ($20$ vs.\ $10$~nm)} & ${\sim}0.02\%$ \\
  Coherent (SATD residual)\footnote{Near-zero intercept} & ${\sim}0.003\%$ \\
  AC Stark shift\footnote{$\Sz$ on $\mathcal{Q}$; NGQC-suppressed, DC-compensated} & $\lesssim 0.01\%$ (comp.) \\
  South Pole singularity\footnote{Neutralised ($\sin\theta=0$ at pole)} & $0.000\%$ \\
  \hline
  \multicolumn{2}{c}{\textit{Unconditional (effective) fidelity}} \\
  \hline
  \textbf{Legacy envelope}\footnote{$F_{\mathrm{leg}} = 99.81\%$ at the optimal sampled point} & $\mathbf{\sim 0.19\%}$ \\
  \textbf{Process-level ($\Favg$, cond.)}\footnote{Magic-time refocusing; $\Favg = 99.88\%$} & $\mathbf{\sim 0.12\%}$ \\
  Leakage ($\overline{\mathcal{L}}$)\footnote{Full noisy Sweep~7 run; the separate surface-only sector diagnostic gives $L_{\mathrm{surf}}=0.456\%$} & $0.484\%$ \\
  \textbf{Effective ($\Favg^{\mathrm{eff}}$, uncond.)}\footnote{$\Favg^{\mathrm{eff}} = 99.40\%$}
    & $\mathbf{\sim 0.60\%}$ \\
\end{tabular}
\end{ruledtabular}
\end{table}

\paragraph{Leakage decomposition.}
SATD eliminates \emph{coherent} leakage to $0.003\%$.
Stochastic surface noise perturbs the instantaneous eigenstates,
generating noise-mediated CD mismatch that drives population into
$|0\rangle$.
A separate sector-projection run using the actual
$B_z^{\mathrm{surf}}(t)S_z$ traces but omitting Lindblad relaxation
gives $L_{\mathrm{surf}}=0.456\%$, confirming that the leakage channel
is produced by the modeled $A_2$ surface perturbation rather than by
an assumed point-group-symmetric bath.
The full noisy Sweep~7 run gives
$\overline{\mathcal{L}}=0.484\%$ using 200 trajectories, whereas
the surface-only sector diagnostic gives $L_{\mathrm{surf}}=0.456\%$
using 96 trajectories with Lindblad relaxation omitted.
We therefore use the full-run leakage for the decoder-facing erasure
rate and use the surface-only number only to identify the $A_2$ origin
of the leakage channel.

\paragraph{Mechanism decomposition.}
The error reduction (from baseline $1.10\%$ to $0.19\%$
on the same single-input metric, about $6\times$, and $0.12\%$
at the process level;
Appendix~\ref{app:extended_sweeps_sm})
is achieved through three mechanisms acting in concert:
(1)~singularity neutralization,
(2)~exact counter-diabatic driving ($\lambda \equiv 1.0$, coherent
error $\sim\!0.003\%$), and
(3)~gate-time compression (dissipative error
${\sim}0.14\%/\mu\text{s}$, minimized at the optimal magic
time $\tgate = \SI{1.833}{\micro\second}$).

\subsection{SiV orbital $\Lambda$-control result}
\label{sec:siv_results}

We next apply the same synthetic-rotation $\Lambda$-leg control to a
physically distinct system whose encoding (orbital vs.\ spin),
symmetry ($D_{3d}$ vs.\ $C_{3v}$), timescale (nanoseconds vs.\
microseconds), and dominant dissipation channel (depolarization
vs.\ erasure) all differ from the NV case.
Following the Regime~C scope in Table~\ref{tab:regimes}, the SiV
benchmark sets $\aCD=0$ and uses only the physical $\Lambda$ legs
established by the orbit-strain selection rule.
Table~\ref{tab:siv_results} summarizes the results.
The best conservative point reaches $\Favg=96.32\%$ on the
SiV$^-$ orbital $\Lambda$ system at millikelvin temperatures.
The ideal lower-manifold SATD simulation gives $\Favg=99.56\%$ at
$\tgate=\SI{6.5}{\nano\second}$, but this number is treated as a
control-resource upper bound pending a hardware-level SiV SATD
channel.

\begin{table}[t]
\caption{\label{tab:siv_results}%
Conservative SiV$^-$ holonomic $Z(\pi)$ gate benchmark with the
lower-manifold SATD Hamiltonian removed ($\aCD=0$).
$\Favg$ and $F_e$ are computed from the Choi matrix over the
computational subspace at the best millikelvin point found in the
gate-time sweep.}
\begin{ruledtabular}
\begin{tabular}{lcccc}
$\Omm$ & $\tgate^*$ &  $\Favg$ & $F_e$ & Leak. \\
(MHz) & (ns) & (mK) & (mK) & (mK) \\
\colrule
100 & 99.2  & 91.18\% & 86.78\% & 7.96\% \\
300 & 46.2  & \textbf{96.32\%} & \textbf{94.48\%} & 3.85\% \\
500 & 123.3 & 90.14\% & 85.21\% & 8.64\% \\
\end{tabular}
\end{ruledtabular}
\end{table}

The benchmark shows that synthetic-rotation $\Lambda$ control remains
quantitatively useful in the orbital platform.  In this conservative
regime, the gate time must be lengthened to recover adiabaticity, so
orbital $T_1$ relaxation becomes the dominant limitation.

\subsection{Single-shot universal control: generic and SiV implementations}
\label{sec:single_shot_results}

\subsubsection{Generic / NV-compatible single-shot suite (Regime D)}

Regime~D validates the second compiler rather than a second QEC
channel.  The simulation uses the same symmetry-generated complex
$\Lambda$ legs as Regime~A, adds the scalar detuning required by
Eq.~\eqref{eq:ss_controls}, and applies a smooth
$\sin^2$ common envelope with
$\Omega_{\max}=\SI{2.220}{\mega\hertz}$ and
$\tgate=\SI{0.9009}{\micro\second}$.
The nominal broadband surface-acoustic-wave in-phase/quadrature
(SAW/IQ) implementation gives
$F_{\mathrm{eff}}\simeq99.86\%$ across a representative universal
gate set, where $F_{\mathrm{eff}}$ counts leakage as error
(Table~\ref{tab:single_shot_suite}).

\begin{table}[h!]
\caption{\label{tab:single_shot_suite}
Regime~D single-shot universal-gate suite for the nominal
\SI{200}{\mega\hertz} SAW/IQ waveform channel.  The compiler is not
used in the XZZX overhead estimate.}
\begin{ruledtabular}
\begin{tabular}{lccc}
Gate & $F_{\mathrm{eff}}$ & $F_{\mathrm{cond}}$ & Leakage \\
\colrule
$Z_\pi$ & $99.8703\%$ & $99.9358\%$ & $0.0655\%$ \\
$X_\pi$ & $99.8675\%$ & $99.9433\%$ & $0.0758\%$ \\
$X_{\pi/2}$ & $99.8601\%$ & $99.9377\%$ & $0.0776\%$ \\
Generic $0.73\pi$ & $99.8657\%$ & $99.9405\%$ & $0.0749\%$ \\
\end{tabular}
\end{ruledtabular}
\end{table}
\FloatBarrier

The non-Abelian diagnostic is performed independently at the unitary
compiler level: $X_{\pi/2}Z_{\pi/2}$ and
$Z_{\pi/2}X_{\pi/2}$ each match their intended target, but the two
composed unitaries have average fidelity $0.5$ relative to one
another, as expected from Eq.~\eqref{eq:ss_commutator}.  The residual
single-shot channels are only moderately biased and gate-dependent, so
they are not used for architecture-level fault-tolerance claims.

\subsubsection{SiV orbital single-shot suite (Regime E)}
\label{sec:siv_single_shot_results}

We next apply the same bright-state compiler to the SiV orbital
$\Lambda$ manifold.
This test is motivated by the absence of an identified SiV
lower-doublet SATD actuator: unlike the SATD echo-lune protocol, the
single-shot compiler requires only proportional $\Lambda$-leg control
and scalar detuning.
The primary millikelvin benchmark uses
$\Omega_{\mathrm{peak}}=\SI{300}{\mega\hertz}$,
$\tgate=2/\Omega_{\mathrm{peak}}=\SI{6.667}{\nano\second}$,
$T_{1,\mathrm{orb}}=\SI{142.7}{\nano\second}$, and peak strain
$\epsilon_{\mathrm{peak}}=2.308\times10^{-7}$.
The resulting four-gate suite is summarized in
Table~\ref{tab:siv_single_shot_suite}.
The \SI{100}{\mega\hertz}, \SI{300}{\mega\hertz}, and
\SI{500}{\mega\hertz} millikelvin points, together with a
\SI{4}{\kelvin} temperature-stress test, are reported in
Appendix~\ref{app:single_shot}, Sec.~5.

\begin{table*}[t]
\caption{\label{tab:siv_single_shot_suite}
Regime~E SiV single-shot bright-state universal-control suite.
The benchmark uses only the two orbit-strain $\Lambda$ legs and
synchronized scalar detuning; no lower-doublet SATD actuator is
assumed.
The dominant error is orbital-$T_1$ depolarization during auxiliary
occupation.
The result is not used in the biased-erasure or XZZX-channel
analysis.}
\begin{ruledtabular}
\begin{tabular}{lccccc}
Gate & $\tgate$ & $\epsilon_{\mathrm{peak}}$ &
$F_{\mathrm{eff}}$ & $F_{\mathrm{cond}}$ & Leakage \\
\colrule
$Z_\pi$ & \SI{6.667}{\nano\second} & $2.308\times10^{-7}$ &
$99.5221\%$ & $99.7488\%$ & $2.275\times10^{-3}$ \\
$X_\pi$ & \SI{6.667}{\nano\second} & $2.308\times10^{-7}$ &
$99.5221\%$ & $99.7488\%$ & $2.275\times10^{-3}$ \\
$X_{\pi/2}$ & \SI{6.667}{\nano\second} & $2.308\times10^{-7}$ &
$99.6699\%$ & $99.7979\%$ & $1.283\times10^{-3}$ \\
Arbitrary axis $0.73\pi$ & \SI{6.667}{\nano\second} &
$2.308\times10^{-7}$ & $99.5671\%$ & $99.7621\%$ &
$1.957\times10^{-3}$ \\
\end{tabular}
\end{ruledtabular}
\end{table*}

The SiV single-shot result changes the SiV narrative.
In Regime~C, removing the lower-doublet SATD Hamiltonian forces a
slower $\Lambda$-leg-only echo-lune trajectory and leaves the best
millikelvin point at $\Favg=96.32\%$.
In Regime~E, the bright-state compiler supplies a shortcut that
remains expressible through the SiV-available $\Lambda$ legs plus
scalar detuning.
The residual channel is qualitatively different from the NV
biased-erasure channel: orbital relaxation returns auxiliary
population to the logical manifold, so the dominant error is
in-subspace depolarization rather than detectable leakage.
Consequently, Regime~E supports a high-fidelity SiV universal-control
layer, but it is not a QEC-channel-engineering result.

\subsection{Conditional symmetry selection of the biased-erasure channel}
\label{sec:erasure}

The $\Lambda$-system traversal through $|0\rangle$, which
introduces leakage that a direct Rabi $\pi$-pulse avoids, becomes
a structural advantage under quantum error correction: the leaked
population is optically distinguishable and therefore convertible
to a heralded erasure error, whereas the Rabi gate's in-subspace
Pauli rotations are not.
This erasure-conversion route is available to $\Lambda$-system gate
topologies whenever leakage to the auxiliary state is efficiently
distinguishable~\cite{wu2022_erasure}.

However, erasure conversion alone underestimates the
code-capacity decoder advantage.
The same $\Gamma_E$ irrep that enables synthetic rotation
(Theorem~\ref{thm:synrot}) also classifies the noise
protection of the holonomic gate:
the $\Lambda$ Hamiltonian has bright eigenstates
$|B_\pm\rangle = (|B\rangle \pm |0\rangle)/\sqrt{2}$ at
energies $\pm\Omm/2$, and a quasi-static perturbation
$V = \sigma\,\hat{V}$ couples the dark state to these
bright states via two channels whose interference pattern
depends on the irrep of~$\hat{V}$.

The symmetry argument has two logically distinct parts. First,
point-group symmetry fixes the allowed operator sectors of the
strain-defect coupling and of weak system-operator perturbations. Second, the
$\Lambda$-manifold dynamics maps those sectors into different
geometric-error and leakage channels. Representation theory fixes the
sector hierarchy, but not the scalar rates. In the Regime-A
simulation, the ME-CCE/KMC surface bath enters specifically through
$B_z^{\mathrm{surf}}(t)S_z$, an $A_2$-sector perturbation; the
resulting erasure rate is therefore a numerical output of the
open-system model, not a group-theoretic prediction. A direct
sector-projection diagnostic gives unit weight in the $A_2$ operator
sector and zero numerical weight in the $A_1$ or $E$ representatives
(Appendix~\ref{app:noise_sectors}, Sec.~4). The following
proposition proves the sector hierarchy; the corollary gives the
effective biased-erasure channel under weak noise and efficient
detection of leakage to $|0\rangle$.

\begin{proposition}[Noise-sector classification]
\label{prop:noise_sectors}
Let $\hat{V}$ transform as irrep~$\Gamma_V$ of the defect
point group, with RMS amplitude~$\sigma$.  The geometric-phase
error of the composite two-loop holonomic gate satisfies:
\begin{enumerate}
  \item[\textbf{(a)}] \textbf{$A_1$ sector}
  ($S_z^2$): $\delta\gamma_{\mathrm{geo}} = 0$
  \textup{(}$S_z^2 = I$ on~$\mathcal{Q}$\textup{)}.

  \item[\textbf{(b)}] \textbf{$A_2$ sector}
  ($S_z$): $\delta\gamma_{\mathrm{geo}} = O(\sigma^2/\Omm^2)$.
  Destructive interference between the $|B_+\rangle$ and
  $|B_-\rangle$ channels \textup{(}``parity filter''\textup{)}
  directs the first-order correction into
  $|0\rangle \perp \mathcal{Q}$, leaving the Berry connection
  invariant at $O(\sigma)$.

  \item[\textbf{(c)}] \textbf{$E$ sector}
  ($S_x,\,S_y$): $\delta\gamma_{\mathrm{geo}} = O(\sigma^3/\Omm^3)$.
  The parity filter reinforces the in-subspace $|B\rangle$
  component, modifying the Berry connection at $O(\sigma)$;
  however, the $\pi$-echo between the two NGQC lunes cancels
  the first- and second-order corrections, with the leading
  surviving Fourier mode appearing at third order.
\end{enumerate}
\end{proposition}

\noindent
The proof is given in Appendix~\ref{app:noise_sectors}.
All geometric-phase errors are at least second order in
$\sigma/\Omm$ (Table~\ref{tab:noise_hierarchy}), so the
conditional process error at the optimal magic time is
decoherence-dominated.
This classification is a general consequence of the
$\Lambda$-system structure and applies to any platform in
Table~\ref{tab:platforms}.

\begin{remark}[Parity-filter duality]
\label{rem:duality}
In the $A_2$ sector, the parity filter has a dual role.
The same bright-manifold interference that directs
$|D^{(1)}\rangle \propto |0\rangle \perp \mathcal{Q}$
(protecting the Berry connection at first order, Face~1) also
means that the leading perturbative population transfer exits
the computational subspace into the optically distinguishable
auxiliary level (Face~2).
Explicitly, the first-order dark-state correction under
$V = \sigma S_z$ evaluates to
$|D^{(1)}\rangle = -[\sigma\langle B|S_z|D\rangle/(\Omm/2)]\,
|0\rangle$,
which lies entirely outside $\mathcal{Q}$ and therefore
points toward the optically distinguishable auxiliary level,
providing a microscopic explanation for why $A_2$-type
mismatch preferentially exits the computational subspace.
For $E$-sector perturbations ($S_x, S_y$), the interference
reverses: $|D^{(1)}\rangle \propto |B\rangle \in \mathcal{Q}$,
and suppression relies instead on the $\pi$-echo.
Thus the geometric-phase protection and the erasure-dominant
noise structure of Corollary~\ref{cor:biased_erasure} are not
independent observations but dual consequences of the
$A_2$-sector parity filter.
\end{remark}

\begin{table}[h]
\caption{\label{tab:noise_hierarchy}Noise-sector suppression
hierarchy from Proposition~\ref{prop:noise_sectors} and the
sector-injection diagnostic of Appendix~\ref{app:noise_sectors}, Sec.~4.  The
direction column gives the first-order response fraction projected
onto the auxiliary state $|0\rangle$ and logical bright state
$|B\rangle$.  The scaling column gives the measured perturbative
response or phase-error slope; for the $E$ sector, $1\to3$ denotes
the change from a single open lune to the $\pi$-shifted two-lune
echo.}
\scriptsize
\begin{ruledtabular}
\begin{tabular}{llll}
  Irrep & Mechanism &
  \begin{tabular}{@{}l@{}}First-order\\ direction\end{tabular} &
  \begin{tabular}{@{}l@{}}Diagnostic\\ scaling\end{tabular} \\
  \hline
  $A_1$ & common mode & none & numerical floor \\
  $A_2$ & parity filter & $f_0=1,\ f_B=0$ & weight slope $2.000$ \\
  $E$   & bright-directed + $\pi$ echo & $f_0=0,\ f_B=1$ &
  phase slope $1\to3$ \\
\end{tabular}
\end{ruledtabular}
\end{table}

\paragraph{Numerical sector-injection verification.}
To verify that the biased-erasure structure is not merely inferred
from the final open-system channel, we performed a sector-injection
diagnostic on the ideal three-level $\Lambda$ Hamiltonian.  Equal-norm
representatives $S_z^2$, $S_z$, $S_x$, and $S_y$ were injected and
the first-order dark-state correction was projected onto the
auxiliary direction $|0\rangle$ and the logical bright direction
$|B\rangle$.  The $A_2$ perturbation $S_z$ gives
$f_{|0\rangle}=1$ and $f_{|B\rangle}=0$, while the transverse
$E$-sector perturbations $S_x,S_y$ give $f_{|0\rangle}=0$ and
$f_{|B\rangle}=1$.  The response weights scale quadratically in
$\sigma/\Omm$, and the $E$-sector geometric phase changes from
approximately first order without echo to approximately cubic with
the two-lune echo.  Thus the sector-to-channel map used in
Corollary~\ref{cor:biased_erasure} is directly verified at the
$\Lambda$-manifold level (Appendix~\ref{app:noise_sectors}, Sec.~4).  The full
four-panel diagnostic, including the surface-sector projection, is
shown in Fig.~S1.  The $A_2$ result
does not mean that $A_2$ noise vanishes; it means that its leading
perturbative population is redirected outside $\mathcal{Q}$, where
efficient $|0\rangle$ readout converts it into erasure structure.

The sector-injection response defines an irrep-to-syndrome
polarimetry protocol.  For a weak perturbation $V_\Gamma$, define
$\mathcal R_{\Gamma\to c}=\lim_{\sigma\to0}p_c[V_\Gamma]/
(\sigma/\Omm)^2$, where $c$ labels a measured response port.  In the
equal-Hilbert--Schmidt-norm convention, the clean $\Lambda$ manifold
gives a block-diagonal port response: $A_1$ is common mode, $A_2$
routes population response to the auxiliary port, and $E$ routes
population response to the logical bright port.  For a mixed perturbation
$V=\sigma[\cos\alpha\,S_z+\sin\alpha(\cos\beta\,S_x+\sin\beta\,S_y)]$,
the routing fractions obey $f_0=\cos^2\alpha$ and
$f_B=\sin^2\alpha$, independent of $\beta$ (Appendix~\ref{app:noise_sectors},
Sec.~4).  Thus the
symmetry-generated $\Lambda$ manifold acts as a crystalline irrep
polarimeter: symmetry fixes the sector-to-port map, while the device
model fixes the absolute rates.

\begin{corollary}[Conditional biased-erasure channel]
\label{cor:biased_erasure}
Under the assumptions stated above, an $A_2$-sector perturbation is
parity-filtered preferentially toward auxiliary-state leakage, while
the $E$ sector is echo-suppressed. For comparable perturbation
strengths, the no-erasure residual is then $Z$-biased, with
\begin{equation}
  \eta \equiv \frac{p_Z}{p_{XY}} \sim \frac{\Omm}{\sigma}.
\end{equation}
Leakage to the optically distinguishable auxiliary state is converted
into a known-location erasure with probability
$p_{\mathrm{era}}=\eta_{\mathrm{det}}L$.
Thus the simulated gate realizes a biased-erasure channel because the
symmetry-enforced $\Lambda$ hierarchy, the holonomic encoding, and
the $|0\rangle$ detection mechanism~\cite{robledo2011_readout} act
together.

Specifically, the extracted Regime-A per-gate noise decomposes into
four channels:
\begin{enumerate}
  \item Erasure: $p_{\mathrm{era}}=\eta_{\mathrm{det}}L=0.47\%$,
  using $\eta_{\mathrm{det}}\geq97.5\%$.  The dominant simulated
  leakage is produced by the explicit
  $B_z^{\mathrm{surf}}(t)S_z$ ME-CCE/KMC perturbation and is
  consistent with the $A_2$ parity filter; a surface-only
  sector diagnostic gives $L_{\mathrm{surf}}=0.456\%$.
  \item $Z$ dephasing: $p_Z=0.168\%$, from $T_{1\rho}$ depolarization
  and residual $A_2$-sector geometric-phase error
  $O(\sigma_{A_2}^2/\Omm^2)$.
  \item Residual depolarizing error: $p_{\mathrm{dep}}=0.012\%$, from
  imperfect erasure detection.
  \item $X/Y$ bit-flip error: at the numerical floor of the extracted
  Regime-A rotating-frame NV channel, associated with the
  echo-suppressed $E$ sector and additional thermal suppression of
  $\Delta m_S=2$ spin-flip processes at $B_z=\SI{50}{\gauss}$.
\end{enumerate}

These probabilities are not group-theoretic constants. Symmetry fixes
the allowed system-operator sectors and their perturbative hierarchy;
the Regime-A device model fixes the rates.
\end{corollary}

The nominal noise bias ratio $\eta \equiv p_Z / p_{XY}$ therefore
lies far beyond the $\eta \sim 100$ saturation point of the XZZX
advantage~\cite{bonilla_ataides2021_xzzx,tuckett2019_tailored};
finite transverse floors are treated explicitly below as a decoder
stress axis, and missed erasures are tracked as a separate validation
limit~\cite{chang2024_imperfect_erasure}.
The resulting channel is a \emph{biased-erasure} channel:
dominant erasure errors at known locations, plus a strongly
Z-biased non-erasure residual.
This places the extracted channel in the biased-erasure-noise setting
studied for XZZX-type decoding~\cite{sahay2023_biased_erasure}.
This structure arises from the crystallographic hierarchy together
with the holonomic encoding, the modeled bath weights, and erasure
detection, rather than from engineered hardware asymmetry alone.

As a conservative baseline, treating the non-erasure error as
isotropic depolarizing
($p_{\mathrm{undet}} = p_Z + p_{\mathrm{dep}} = 0.18\%$)
and using the code-capacity thresholds
$p_{\mathrm{dep}}^{\mathrm{th}} = 10.3\%$~\cite{dennis2002_tqm}
and
$p_{\mathrm{era}}^{\mathrm{th}} = 50\%$~\cite{stace2009_prl}
yields a Pauli-equivalent effective error rate
\begin{equation}
  p_{\mathrm{eff}}^{\mathrm{(iso)}}
  = p_{\mathrm{undet}}
  + \frac{p_{\mathrm{dep}}^{\mathrm{th}}}{p_{\mathrm{era}}^{\mathrm{th}}}
    \,p_{\mathrm{era}}
  = 0.28\%,
  \label{eq:peff}
\end{equation}
requiring fit-extrapolated code distance $d = 11$
($121$ data qubits) for $p_L = 10^{-10}$, a $46\%$
code-capacity overhead estimate versus the Rabi gate
($d = 15$, $225$ qubits).
This estimate, however, discards the Z-bias entirely.
The bias-aware analysis of Sec.~\ref{sec:qec_advantage}
shows that substantially larger overhead reductions are possible
within the same code-capacity model.

\subsection{Biased-erasure decoder estimates and scheduled stress diagnostic}
\label{sec:qec_advantage}

Having verified the microscopic sector-to-channel map, we next ask how
the extracted Regime-A biased-erasure channel is seen by a
code-capacity decoder.  We perform Monte Carlo QEC simulations comparing
CSS and XZZX surface codes~\cite{bonilla_ataides2021_xzzx} on
the same toric lattice under the extracted noise model of
Sec.~\ref{sec:erasure}
(5000 trials per data point, minimum-weight perfect matching
via PyMatching; full methodology in Appendix~\ref{app:qec_montecarlo}).
This QEC analysis uses the extracted Regime-A rotating-frame NV channel
from Table~\ref{tab:errorbudget}.
The decoder estimates should be read as model-channel diagnostics of
the symmetry-generated biased-erasure mechanism, not as calibrated
thresholds for the specific flexural membrane geometry of
Fig.~\ref{fig:architecture} or for a particular HBAR transfer function.
The SATD Regime-A channel is used for the biased-erasure and XZZX
overhead estimates; the single-shot Regime-D suite is used only to
establish compact universal non-Abelian control on the same
symmetry-generated $\Lambda$ manifold.
The code-capacity layer intentionally isolates decoder response to
the extracted biased-erasure channel; below we add a scheduled
two-sector detector-model stress diagnostic for the first circuit layer.
All physical error rates are uniformly scaled by a common
factor $s \in [1,\,100]$ to probe behavior near and above
threshold.

On the same toric lattice and under the same extracted
biased-erasure noise model, CSS and XZZX
exhibit opposite finite-size scaling:
CSS logical error $p_L$ increases with code distance~$d$
for $s \gtrsim 2$--$5$, indicating near-threshold or
above-threshold operation, whereas XZZX $p_L$ decreases
monotonically
with~$d$ at every tested scale including $s = 100$,
indicating deeply sub-threshold operation
(Fig.~\ref{fig:qec_threshold},
Table~\ref{tab:qec_comparison}).
At $s = 50$ and $d = 11$, the XZZX logical error rate is
$197\times$ lower than CSS ($p_L = 0.0034$ versus $0.671$).
This advantage grows with distance, the hallmark of the
two codes residing on opposite sides of a threshold
boundary.
No threshold crossing is observed for XZZX up to $100\times$
the physical noise.
This opposite scaling behavior is qualitatively significant:
it indicates that the extracted biased-erasure noise model places the
two codes on opposite sides of the threshold boundary, so that
the holonomic gate does not merely lower the physical error rate
but changes the \emph{class} of logical noise seen by the
code-capacity decoder.

The square-code overhead estimate is obtained from the fitted
finite-distance code-capacity scaling: for extrapolated
$p_L = 10^{-10}$,
the XZZX code under the nominal extracted biased-erasure model
requires $d = 9$ ($d^2 = 81$ data qubits), versus $d = 11$
($121$ qubits) for the conservative erasure-only CSS baseline
and $d = 15$ ($225$ qubits) for the Rabi gate without erasure
conversion, a $64\%$ benchmark code-capacity overhead estimate relative
to the Rabi baseline
(Table~\ref{tab:qec_comparison}).
This $64\%$ value is therefore a nominal model-channel estimate, not
an architecture-level fault-tolerance overhead.
The nominal Regime-A channel lies at the numerically resolved
transverse floor.  Since this floor is not a group-theoretic constant,
we treat $p_{XY}$ as an explicit stress parameter
(Fig.~\ref{fig:transverse_floor_envelope} and
Appendix~\ref{app:qec_montecarlo}, Sec.~6).  At
$\eta_{\rm det}=97.5\%$, the nominal floor gives the $d=9$ / $64\%$
code-capacity saving proxy; finite transverse floors
$p_{XY}=10^{-6}$--$10^{-4}$ move the estimate to $d=11$ / $46.2\%$,
and $p_{XY}=10^{-3}$ moves it to $d=13$ / $24.9\%$.  Thus transverse
faults reduce the decoder advantage continuously rather than
invalidating it abruptly.

To connect this decoder axis to physical controls, we also construct a
perturbative map from transverse strain noise, quadrature imbalance,
SATD calibration errors, hyperfine transverse components,
transverse-field misalignment, and HBAR filtering residuals to an
effective $p_{XY}$ (Appendix~\ref{app:qec_montecarlo}, Sec.~6).  This map
defines calibration targets for keeping the channel in the $d=9$,
$d=11$, or $d=13$ decoder envelope.

\begin{figure}[t]
  \centering
  \includegraphics[width=\columnwidth]{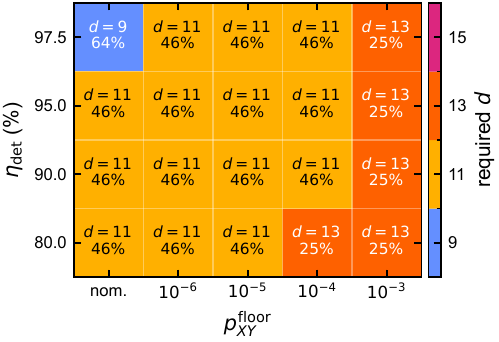}
  \caption{\label{fig:transverse_floor_envelope}%
  Transverse-floor validation envelope for the extracted Regime-A
  biased-erasure channel.
  Required XZZX distance and corresponding data-qubit saving relative
  to the $d=15$ Rabi/CSS baseline are shown as functions of an imposed
  transverse floor $p_{XY}^{\rm floor}$ and erasure-detection
  efficiency $\eta_{\rm det}$.
  The nominal extracted channel gives the $d=9$ / $64\%$ proxy, finite
  floors through $10^{-4}$ remain in the $d=11$ / $46.2\%$ envelope,
  and the $10^{-3}$ stress case moves to $d=13$ / $24.9\%$.}
\end{figure}

We next embedded the same extracted channel into a scheduled
two-sector XZZX detector-model stress diagnostic.  The diagnostic uses
repeated syndrome rounds, dynamic erasure-aware weights,
measurement/reset faults, explicit $Z$- and $X/Y$-sector matching
graphs, leakage persistence, delayed erasure flags, finite
erasure-detection efficiency, and local crosstalk.
We use ``validation envelope'' to mean stress cases that preserve,
reduce, or break the nominal $d=9$ proxy without interpreting the
nonmonotonic adversarial slices as calibrated thresholds.
In the 50k-shot nominal cleanup, the erasure-aware decoder remains at
the $10^{-3}$ logical-failure proxy scale at $d=9$,
$p_L=1.00\times10^{-3}$ with 95\% CI
$[0.759,1.318]\times10^{-3}$, while the no-flag decoder gives
$2.2\times10^{-3}$.
An explicit $p_{XY}=10^{-3}$ floor moves the proxy from $d=9$ to
$d=11$, and combined $\eta_{\rm det}=0.9$ plus finite-$p_{XY}$ stress
remains in the $d=11/d=13$ regime.
The only stress case that is not recovered by $d=15$ is
$2\times$ local crosstalk, $p_L=2.6\times10^{-3}$ with 95\% CI
$[1.77,3.81]\times10^{-3}$, identifying crosstalk as the leading
hardware-specific validation target rather than a weakness of the
extracted biased-erasure channel.

Rectangular XZZX planar codes, which exploit the extreme
Z-bias by using a short row distance $d_r$
(protecting the rare $X$-error direction) and a long column
distance $d_c$ (protecting the dominant $Z$-error direction),
show no observed failures in 5000 code-capacity trials for
all tested $d_r < d_c$ configurations (including the
$3 \times 7$ code with 21~data qubits) up to $100\times$
the physical noise ($p_L < 7.7 \times 10^{-4}$ at $95\%$~CL;
Appendix~\ref{app:qec_montecarlo}).
This zero-failure observation only bounds $p_L$ at the
$10^{-4}$ scale and is not extrapolated to $10^{-10}$.
Because these are finite-statistics code-capacity results, we present
the rectangular codes as a promising extension rather than the
primary overhead estimate.

\begin{table}[t]
\caption{\label{tab:qec_comparison}Noise channel decomposition
and code-capacity overhead estimate for four gate/code configurations.
Physical error rates are from the simulated results of this
work (Table~\ref{tab:errorbudget} and Sweep~2).
The XZZX bias-aware analysis resolves the non-erasure channel
into its Z and depolarizing components, exploiting the
symmetry-resolved noise hierarchy of
Proposition~\ref{prop:noise_sectors}.}
\begin{ruledtabular}
\begin{tabular}{lcccc}
  & \textbf{Rabi} & \textbf{Hol.} & \textbf{Hol.} & \textbf{Hol.} \\
  & \textbf{CSS} & \textbf{era-CSS} & \textbf{XZZX} &
  \textbf{XZZX} \\
  & & & \textbf{square} & \textbf{rect.}\footnote{No failures
  observed in 5000 code-capacity trials up to $100\times$ noise,
  which bounds $p_L$ only at the $10^{-4}$ scale; presented as an
  extension, not an overhead claim.} \\
  \hline
  $p_{\mathrm{era}}$ & $0$ & $0.47\%$ & $0.47\%$ & $0.47\%$ \\
  $p_Z$ & --- & --- & $0.168\%$ & $0.168\%$ \\
  $p_{\mathrm{dep}}$ & $0.63\%$ & $0.18\%$\footnote{Unresolved
  $p_{\mathrm{undet}}$; treated as isotropic depolarizing.} &
  $0.012\%$ & $0.012\%$ \\
  \hline
  Code & CSS & CSS & XZZX & XZZX \\
  Distance & $15$ & $11$ & $9$ & $3\!\times\!7$ \\
  Data qubits & $225$ & $121$ & $81$ & $21$ \\
  Saving vs.\ Rabi & --- & $46\%$ & $\mathbf{64\%}$ &
  ---\footnote{Consistent with ${\sim}91\%$ reduction in the
  finite-statistics code-capacity data, but not extrapolated or used
  as a headline overhead estimate.} \\
\end{tabular}
\end{ruledtabular}
\end{table}

\begin{figure*}[t]
  \centering
  \includegraphics[width=\textwidth]{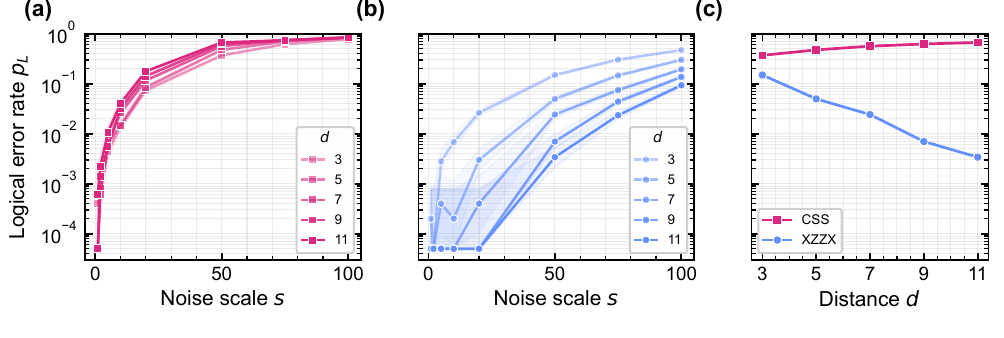}
  \caption{\label{fig:qec_threshold}%
  Decoder-level threshold separation under the extracted
  biased-erasure noise model.
  \textbf{(a)}~Logical error rate $p_L$ versus noise scale~$s$
  for CSS toric codes at distances $d = 3$--$11$.
  Lines fan upward with increasing~$d$ for $s \gtrsim 2$,
  indicating above-threshold operation: larger codes perform
  \emph{worse}.
  \textbf{(b)}~Same plot for XZZX toric codes.
  Lines fan downward with increasing~$d$ at every tested scale
  including $s = 100$, indicating deeply sub-threshold operation:
  larger codes provide exponential error suppression.
  \textbf{(c)}~$p_L$ versus $d$ at fixed $s = 50$.
  CSS (squares) increases with $d$ while XZZX (circles) decreases,
  yielding a $197\times$ advantage at $d = 11$.
  All data points: 5000 Monte Carlo trials with minimum-weight
  perfect matching (MWPM) decoding
  (Appendix~\ref{app:qec_montecarlo}).
  The conservative erasure-only baseline in Appendix~\ref{app:qec_montecarlo}
  is superseded by this bias-aware analysis.}
\end{figure*}

\section{Discussion}
\label{sec:discussion}

\subsection{Cross-platform $T_1/\tgate$ scaling}
\label{sec:scaling}

The NV spin-qubit ($T_1 = \SI{1}{\milli\second}$, $\tgate =
\SI{1.833}{\micro\second}$, $T_1/\tgate \approx 546$) validates the
full SATD architecture, because the required lower-manifold correction
is supplied by the resonant DQ strain channel.  The SiV orbital-qubit benchmark
($T_1 = \SI{143}{\nano\second}$, $\tgate = \SI{46.2}{\nano\second}$,
$T_1/\tgate \approx 3.1$) is the Regime~C lower-doublet-free
$\Lambda$-control benchmark.

Within their respective control assumptions, both platforms follow the
scaling
$\varepsilon_{\mathrm{dissipative}} \propto \tgate/T_1$
from the $T_1/T_{1\rho}$ Lindbladian error channel
(Sec.~\ref{sec:errorbudget}),
whereby the fidelity ceiling is determined entirely by the
$T_1/\tgate$ ratio, not by the absolute timescale.
For Regime~A this scaling appears after NGQC + SATD remove the
non-adiabatic and singularity errors; for Regime~C it appears after
the $\Lambda$-leg-only gate time is lengthened enough to recover
adiabaticity.

Within a fixed control family and selected commensurability branch,
$\tgate$ scales approximately as $1/\Omm$, and
$\varepsilon_{\mathrm{process}} \approx \varepsilon_{\mathrm{coh}} +
c\,\tgate/T_1$ (Eq.~\eqref{eq:error_leg}), where
$\varepsilon_{\mathrm{coh}} \approx 0.003\%$ is the coherent floor.
The mechanical Rabi frequency entering this ratio scales
linearly with the strain coupling constant
($\Omm = g\,\varepsilon_0$) and thus with the material-specific
parameter $g \in \{h_{26}^{\mathrm{NV}},\,f_\perp^{\mathrm{SiV}}\}$.
Any platform satisfying the Theorem~\ref{thm:synrot} selection rule
gains immediate access to this scaling via synthetic rotation,
provided that at least two degenerate strain modes are simultaneously
driven in quadrature.

The SiV single-shot benchmark realizes the second route: a shortcut
construction using only the two orbit-strain $\Lambda$ legs and scalar
detuning.
At $\Omega_{\mathrm{peak}}=\SI{300}{\mega\hertz}$, it reaches
$F_{\mathrm{eff}}>99.5\%$ across the four-gate suite at
$\tgate=\SI{6.667}{\nano\second}$, with the remaining error set by
orbital-$T_1$ depolarization during auxiliary occupation.
Thus SiV is not the preferred platform for the biased-erasure SATD
channel, but it is a natural platform for fast non-Abelian
single-shot control.

\subsection{$C_{3v}$ platform-selection rule}
\label{sec:c3v_platform_rule}

The same $C_{3v}$ spin-strain tensor structure also suggests a
platform-selection rule.  At fixed target $\Omm$, the off-resonant DQ
Stark scale obeys
$\delta_{\mathrm{AC}}\propto
|\honesix/\htwosix|^2\Omm^2/D$
[Eq.~\eqref{eq:c3v_stark_scaling}].  For the assumed 3C-SiC neutral
divacancy parameters, this reduces the Stark scale from
\SI{20.718}{\kilo\hertz} in NV to \SI{0.521}{\kilo\hertz}, a factor
of $39.76$.  In a controlled Regime-A comparison using the same
SATD echo-lune trajectory, surface-bath traces, $T_1$, $T_{1\rho}$,
  and erasure bookkeeping as the NV channel benchmark, uncompensated 3C-SiC
lies essentially on the compensated baseline, while uncompensated NV
incurs a visible Stark penalty.  This supports the design rule that
$C_{3v}$ hosts with small $|\honesix/\htwosix|$ can preserve the
$\Lambda$-sector symmetry-to-channel mechanism while suppressing the
parasitic DQ Stark channel before active compensation.

This comparison is a controlled Hamiltonian comparison, not a measured
3C-SiC device forecast.  A deployable claim requires
platform-specific surface noise, relaxation, strain limits, mechanical
mode shapes, optomechanical coupling, and auxiliary-state detection.

\subsection{Toward universal SU(2) control}
\label{sec:universality}

The merged architecture separates universal control from QEC channel
engineering.  The SATD echo-lune compiler is optimized for the
Regime-A biased-erasure channel, while the single-shot compiler gives
the simpler universal layer.  In the latter, the rotation axis
$\mathbf n$ is set by the relative amplitude and phase of the two
$\Lambda$ legs, and the rotation angle
$\gamma=\pi(1+\sin\alpha)$ is set by the detuning ratio.  Thus a
single cyclic bright-state pulse implements
$U_L(\mathbf n,\gamma)\doteq
\exp[-i\gamma\,\mathbf n\cdot\boldsymbol\sigma/2]$.
Changing $\mathbf n$ changes the bright-state projector, so
nonparallel pulses generate noncommuting holonomies natively rather
than through a compiled sequence of adiabatic loops.

Optical~\cite{sekiguchi2017} and
microwave~\cite{zu2014_holonomic_mw} holonomic-gate demonstrations on
NV centers have established the viability of geometric-phase control;
the present architecture replaces the photonic/MW drive with the
symmetry-generated strain $\Lambda$ manifold and supplies two
compilers on that manifold.

This universal $\mathrm{SU}(2)$ control is currently restricted to a
single subsystem per mode pair.
Multi-qubit entangling gates require phonon-mediated coupling
between distinct defect centers, which is beyond the present
single-qubit framework.
The same bright-state projector also gives a compact entangling
extension at the effective $\Lambda$ level.  Consider two
$\Lambda$ manifolds coupled to one detuned mechanical mode through
their logical bright projectors,
\begin{equation}
K_I(t)=
\left(f_1P_{b1}+f_2P_{b2}\right)
\left(ae^{-i2\pi\delta t}+a^\dagger e^{i2\pi\delta t}\right),
\end{equation}
where $K=H/h$.  After one closed bus loop, $T=1/\delta$, the
displacement vanishes and the remaining geometric phase is
\begin{equation}
U(T)=\exp\!\left[
i\frac{2\pi}{\delta^2}
\left(f_1P_{b1}+f_2P_{b2}\right)^2
\right].
\end{equation}
Since $P_b^2=P_b$, the nonlocal phase is
\begin{equation}
U_{\rm ent}=
\exp\!\left[
i\frac{4\pi f_1f_2}{\delta^2}P_{b1}P_{b2}
\right].
\end{equation}
For $f_1=f_2=\delta/2$, this gives a $\pi$ phase on the joint bright
state.  Choosing the bright axes along logical $Z$ makes the gate
locally equivalent to CZ, up to single-qubit $Z$ phases.
This is an effective projector-force identity, not yet a microscopic
actuator-level theorem.  The corresponding validation is therefore
treated as Regime~F and reported as an architecture diagnostic.

Regime~F tests this extension in a two-qutrit $\Lambda$ model coupled
to a truncated oscillator.  At $f=\SI{2.0}{\mega\hertz}$,
$\delta=\SI{4.0}{\mega\hertz}$, $Q=10^7$,
$T=\SI{0.25}{\micro\second}$, and $n_{\rm Fock}=16$, the extracted
CZ-class channel gives conditional $\Favg=99.8096\%$,
$p_{\mathrm{era}}=0.04874\%$, $p_Z=0.23775\%$,
$p_{ZZ}=0.000229\%$, and no measurable $XY$ component.  The result is
bias-preserving and biased-erasure compatible, but not
erasure-dominated, since $p_Z/p_{\mathrm{era}}=4.88$.  A repeated-syndrome
XZZX proxy using the combined one- and two-qubit channels gives no
logical failures at the base scale for distances $3,5,7$ over
$400$ shots; full data and cross-checks are in
Appendix~\ref{app:2q_projector_extension}.

\subsection{Symmetry-to-decoder design principle}
\label{sec:symmetry_ft}

The results of this work reveal a design principle:
crystallographic symmetry can be promoted from a
control-enabling constraint to a symmetry-to-decoder design principle
whose influence propagates from the microscopic Hamiltonian
to code-capacity logical-code overhead estimates.

The chain has four links, each shaped by the
two-dimensional irreducible representation $\Gamma_E$:
\begin{enumerate}
  \item \textbf{Gate mechanism.}
  The dot-product coupling mandated by $\Gamma_E \otimes
  \Gamma_E \supset \Gamma_{A_1}$
  (Theorem~\ref{thm:synrot}) enables synthetic rotation
  and thereby holonomic control.
  \item \textbf{Noise classification.}
  The $\Lambda$-system bright manifold created by the same
  coupling produces a parity filter that classifies noise
  errors by irrep: $A_1 \to 0$, $A_2 \to O(\sigma^2)$,
  $E \to O(\sigma^3)$
  (Proposition~\ref{prop:noise_sectors}).
  \item \textbf{Conditional biased-erasure channel.}
  The parity-filter duality
  (Remark~\ref{rem:duality}) directs $A_2$-sector
  perturbations toward the heraldable auxiliary level, while
  the $\pi$-echo suppresses $E$-sector errors.
  For weak noise with comparable irrep-sector bath weights and
  efficient $|0\rangle$ detection, this yields the biased-erasure
  channel
  (Corollary~\ref{cor:biased_erasure}).
  \item \textbf{Decoder estimates and stress diagnostic.}
  In code-capacity simulations of the nominal extracted channel, the
  resulting biased-erasure noise
  model places CSS and XZZX codes on opposite sides of the threshold
  boundary, producing a $64\%$ fit-extrapolated benchmark overhead
  estimate whose finite-transverse-floor and erasure-detection
  sensitivity is quantified in Fig.~\ref{fig:transverse_floor_envelope}
  and Appendix~\ref{app:qec_montecarlo}, Sec.~6; a scheduled two-sector detector-model stress
  diagnostic then tests the first circuit layer and identifies strong
  local crosstalk as the leading architecture-specific validation target
  (Sec.~\ref{sec:qec_advantage}).
\end{enumerate}
A common $\Gamma_E$ symmetry provides the control structure
(link~1) and the noise hierarchy (links~2--3); the full
estimated overhead advantage then emerges from co-design of encoding,
leakage-detection mechanism, holonomic control, and decoder
(link~4), and its numerical value remains tied to the simulated
channel parameters and validation level used here.
In this chain, links~2--4 refer specifically to the SATD Regime-A NV
channel.  Regime~D establishes universal control on the same
$\Lambda$ manifold, but it is not the source of the XZZX overhead
number.
This perspective suggests a broader principle: for
solid-state platforms with sufficiently rich point-group
symmetry, co-design across the encoding, gate, and decoder
layers may extract decoder advantages that are
inaccessible to any single layer optimized in isolation; establishing
the quantitative overhead in a hardware-calibrated architecture-level
detector-error model remains future work.

An important aspect of this chain is the role of encoding.
Under Rabi driving in the $\{|0\rangle, |{+1}\rangle\}$
subspace, $T_1$ relaxation is an in-subspace error
(non-heralded), whereas in the holonomic
$\{|{-1}\rangle, |{+1}\rangle\}$ encoding, the same $T_1$
decay populates $|0\rangle$ outside the computational
subspace, converting it to a heralded erasure.
Combined with the parity-filter noise shaping and the
$\pi$-echo, the joint choice of encoding and gate protocol
reshapes the effective noise class on the \emph{same physical
hardware} from a conservatively treated unstructured
dynamical-gate baseline to an extracted
biased-erasure channel.  For the nominal extracted channel this gives
the $64\%$ code-capacity overhead estimate of
Table~\ref{tab:qec_comparison}; finite bit-flip floors and imperfect
erasure detection reduce this advantage as quantified in the
supplement.

\subsection{Beyond defect centers}
\label{sec:beyond_defects}

The defect platforms considered here are experimentally mature
realizations of a more general symmetry criterion.  The
synthetic-rotation construction requires neither a vacancy nor an
impurity in an essential way; it requires a strain-active
$\Lambda$ manifold whose projected strain and transition operators
share a two-dimensional irrep of the effective device symmetry.
For point defects this is the crystallographic defect point group,
whereas for acceptor, valley-orbital, hole-spin, or gate-defined
manifolds it is the little group left after crystal symmetry,
confinement, device geometry, static fields, and strain-mode symmetry
are all imposed.  The generalization is conditional: platforms with
only one-dimensional irreps, strong symmetry breaking, or no
phase-locked degenerate strain-mode pair may still allow engineered
two-axis mechanical control, but they do not inherit the
symmetry-protected circular-strain selection rule.

\subsection{Minimal experimental tests}
\label{sec:minimal_tests}

A decisive first demonstration need not implement the full
fault-tolerance stack.  The minimal experimental program is:
\begin{enumerate}
  \item demonstrate two phase-locked mechanical modes with calibrated
  quadrature phase at the defect site;
  \item measure circular-strain selectivity between the two
  $\Lambda$ legs, including handedness reversal;
  \item verify that symmetry-preserving perturbations primarily
  renormalize the scalar coupling rather than mixing handedness;
  \item implement one holonomic loop and compare it with a linearly
  driven mechanical baseline;
  \item detect leakage into the auxiliary state as a heralded erasure
  channel.
\end{enumerate}
Passing these tests would establish the central quantum-information
claim:
crystallographic symmetry has generated not only a control field, but
a $\Lambda$-control resource with observable gate and error-channel
consequences.

\subsection{Experimental outlook}
\label{sec:outlook}

The NV membrane variant requires a $(100)$-oriented diamond plate
($L \sim \SI{10}{\micro\meter}$, $h_m \sim \SI{200}{\nano\meter}$),
two piezoelectric transducers for quadrature driving,
and DC electrodes for electrostatic degeneracy
tuning~\cite{barfuss2019_prb}; all components have been
individually demonstrated~\cite{burek2014_nanoletter,barfuss2015_natphys}.
The GHz HBAR variant retains the same quadrature topology
(Appendix~\ref{app:ghz}), but its finite acoustic bandwidth means that
the fast Regime-A envelope cannot be assumed without predistortion,
reduced effective~$Q$, or a slower HBAR-specific waveform.
The SiV orbital platform requires
$\Omm \geq \SI{300}{\mega\hertz}$ at millikelvin temperatures,
which is within reach of recent
SiV--resonator demonstrations~\cite{cornell2025}.
In situ calibration of $\Omm$ is achievable via strain-modulation
spectroscopy~\cite{macquarrie2013_prl,barfuss2015_natphys}, with
the ideal bright-state spacing $\Delta T = 1/\Omm$ providing a secondary
calibration signature.
A step-by-step experimental protocol is given in
Appendix~\ref{app:protocol}.
The 3C-SiC neutral-divacancy comparison in
Sec.~\ref{sec:c3v_platform_rule} identifies a complementary
platform-design direction within the same $C_{3v}$ class.

\section{Conclusion}
\label{sec:conclusion}

We have introduced crystallographic symmetry as a quantum-information
resource for phononic solid-state processors.
The central result is a shared-irrep selection rule: when the
projected strain tensor and transition operators of a strain-active
$\Lambda$ manifold share a multiplicity-one two-dimensional irrep,
symmetry fixes the linear strain interaction to a scalar dot product.
Two degenerate mechanical modes driven in quadrature then synthesize a
complex circular $\Lambda$ drive.  Defect centers with $C_{3v}$ and
$D_{3d}$ symmetry provide experimentally developed realizations, but
the criterion itself is representation-theoretic and applies to any
effective device manifold satisfying the same shared-irrep condition.

For NV centers, the full composite NGQC + SATD protocol is modeled in
Regime~A with a resonant double-quantum strain tone realizing the
exact counter-diabatic operators.
Open-system simulations give a Regime~A conditional
$\Favg = 99.88\%$ at $\tgate = \SI{1.833}{\micro\second}$, with
$\Favg^{\mathrm{eff}}=99.40\%$ after leakage is counted as error.
Regime~B gives the corresponding GHz-HBAR projection at
resonant-drive Rabi frequencies and treats envelope filtering as an
implementation constraint rather than as part of the Regime-A erasure
extraction.
For SiV centers, Regime~C gives a conservative $\Lambda$-leg-only
orbital-control benchmark with $\Favg = 96.32\%$ at
$\SI{46.2}{\nano\second}$.
On the same symmetry-generated $\Lambda$ manifold, Regime~D validates
a single-shot bright-state compiler for compact universal
non-Abelian $\mathrm{SU}(2)$ control, with
$F_{\mathrm{eff}}\simeq99.86\%$ across a representative gate suite at
$\tgate=\SI{0.9009}{\micro\second}$.
Regime~E applies that compiler directly to the SiV orbital
$\Lambda$ manifold using only orbit-strain $\Lambda$ legs and scalar
detuning, reaching
$F_{\mathrm{eff}}=99.52\%$--$99.67\%$ across the four-gate suite at
$\tgate=\SI{6.667}{\nano\second}$ without assuming a lower-doublet
SATD actuator.

The same symmetry that enables synthetic rotation also organizes the
noise response.
In the Regime-A NV channel benchmark, the $\Gamma_E$-enabled $\Lambda$ manifold
produces a perturbative hierarchy in which $A_1$ perturbations are
common-mode, $A_2$ perturbations preferentially generate
auxiliary-state leakage, and $E$-sector bit-flip errors are
echo-suppressed.
Sector-injection diagnostics verify this map directly:
$A_2$ perturbations point toward $|0\rangle$, whereas $E$ perturbations
point toward the logical bright state and require the two-lune echo.
Combined with erasure detection, this yields an extracted
biased-erasure channel.
Code-capacity simulations of the nominal extracted channel give a
substantial XZZX overhead reduction relative to an unstructured
dynamical-gate baseline, with supplemental sensitivity tests showing
how finite $p_{XY}$ floors and imperfect erasure detection reduce
that estimate.
This overhead estimate is tied to the SATD Regime-A channel; the
single-shot suite is the universal-gate validation rather than a
second architecture-level QEC claim.
Likewise, the membrane and HBAR calculations are implementation
routes for the same control Hamiltonian: the membrane topology does
not by itself validate the erasure advantage, and HBAR-specific
residuals enter through the explicit transverse-floor and missed-flag
stress envelopes.

These results support a broader design principle: crystallographic
symmetry can be used not only to enable control, but to co-design the
encoding, gate, noise bias, and decoder.
The matched-environment 3C-SiC comparison turns the same principle into
a platform-selection rule: smaller $|\honesix/\htwosix|$ within the
$C_{3v}$ class suppresses the parasitic DQ Stark burden while
preserving the $\Lambda$-sector mechanism.
Finally, the decoder analysis now spans both the nominal code-capacity
model and a scheduled two-sector detector-model stress diagnostic,
complementing broader erasure-qubit architecture
analyses~\cite{gu2025_erasure_arch}.
The former gives the $64\%$ model-channel saving; the latter shows
that the nominal biased-erasure advantage survives the first circuit
layer and isolates strong local crosstalk as the dominant
hardware-specific validation target for a calibrated
architecture-level detector-error model.
The bright-state-projector structure also points toward
phonon-bus entangling gates; the effective-model analysis in the
Supplement suggests a CZ-class phase-biased extension, but a
hardware-level two-qubit gate requires a separate
$\Lambda$-actuator validation.

\section{Author contributions}

E.M.M. conceived the project, developed the theoretical model, performed the numerical simulations, and wrote the manuscript. K.A. reviewed the manuscript, provided logistical support, and supervised the project. All authors contributed to the interpretation of the results and the finalization of the manuscript.

\begin{acknowledgments}
We thank Eikichi Kimura and Masato Koga for insightful discussions and valuable feedback that helped improve this work. This research was supported by JSPS KAKENHI Grant Number 24K21730.
\end{acknowledgments}

\section*{Data availability}

All simulation code (Python modules), raw sweep data (JSON), and
figure-generation scripts are deposited in a public repository
at \url{https://github.com/E-zClap/PhononQ}.
The repository includes software requirements in
\texttt{requirements.txt} and execution instructions for reproducing
the ten core parametric sweeps (Sweeps~1--10).

\clearpage
\appendix
\section*{Supplemental Material}

\setcounter{section}{0}
\setcounter{figure}{0}
\setcounter{table}{0}
\setcounter{equation}{0}
\renewcommand{\thefigure}{S\arabic{figure}}
\renewcommand{\thetable}{S\arabic{table}}
\renewcommand{\theHsection}{SM.\arabic{section}}
\renewcommand{\theHfigure}{SM.\arabic{figure}}
\renewcommand{\theHtable}{SM.\arabic{table}}
\renewcommand{\theHequation}{SM.\thesection.\arabic{equation}}

\section{Biharmonic mode analysis}
\label{app:modes}

The flexural dynamics of a thin isotropic plate are governed by the
Kirchhoff--Love equation $\Dbend\,\nabla^4 w + \rho h_m\,\ddot{w} = 0$
with flexural rigidity
\begin{equation}
  \Dbend = \frac{Y\,h_m^3}{12(1-\nu^2)}
  \approx \SI{7.08e-10}{\joule},
\end{equation}
where all parameters are as defined in Sec.~IV\,C of the main text
and main-text Fig.~1.
Simply-supported eigenmodes
$w_{nm} \propto \sin(n\pi x/L)\sin(m\pi y/L)$ have eigenfrequencies
$f_{nm} = (\pi/2L^2)\sqrt{\Dbend/\rho h_m}\,(n^2+m^2)$.
For $(1{,}2)/(2{,}1)$: $n^2+m^2 = 5 \Rightarrow f_0 = \SI{78.8}{\mega\hertz}$.

The lumped parameters are:
$m_{\mathrm{eff}} = \rho L^2 h_m/4 = \SI{17.6}{\pico\gram}$,
$\kmech^{\mathrm{SS}} = m_{\mathrm{eff}}\,\omega_0^2 = \SI{4310}{\newton\per\meter}$,
$\kmech^{\mathrm{clamp}} \approx \SI{9530}{\newton\per\meter}$
(stiffness ratio $2.21$ from Leissa's tabulated eigenvalues).

\section{Strain topology at candidate NV sites}
\label{app:topology}

At the off-center site $(L/4, L/4)$:
\begin{align}
  \frac{\partial w_{12}}{\partial x}\bigg|_{L/4}
  &= +\frac{\pi}{\sqrt{2}\,L}, \qquad
  \frac{\partial w_{12}}{\partial y}\bigg|_{L/4} = 0, \\
  \frac{\partial w_{21}}{\partial x}\bigg|_{L/4}
  &= 0, \qquad
  \frac{\partial w_{21}}{\partial y}\bigg|_{L/4}
  = +\frac{\pi}{\sqrt{2}\,L}.
\end{align}
Each mode selects one orthogonal strain-channel direction in the
topology model, with equal magnitude.
The in-plane correction introduces a minor ellipticity
$\eta = (3/2)(|\htwofive|/|\htwosix|)(\pi|z|/\kappa_s L) \approx 4.3\%$,
contributing $\mathcal{O}(\eta^2) \lesssim 2\times10^{-3}$ to the
fidelity error.
This thin-plate result is used as a strain-topology guide; the actual
local Rabi rate must be calibrated experimentally or supplied by the
resonant HBAR route.

\section{Electrostatic tuning law}
\label{app:tuning}

For an asymmetric membrane with $L_y = L_0(1+\delta)$, Taylor
expansion of the biharmonic eigenvalues gives
\begin{equation}
  \frac{\Delta f}{f_0} = 0.6\,\delta.
  \label{eq:splitting_deriv}
\end{equation}
For $\delta = 1\%$: $\Delta f = \SI{0.47}{\mega\hertz}$.

The electrostatic spring-softening stiffness is
$k_{\mathrm{el}} = -\varepsilon_0 A_{\mathrm{eff}} V^2/d^3$
(here $\varepsilon_0$ denotes the vacuum permittivity).
The required voltage evaluates to $V_{\mathrm{req}} = \SI{21.6}{\volt}$
for a full-area electrode, safely below
$V_{\mathrm{bd}} = \SI{35}{\volt}$, a conservative design voltage
rather than a geometry-calibrated breakdown threshold; it is informed
by nanoscale vacuum-gap breakdown and MEMS field-emission
studies~\cite{wang2022_vacuum_breakdown,michalas2012_mems_field_emission}.
A $50\%$-coverage electrode requires $\SI{30.5}{\volt}$, eliminating
safety margin; full-area coverage is therefore mandatory.
For gigahertz high-overtone bulk acoustic resonator (GHz-HBAR)
devices (Appendix~\ref{app:ghz}), residual
frequency mismatch between the mechanical overtone and the NV
zero-field splitting $D$ can be fine-tuned in situ via the known
temperature dependence $dD/dT \approx -\SI{74}{\kilo\hertz\per\kelvin}$,
providing sub-kHz resolution at millikelvin
stability, complementary to the coarse RIE adjustment of the
HBAR thickness.

\section{Derivative-removal-by-adiabatic-gate (DRAG) and superadiabatic transitionless driving (SATD) Hamiltonian derivation}
\label{app:drag}

The adiabatic dark and bright eigenstates of the RWA Hamiltonian are
\begin{align}
  |D\rangle &= \cos(\theta/2)|{-1}\rangle + \sin(\theta/2)e^{i\phi}|{+1}\rangle, \\
  |B\rangle &= \sin(\theta/2)|{-1}\rangle - \cos(\theta/2)e^{i\phi}|{+1}\rangle.
\end{align}
The non-adiabatic coupling $\langle B|\partial_t|D\rangle = -\dot{\theta}/2$
drives leakage.
The exact counter-diabatic Hamiltonian canceling this coupling
is~\cite{berry2009_transitionless}
\begin{equation}
  H_{\mathrm{CD}} = i\lambda\,\frac{\dot{\theta}}{2}
  (|D\rangle\langle B| - |B\rangle\langle D|),
\end{equation}
which in the bare basis gives the counter-diabatic operator used in
Eqs.~(4) and~(6) of the main text.
For $\lambda = 1$ (exact SATD), this perfectly cancels leakage at all
speeds; for the heuristic DRAG variant, $\lambda$ is scanned
empirically.

\section{Hardware decomposition of the SATD correction}
\label{app:satd_hardware}

The counter-diabatic Hamiltonian used in the simulations is not a
third abstract control axis.  For the NV implementation it is
realized by a resonant double-quantum (DQ) strain tone on the
$|{-1}\rangle\leftrightarrow|{+1}\rangle$ transition.
This section gives the laboratory decomposition and separates this
resonant SATD use of the DQ tensor from the off-resonant DQ Stark
shift derived in Appendix~\ref{app:dqstark}.

In angular-frequency units the SATD term is
\begin{equation}
  \frac{H_{\mathrm{CD}}}{\hbar}
  =
  -\frac{\dot{\theta}}{2}
  \left[
    \sin\phi\,\mathrm{Op}_{\mathrm{re}}
    + \cos\phi\,\mathrm{Op}_{\mathrm{im}}
  \right],
  \label{eq:cd_hw_start}
\end{equation}
with
$\mathrm{Op}_{\mathrm{re}} =
|{-1}\rangle\langle{+1}| + |{+1}\rangle\langle{-1}|$
and
$\mathrm{Op}_{\mathrm{im}} =
i(|{+1}\rangle\langle{-1}|-|{-1}\rangle\langle{+1}|)$.
In the rotating frame of the resonant DQ transition, a phase-controlled
DQ strain tone has the form
\begin{equation}
  \frac{H_{\mathrm{DQ}}^{\mathrm{res}}}{\hbar}
  =
  \frac{1}{2}
  \left[
    \Omega_{\mathrm{re}}(t)\,\mathrm{Op}_{\mathrm{re}}
    + \Omega_{\mathrm{im}}(t)\,\mathrm{Op}_{\mathrm{im}}
  \right],
  \label{eq:dq_res_operator}
\end{equation}
where $\Omega_{\mathrm{re}}$ and $\Omega_{\mathrm{im}}$ are angular
Rabi rates set by the two quadratures of the resonant DQ drive.
Exact Hamiltonian matching to Eq.~\eqref{eq:cd_hw_start} is therefore
obtained with
\begin{equation}
  \Omega_{\mathrm{re}}(t)=-\dot{\theta}(t)\sin\phi(t),
  \qquad
  \Omega_{\mathrm{im}}(t)=-\dot{\theta}(t)\cos\phi(t).
  \label{eq:dq_satd_quadratures}
\end{equation}
Equivalently, the lower-manifold matrix element is
\begin{equation}
  \frac{H_{\mathrm{CD}}}{\hbar}
  =
  \frac{i\dot{\theta}}{2}e^{i\phi}
  |{-1}\rangle\langle{+1}| + \mathrm{h.c.}
  \label{eq:cd_complex_element}
\end{equation}
Thus a resonant DQ drive with angular Rabi frequency
\begin{equation}
  \Omega_{\mathrm{DQ,CD}}(t)=|\dot{\theta}(t)|,
  \qquad
  \varphi_{\mathrm{DQ}}(t)=
  \phi(t)+\frac{\pi}{2}
  +\pi\,\Theta[-\dot{\theta}(t)]
  \label{eq:dq_satd_controls}
\end{equation}
implements exactly the two operators included in the numerical
Hamiltonian.  In ordinary frequency units,
$\Omega_{\mathrm{DQ,CD}}/(2\pi)=|\dot{\theta}|/(2\pi)$.  Equivalently,
the coefficient multiplying each Pauli-like operator in
Eq.~\eqref{eq:dq_res_operator} is smaller by a factor of two; this is
the origin of the two common quoted cycle-frequency scales for the same
drive.

For NV$^-$ the carrier frequency is the energy splitting within the
computational doublet,
\begin{equation}
  f_{\mathrm{DQ}} = 2\gamma_e B_z
  \simeq \SI{280}{\mega\hertz}
  \quad (B_z=\SI{50}{\gauss}),
  \label{eq:dq_satd_carrier}
\end{equation}
not the single-quantum frequencies
$D_\pm = D\pm\gamma_e B_z$ used for the two $\Lambda$ legs.
The same $C_{3v}$ DQ spin-strain tensor
$\honesix(S_x^2-S_y^2,\,2\{S_x,S_y\})$ therefore supplies the
SATD control when driven resonantly at~$f_{\mathrm{DQ}}$.
For the optimal NV gate time
$\tgate=\SI{1.833}{\micro\second}$ and the two-loop trajectory
$\theta(t)=\pi\sin(\pi t/\tau)$ with $\tau=\tgate/2$,
\begin{equation}
  \max|\dot{\theta}| = \frac{2\pi^2}{\tgate},
  \qquad
  \frac{\Omega_{\mathrm{DQ,CD}}^{\max}}{2\pi}
  = \frac{\pi}{\tgate}
  \simeq \SI{1.71}{\mega\hertz}.
  \label{eq:dq_satd_amplitude}
\end{equation}
Thus the full resonant DQ Rabi rate is $\SI{1.71}{\mega\hertz}$,
while the Hamiltonian coefficient in Eq.~\eqref{eq:dq_res_operator}
is $\Omega_{\mathrm{DQ,CD}}^{\max}/(4\pi)\simeq
\SI{0.86}{\mega\hertz}$ in cycle-frequency units.
Using $\honesix=\SI{19.66}{\giga\hertz\per\strain}$, the corresponding
peak additional DQ strain required for the SATD tone is
\begin{equation}
  \varepsilon_{\mathrm{DQ,CD}}^{\max}
  =
  \frac{\Omega_{\mathrm{DQ,CD}}^{\max}/2\pi}{|\honesix|}
  \simeq 8.7\times10^{-5},
  \label{eq:dq_satd_strain}
\end{equation}
or $4.4\times10^{-5}$ if one quotes the half-rate Hamiltonian
coefficient rather than the full DQ Rabi rate.  Both conventions are
well below the main single-quantum strain amplitude used for the NV
benchmark.

For the assumed 3C-SiC neutral-divacancy parameters used in the
matched $C_{3v}$ comparison, the same SATD condition requires
\begin{equation}
  \varepsilon_{\mathrm{DQ,CD}}^{\max}(\mathrm{3C\mbox{-}SiC})
  =
  \frac{\pi/\tgate}{|h_{16}^{\mathrm{3C\mbox{-}SiC}}|}
  \simeq 1.27\times10^{-3}.
\end{equation}
Thus 3C-SiC suppresses the parasitic off-resonant Stark channel but
requires a larger intentional resonant SATD strain than NV.  This is
still below a percent-level strain scale in the parameter check, but
it requires a dedicated SiC mechanical actuator design.

The bandwidth requirement is set by the SATD envelope, not by
mechanical ring-down.  The smooth part of the DQ envelope varies on
the sub-loop timescale $\tau=\tgate/2$, giving an envelope bandwidth
of order $1/\tau\simeq\SI{1.1}{\mega\hertz}$.  The programmed phase
step at the South Pole occurs where $\dot{\theta}=0$, so it does not
require a large-amplitude broadband burst; experimentally it can be
implemented as a phase update through the zero of the DQ envelope or
by standard RF pre-distortion of the transducer response.

This resonant DQ SATD tone is distinct from the off-resonant DQ Stark
shift.  The Stark term arises because the two single-quantum
$\Lambda$-leg drives at $D_\pm$ also couple to the DQ tensor but are
detuned from the $|{-1}\rangle\leftrightarrow|{+1}\rangle$
transition by approximately~$D$.  Its effect is therefore
second-order and dispersive.  In contrast, the SATD tone is a
deliberate first-order resonant drive at~$2\gamma_e B_z$ with a
prescribed amplitude and phase; its Hamiltonian is exactly the
operator pair used in the simulations.

For SiV$^-$ we do not assume such a lower-doublet SATD control in the
revised benchmark.  The orbit-strain selection rule establishes the
two $\Lambda$ legs between the lower and upper orbital branches, but
the manuscript does not identify a measured strain matrix element
that resonantly couples the two lower-doublet computational states
with the phase and bandwidth required by Eq.~\eqref{eq:dq_satd_controls}.
We therefore use the conservative SiV benchmark with
$\alpha_{\mathrm{CD}}=0$, keeping the same open-system model but
removing the lower-manifold SATD Hamiltonian.  The best millikelvin
benchmark points are
\begin{center}
\begin{tabular}{cccc}
  $\Omm$ (MHz) & $\tgate$ (ns) & $\Favg$ & Leakage \\
  \hline
  100 & 99.23 & 91.18\% & 7.96\% \\
  300 & 46.24 & 96.32\% & 3.85\% \\
  500 & 123.31 & 90.14\% & 8.64\%
\end{tabular}
\end{center}
The previously quoted SiV SATD numbers should therefore be read only
as an ideal lower-manifold-control upper bound unless a physical
SiV lower-doublet SATD channel is supplied.

\section{Decoherence channel inventory}
\label{app:decoherence}

Table~\ref{tab:decoherence} summarizes the complete decoherence
environment.

\begin{table*}[!t]
\caption{\label{tab:decoherence}Decoherence channel inventory.
Channels included in the model (top),
and channels excluded with justification (bottom).}
\begin{ruledtabular}
\begin{tabular}{llll}
  Channel & Mechanism & Status & Note \\
  \hline
  Surface spin hopping & ME-CCE KMC & Included & Dominant for shallow NV \\
  Dipolar field fluctuation & $B_z^{\mathrm{surf}}$ sum & Included & Direct coupling to qubit \\
  Quasi-static $T_2^*$ & Inter-trajectory DC offsets & Included &
    Restored by removing variance normalization \\
  $T_1$ relaxation & Lindbladian collapse & Included & $T_1 = \SI{1}{\milli\second}$ \\
  $T_{1\rho}$ depolarization & Lindbladian collapse & Included & $T_{1\rho} = \SI{500}{\micro\second}$; reduced by enrichment \\
  AC Stark shift ($\Sz$) & Double-quantum tensor & Included &
    Suppressed by NGQC (Theorem~2 of the main text) \\
  \hline
  ${}^{13}$C nuclear bath & Flip-flop dynamics & Excluded & Slow ($\tgate \ll T_2^{\mathrm{C13}}$) \\
  Charge noise & Electric field & Excluded & $< \SI{1}{\kilo\hertz}$ dephasing \\
  Mechanical damping & $1/Q$ loss & Excluded & $Q > 10^4$; stable over $\tgate$ \\
\end{tabular}
\end{ruledtabular}
\end{table*}

\section{Symmetry protection and noise-sector proof}
\label{app:noise_sectors}

This appendix proves the two formal claims used in the main text and
then records the sector-injection validation of the resulting
sector-to-channel map.
Section G.1 proves Theorem~1: the $\Gamma_E$ strain-defect coupling is
a unique scalar invariant, and non-scalar bilinears require explicit
symmetry-breaking spurions. Section G.2 proves Proposition~1 and
Corollary~1: once the protected coupling generates the ideal
$\Lambda$ manifold, weak perturbations are filtered into distinct
$A_1$, $A_2$, and $E$ error channels. Sections G.4--G.5 give the
numerical sector-injection and compact robustness diagnostics.

\subsection*{1. Group-theoretic protection of the synthetic-rotation selection rule}

Let $\boldsymbol{\epsilon}=(\epsilon_1,\epsilon_2)$ and
$\mathbf{O}=(O_1,O_2)$ transform under the same real two-dimensional
irreducible representation $\Gamma_E$ of the defect point group $G$.
For $D_{3d}$ centers, the same statements hold with
$E,A_1,A_2$ replaced by $E_g,A_{1g},A_{2g}$.

The most general linear strain-defect coupling can be written
\[
H_\epsilon=\boldsymbol{\epsilon}^{T}C\mathbf{O}
=
\sum_{i,j=1}^{2}\epsilon_i C_{ij}O_j .
\]
Under a point-group operation $g\in G$,
\[
\boldsymbol{\epsilon}\mapsto R_E(g)\boldsymbol{\epsilon},\qquad
\mathbf{O}\mapsto R_E(g)\mathbf{O}.
\]
Invariance of $H_\epsilon$ requires
\[
R_E(g)^T C R_E(g)=C\qquad \forall g\in G .
\]
Because $\Gamma_E$ is irreducible, Schur's lemma gives
$C=g_0 I_2$. Therefore the only $G$-invariant linear coupling is
\[
H_\epsilon=g_0(\epsilon_1O_1+\epsilon_2O_2).
\]

Equivalently, the character formula
\[
n_\lambda=\frac{1}{|G|}\sum_{g\in G}\chi_E(g)^2\chi_\lambda^*(g)
\]
gives
\[
E\otimes E=A_1\oplus A_2\oplus E
\]
for $C_{3v}$ and
\[
E_g\otimes E_g=A_{1g}\oplus A_{2g}\oplus E_g
\]
for $D_{3d}$. A convenient bilinear basis is
\[
B_{A_1}=\epsilon_1O_1+\epsilon_2O_2,
\]
\[
B_{A_2}=\epsilon_1O_2-\epsilon_2O_1,
\]
and
\[
\mathbf{B}_E=
\left(
\epsilon_1O_1-\epsilon_2O_2,\;
\epsilon_1O_2+\epsilon_2O_1
\right).
\]
Only $B_{A_1}$ is a scalar. The $A_2$ and $E$ bilinears can enter the
Hamiltonian only if the device or environment supplies spurions
$\chi_{A_2}$ or $\boldsymbol{\chi}_E$ transforming in the corresponding
irreps:
\[
H_\epsilon=
g_0 B_{A_1}
+\chi_{A_2}B_{A_2}
+\boldsymbol{\chi}_E\cdot\mathbf{B}_E .
\]
Thus a $G$-preserving perturbation can renormalize $g_0$ but cannot
generate ellipticity, handedness mixing, or anisotropic strain coupling
at first order.

For a quadrature drive,
\[
\epsilon_1(t)=\epsilon_0\cos\omega_d t,\qquad
\epsilon_2(t)=\epsilon_0\sin\omega_d t,
\]
define $O_\pm=O_1\pm iO_2$. Substitution gives
\[
H_\epsilon(t)=\frac{g_0\epsilon_0}{2}
\left(O_+e^{-i\omega_d t}+O_-e^{+i\omega_d t}\right).
\]
When $\omega_d$ is resonant with one arm of the $\Lambda$ system and
the opposite circular component is detuned by $\Delta_{\mathrm{off}}$,
the rotating-wave approximation is valid for
\[
\Omega_m,\;|\dot\theta|,\;|\dot\phi|,\;\sigma,\;|\Delta_m|
\ll \Delta_{\mathrm{off}},\omega_d .
\]
The resonant component then gives the single-arm circular Hamiltonian
used in the main text. This completes the proof of Theorem~1.

\subsection*{2. Scope of the noise-sector statement}

The point-group argument above determines the allowed operator sectors
and the protection of the circular strain selection rule. The following
noise-sector argument determines the perturbative scaling of geometric
errors once the ideal $\Lambda$ manifold has been generated. It does
not determine scalar Lindblad rates. We assume:
(i)~$\sigma/\Omm\ll1$; (ii)~the system operators appearing in the
system--bath coupling decompose into $A_1$, $A_2$, and $E$
irreducible tensor sectors;
(iii)~the bath spectral weights are not parametrically concentrated
in the $E$ sector; (iv)~device asymmetries do not introduce leading
symmetry-breaking Lindblad operators; and (v)~population in
$|0\rangle$ is detected as an erasure with efficiency
$\eta_{\mathrm{det}}$.  Under these assumptions, the symmetry-derived
hierarchy implies the biased-erasure channel used in the main text.
The ME-CCE/KMC bath used in the Regime-A simulation is a more specific
case: it is injected through the system operator
$B_z^{\mathrm{surf}}(t)S_z$, i.e. the $A_2$ sector.  Its erasure rate
is therefore a numerical simulation result, not a prediction of
representation theory alone.

\noindent
\textbf{Supplemental Theorem G1 (Conditional symmetry selection of the
biased-erasure channel).}
Let the controlled $\Lambda$ system be generated by the same
two-dimensional irrep $\Gamma_E$ that appears in the strain tensor,
with bright states
$|B_\pm\rangle=(|B\rangle\pm|0\rangle)/\sqrt{2}$ and
energies $\pm\Omm/2$.  Under the assumptions above, the effective
per-gate channel can be written as
\begin{equation}
  \mathcal{E}(\rho)
  =
  (1-p_{\mathrm{era}})\mathcal{E}_\mathcal{Q}(\rho)
  + p_{\mathrm{era}}\mathcal{E}_{\mathrm{flag}}(\rho),
\end{equation}
where $\mathcal{E}_{\mathrm{flag}}$ is a known-location erasure
channel.  The no-erasure conditional channel has a strongly
$Z$-biased perturbative hierarchy,
\begin{align}
  p_Z &=
  O(\sigma_{A_2}^2/\Omm^2)+p_{T_{1\rho}}+p_{\mathrm{det}},\\
  p_{XY}
  &\lesssim
  C_E(\sigma_E/\Omm)^3+p_{\mathrm{therm}}+p_{\mathrm{ctrl}},
\end{align}
so that comparable system-operator sector strengths give
$p_Z/p_{XY}\sim\Omm/\sigma\gg1$.
Here $p_{\mathrm{therm}}$ denotes thermally activated spin-flip
processes and $p_{\mathrm{ctrl}}$ denotes residual control
imperfections. The group theory fixes the perturbative hierarchy,
while the numerical model fixes the coefficients and additive
device-dependent contributions.

\begin{proof}[Proof of Proposition~1]

\textit{(a) $A_1$ sector} ($\hat{V} = S_z^2$, etc.).---Because
$S_z^2 = I$ on
$\mathcal{Q}=\mathrm{span}\{|{-1}\rangle,|{+1}\rangle\}$,
$\hat{V}$ acts as a common-mode shift and produces no
$|D\rangle$--$|B\rangle$ coupling at any order:
$\delta\gamma_{\mathrm{geo}} = 0$.

\textit{(b) $A_2$ sector} ($\hat{V} = S_z$).---The matrix elements
$\langle B_\pm|S_z|D\rangle$ have the \emph{same} sign,
while the energy denominators $E_D - E_{B_\pm} = \mp\Omm/2$
have opposite sign.
The perturbative weights therefore cancel in the $|B\rangle$
component, yielding $|D^{(1)}\rangle \propto |0\rangle \perp
\mathcal{Q}$.
Since $\partial_\mu|D\rangle \in \mathcal{Q}$ and
$\langle 0|\partial_\mu D\rangle = 0$, the Berry connection
is invariant: $\delta\mathcal{A}_\mu^{(D)} = 0$ at
$O(\sigma)$, and
$\delta\gamma_{\mathrm{geo}} = O(\sigma^2/\Omm^2)$.

\textit{(c) $E$ sector} ($\hat{V} = S_x,\,S_y$).---The matrix
elements
$\langle B_\pm|S_x|D\rangle$ have \emph{opposite} sign,
so the weights reinforce in the $|B\rangle$ component:
$|D^{(1)}\rangle \propto |B\rangle \in \mathcal{Q}$.
The Berry connection is modified at $O(\sigma)$.
However, the first-order correction
$\delta\mathcal{A}_\theta \propto \sin(\theta/2)\sin\phi$
is odd under the $\pi$-echo ($\phi_0 \to \phi_0 + \pi$),
so the line integrals from the two NGQC lunes cancel:
$\oint\delta\mathcal{A}\,d\ell = 0$ at $O(\sigma)$.
At second order the relevant Berry-connection correction
$\delta\mathcal{A}_\theta^{(2)} \propto c^*\partial_\theta c$
contains exclusively $n = \pm 1$ Fourier modes in~$\phi$
(no $n = 0$ component), because
$c^*\partial_\theta c =
\tfrac{1}{4}[\cos^2\!\tfrac{\theta}{2}\,e^{i\phi}
- \sin^2\!\tfrac{\theta}{2}\,e^{-i\phi}]$.
These $n = \pm 1$ modes are odd under $\phi_0 \to \phi_0 + \pi$
and cancel under the echo.
Although $|c|^2 = (1 + \sin\theta\cos\phi)/2$ does carry an
$n = 0$ mode, it enters only through~$\mathcal{A}_\phi$, where
the lune line-integral structure
$[f(\phi_0) - f(\phi_0 + \delta\phi)]$ annihilates the
$\phi$-independent component before the echo acts.
The leading surviving mode $n = \pm 2$ first appears at
$O(\sigma^3)$ from cubic products $|c|^2 c$, giving
$\delta\gamma_{\mathrm{geo}} = O(\sigma^3/\Omm^3)$.
\end{proof}

Supplemental Theorem G1 follows by combining this
sector hierarchy with the assumed erasure detection of $|0\rangle$
population and the usual conditional/no-erasure decomposition of the
gate channel.
The underlying mechanism (destructive interference between
the $|B_+\rangle$ and $|B_-\rangle$ perturbative channels,
which we call the \emph{parity filter}) follows directly
from the symmetric $\pm\Omm/2$ bright-state splitting
of the $\Lambda$ Hamiltonian.
Corollary~1 of the main text follows immediately:
under comparable system-operator sector strengths, the $A_2$-sector
error ($O(\sigma^2)$) dominates the echo-suppressed $E$-sector error
($O(\sigma^3)$), so the Pauli-bias ratio
$\eta \equiv p_Z / p_{XY}$ grows as $(\Omm/\sigma)$.
Together with $|0\rangle$ erasure detection, the resulting conditional
biased-erasure channel is exploitable by
XZZX-type surface codes~\cite{bonilla_ataides2021_xzzx},
which achieve qualitatively higher thresholds than
Calderbank-Shor-Steane (CSS)
codes under biased noise~\cite{tuckett2019_tailored,tuckett2020_faulttolerant}.
The code-capacity quantum error correction (QEC) advantage estimate
for the high-strain benchmark
error budget is reported in
Sec.~\ref{app:qec_montecarlo}.

\subsection*{3. Dual interpretation of the $A_2$ parity filter}
\label{app:duality}

The proof of Proposition~1(b) establishes that
$|D^{(1)}\rangle \propto |0\rangle \perp \mathcal{Q}$ for
$A_2$-sector perturbations.
Here we spell out the two consequences of this single
equation.

Under a quasi-static perturbation $V = \sigma S_z$
($A_2$ irrep), the first-order dark-state correction is
\begin{equation}
  |D^{(1)}\rangle
  = \frac{\sigma}{\sqrt{2}}
    \left[
      \frac{\langle B_+|S_z|D\rangle}{-\Omm/2}\,|B_+\rangle
      + \frac{\langle B_-|S_z|D\rangle}{+\Omm/2}\,|B_-\rangle
    \right].
  \label{eq:D1_A2}
\end{equation}
The matrix elements $\langle B_\pm|S_z|D\rangle$ have the same
sign (since $\langle 0|S_z|D\rangle = 0$ for spin-1), so
\begin{equation}
  |D^{(1)}\rangle
  = \frac{\sigma\langle B|S_z|D\rangle}{\sqrt{2}\cdot\Omm/2}
    \bigl[-|B_+\rangle + |B_-\rangle\bigr].
\end{equation}
Substituting $|B_\pm\rangle = (|B\rangle \pm |0\rangle)/\sqrt{2}$:
\begin{equation}
  -|B_+\rangle + |B_-\rangle
  = -\frac{|B\rangle + |0\rangle}{\sqrt{2}}
    + \frac{|B\rangle - |0\rangle}{\sqrt{2}}
  = -\sqrt{2}\,|0\rangle.
  \label{eq:A2_interference}
\end{equation}
Therefore
\begin{equation}
  \boxed{|D^{(1)}\rangle
  = -\frac{\sigma\langle B|S_z|D\rangle}{\Omm/2}\;|0\rangle.}
  \label{eq:D1_result}
\end{equation}

\textbf{Face~1 (Berry-connection invariance).}
Because $|D^{(1)}\rangle \propto |0\rangle$ and
$\partial_\mu|D\rangle \in \mathcal{Q}$ for all control
parameters~$\mu$, and $\langle 0|\partial_\mu D\rangle = 0$,
the first-order Berry-connection correction vanishes:
\begin{equation}
  \delta\mathcal{A}_\mu^{(1)}
  = i\langle D^{(1)}|\partial_\mu D\rangle
  + i\langle D|\partial_\mu D^{(1)}\rangle = 0.
\end{equation}
This is Proposition~1(b): the geometric phase is protected
at $O(\sigma)$, with the leading correction appearing at
$O(\sigma^2/\Omm^2)$.

\textbf{Face~2 (Leakage direction).}
The same equation shows that the perturbed dark state
$|D\rangle + |D^{(1)}\rangle$ acquires a component in
$|0\rangle$ with probability
$p_{\mathrm{leak}} \propto
\sigma^2|\langle B|S_z|D\rangle|^2/(\Omm/2)^2$.
This population lies outside the computational subspace
$\mathcal{Q}$ and points toward the optically distinguishable
auxiliary level, providing a microscopic explanation for why
$A_2$-type noise preferentially produces leakage directed at
$|0\rangle$ rather than at $|B\rangle$.

\textbf{Inseparability.}
The two faces are mathematically inseparable.
The parity filter causes destructive interference in the
$|B\rangle$ direction (the $|B\rangle$ components of
$|B_+\rangle$ and $|B_-\rangle$ cancel), which is
\emph{simultaneously} the reason the Berry connection is
invariant (Face~1) and the reason the leakage points toward
$|0\rangle$ (Face~2).

\textbf{$E$-sector contrast.}
For $E$-sector perturbations ($V = \sigma S_x$), the matrix
elements $\langle B_\pm|S_x|D\rangle$ have \emph{opposite}
sign (since $\langle B|S_x|D\rangle = 0$ but
$\langle 0|S_x|D\rangle \neq 0$ for spin-1), giving
\begin{equation}
  +|B_+\rangle + |B_-\rangle
  = +\sqrt{2}\,|B\rangle.
\end{equation}
The first-order correction is
$|D^{(1)}\rangle \propto |B\rangle \in \mathcal{Q}$:
it remains in the computational subspace as an in-subspace
error, not leakage to~$|0\rangle$.
Suppression in this sector relies on the $\pi$-echo between
the two NGQC lunes, not on the parity filter's leakage
direction.

\subsection*{4. Sector-injection diagnostic}
\label{app:surface_sector_projection}

The ideal proof above classifies possible system-operator sectors; it
does not assume that the fabricated surface-spin environment itself
respects the NV point group.  To connect the proof to the actual
Regime-A noise model and to verify the parity filter directly, we ran
a sector-injection diagnostic on the ideal three-level $\Lambda$
Hamiltonian
\begin{equation}
  H_0/\Omm = \frac{1}{2}(|0\rangle\langle B|+|B\rangle\langle 0|).
\end{equation}
Equal-Hilbert--Schmidt-norm sector representatives
\[
  \{S_z^2,\;S_z,\;S_x,\;S_y\}
\]
were injected with strength $\sigma$, and the first-order dark-state
correction was projected onto the auxiliary direction $|0\rangle$ and
the logical bright direction $|B\rangle$:
\begin{equation}
  f_{|0\rangle}=\frac{W_0}{W_0+W_B},\qquad
  f_{|B\rangle}=\frac{W_B}{W_0+W_B}.
\end{equation}
The resulting sector directions and scaling laws are summarized in
Fig.~\ref{fig:sector_injection_scaling}.

\begin{figure*}[!t]
  \centering
  \includegraphics[width=0.86\textwidth]{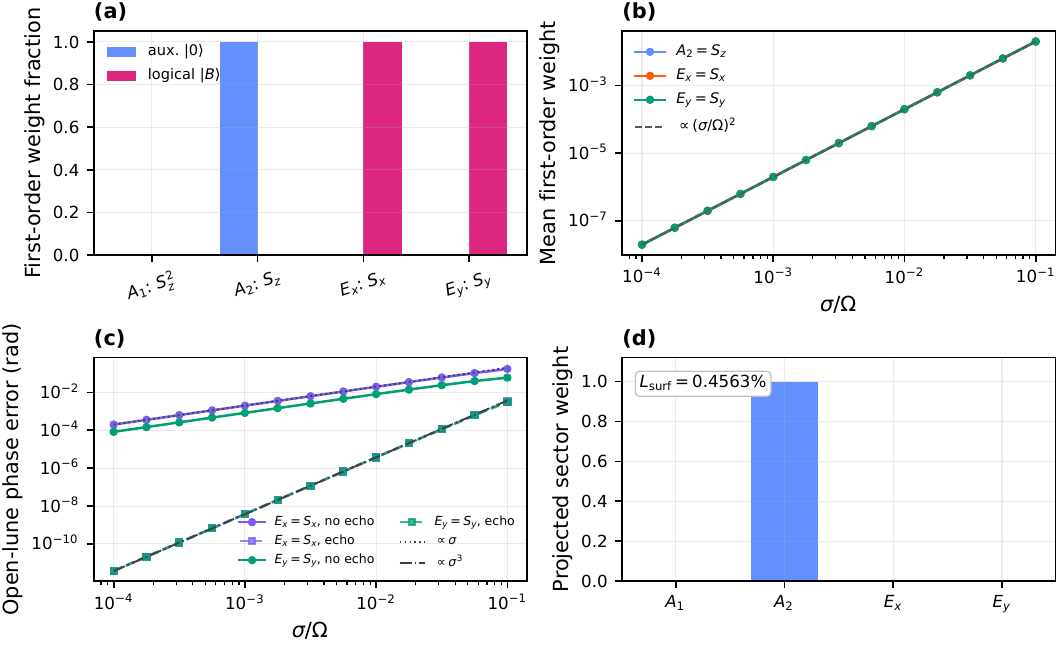}
  \caption{\label{fig:sector_injection_scaling}%
  Sector-injection verification of the $\Lambda$-manifold
  symmetry-to-channel map.  (a)~First-order response direction:
  $A_2=S_z$ points entirely toward $|0\rangle$, while
  $E=S_x,S_y$ points entirely toward $|B\rangle$.
  (b)~Nonzero response weights scale quadratically in
  $\sigma/\Omm$, as expected for probabilities from first-order
  amplitudes.  (c)~The $E$-sector geometric phase is approximately
  linear without echo and approximately cubic with the $\pi$-shifted
  two-lune echo.  (d)~The actual Regime-A surface-noise operator
  projects as $(w_{A_1},w_{A_2},w_{E_x},w_{E_y})=(0,1,0,0)$ and gives
  $L_{\mathrm{surf}}=0.4563\%$ in the surface-only diagnostic.}
\end{figure*}

The first-order directions are:
\begin{center}
\begin{tabular}{lccc}
  Sector & Operator & $f_{|0\rangle}$ & $f_{|B\rangle}$ \\
  \hline
  $A_1$ & $S_z^2$ & $0$ & $0$ \\
  $A_2$ & $S_z$ & $1.000000$ & $0.000000$ \\
  $E_x$ & $S_x$ & $0.000000$ & $1.000000$ \\
  $E_y$ & $S_y$ & $0.000000$ & $1.000000$
\end{tabular}
\end{center}
The $A_1$ row has zero first-order coupling because $S_z^2$ is the
identity on $\mathcal Q$; its fractions are therefore reported as zero
rather than interpreted as a direction.

The nonzero response weights have log--log slopes
\[
  A_2\to |0\rangle:\;2.000,\qquad
  E_x,E_y\to |B\rangle:\;2.000.
\]
This confirms perturbative scaling: the first-order correction is an
amplitude, while the plotted response is a probability/weight.
For the open-lune geometric phase, the $E$-sector slopes are
\[
  E_x:\;0.986\to 2.986,\qquad
  E_y:\;0.968\to 2.984,
\]
where the arrow denotes adding the $\pi$-shifted second lune.  Thus the
two-lune echo removes the leading $E$-sector geometric term.

The same calculation can be written as an irrep-to-channel response
matrix.  For a weak sector representative $V_\Gamma$, define
\begin{equation}
  \mathcal R_{\Gamma\to c}
  =
  \lim_{\sigma\to0}
  \frac{p_c[V_\Gamma]}{(\sigma/\Omm)^2},
  \label{eq:sm_irrep_response_matrix}
\end{equation}
where $c$ denotes a measured port.  Table~\ref{tab:irrep_polarimetry_response}
reports the finite-$\sigma/\Omm=10^{-3}$ estimate of this response in
the same equal-Hilbert--Schmidt-norm convention.  The two decisive
entries are $\mathcal R_{A_2\to P_0}=2$ and
$\mathcal R_{E\to W_B}=2$, with crossed population responses at the
numerical floor.

\begin{table*}[!t]
\caption{\label{tab:irrep_polarimetry_response}%
Irrep-to-channel response matrix of the ideal $\Lambda$ manifold.
Probability-like columns are normalized by $(\sigma/\Omm)^2$ at
$\sigma/\Omm=10^{-3}$.  The $A_2$ row also carries a static
longitudinal calibration component, shown separately as ``static
$Z/\sigma$''; the polarimetry statement concerns population-port
routing between $P_0$ and $W_B$.}
\begin{ruledtabular}
\begin{tabular}{lccccccc}
Sector & $R_{P_0}$ & $R_{W_B}$ & $R_{XY}$ & $R_Z^{(1)}$ &
$R_Z^{({\rm echo})}$ & common$/\sigma$ & static $Z/\sigma$ \\
\hline
$A_1:S_z^2$ & $9.447{\times}10^{-33}$ & $0$ & $0$ & $0$ & $0$ &
$1.000000$ & $0$ \\
$A_2:S_z$ & $2.000000$ & $0$ & $0$ & $0$ & $0$ &
$0$ & $1.000000$ \\
$E_x:S_x$ & $0$ & $2.000000$ & $2.000000$ & $0.997774$ &
$3.546{\times}10^{-12}$ & $0$ & $0$ \\
$E_y:S_y$ & $0$ & $2.000000$ & $2.000000$ & $0.170654$ &
$3.550{\times}10^{-12}$ & $0$ & $0$
\end{tabular}
\end{ruledtabular}
\end{table*}

Therefore, the clean $\Lambda$ manifold acts as a crystalline irrep
polarimeter.  In the syndrome-port response,
$A_2$ routes its first-order population weight to the auxiliary port,
whereas the transverse $E$ sector routes it to the logical bright port.
Separately, $S_z$ is also a static longitudinal splitting of the
logical doublet, so the polarimetry result should not be read as the
absence of all $Z$-like calibration shifts.

A useful continuous check is obtained by injecting
\begin{equation}
  V(\alpha,\beta)
  =
  \sigma\!\left[
  \cos\alpha\,S_z
  +
  \sin\alpha(\cos\beta\,S_x+\sin\beta\,S_y)
  \right].
  \label{eq:sm_irrep_interpolation}
\end{equation}
The measured routing fractions obey
\[
  f_{0}(\alpha)=\cos^2\alpha,\qquad
  f_B(\alpha)=\sin^2\alpha,
\]
independent of the transverse angle $\beta$.  Numerically, the maximum
absolute deviations from the two laws are
$6.661\times10^{-15}$, and the maximum $\beta$-standard deviations are
$7.474\times10^{-16}$ for $f_0$ and $7.549\times10^{-16}$ for $f_B$.
Figure~\ref{fig:irrep_polarimetry} summarizes the response matrix and
the continuous interpolation.

\begin{figure*}[!t]
  \centering
  \includegraphics[width=0.9\textwidth]{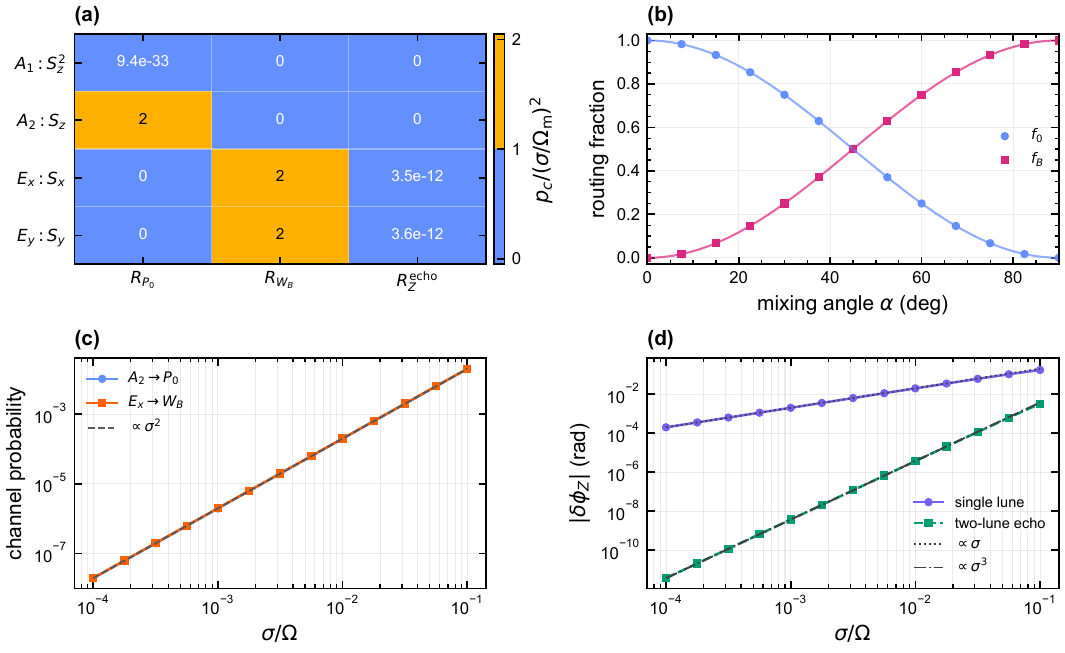}
  \caption{\label{fig:irrep_polarimetry}%
  Irrep-to-syndrome polarimetry of the ideal $\Lambda$ manifold.
  \textbf{(a)}~Response matrix $\mathcal R_{\Gamma\to c}$ for
  equal-norm representatives $S_z^2$, $S_z$, $S_x$, and $S_y$.
  The $A_2$ row is routed to the auxiliary port $P_0$, while the
  transverse $E$ rows are routed to the logical bright port $W_B$.
  \textbf{(b)}~Mixed $A_2$--$E$ perturbations obey
  $f_0=\cos^2\alpha$ and $f_B=\sin^2\alpha$, independent of the
  transverse angle $\beta$.
  \textbf{(c)}~The nonzero population-response ports scale as
  $(\sigma/\Omm)^2$.
  \textbf{(d)}~For an $E$-sector perturbation, the single-lune
  geometric phase is approximately first order in $\sigma/\Omm$,
  while the two-lune echo leaves an approximately cubic residual.}
\end{figure*}

The same diagnostic also projects the operator multiplying the
stochastic surface field used in the Regime-A simulation.  For
\[
  H_{\mathrm{surf}}(t)/h
  =
  \gamma_e B_z^{\mathrm{surf}}(t)S_z ,
\]
the Hilbert--Schmidt sector weights are
\[
  (w_{A_1},w_{A_2},w_{E_x},w_{E_y})
  =
  (0,\;1,\;0,\;0),
\]
to numerical precision.  Propagating 96 ME-CCE/KMC traces through the
Regime-A SATD control with only this surface-noise term present gives
\[
  L_{\mathrm{surf}} = 0.4563\% .
\]
This provides a direct numerical bridge between the actual surface
noise used in the simulation and the $A_2$ parity-filter mechanism:
the operator sector is fixed by the modeled coupling
$B_z^{\mathrm{surf}}(t)S_z$, while the absolute erasure probability is
fixed by the Regime-A device simulation.

\subsection*{5. Effective-spurion robustness diagnostic}
\label{app:sector_robustness}

We also ran a compact robustness diagnostic that perturbs the ideal
$A_2\to |0\rangle$ map at the effective $\Lambda$-Hamiltonian level.
This is not a full actuator-transfer-function simulation.  It asks
whether the parity-filter fractions degrade continuously under
detuning, arm-amplitude, and arm-phase spurions.

For a dressed-state detuning $\Delta/\Omm$, the analytic leakage of the
first-order direction into $|B\rangle$ is
\begin{equation}
  f_{|B\rangle}^{(\Delta)}
  =
  \frac{4(\Delta/\Omm)^2}{1+4(\Delta/\Omm)^2}.
\end{equation}
For fractional arm-amplitude skew $a$ and phase skew $\delta\phi$, the
small-error scalings are
\begin{equation}
  f_{|B\rangle}^{(a)}\approx a^2,\qquad
  f_{|B\rangle}^{(\delta\phi)}\approx \sin^2\delta\phi .
\end{equation}
The numerical robustness scan is shown in
Fig.~\ref{fig:sector_injection_robustness}.

\begin{figure*}[t]
  \noindent\makebox[\textwidth][c]{%
    \includegraphics[width=0.86\textwidth]{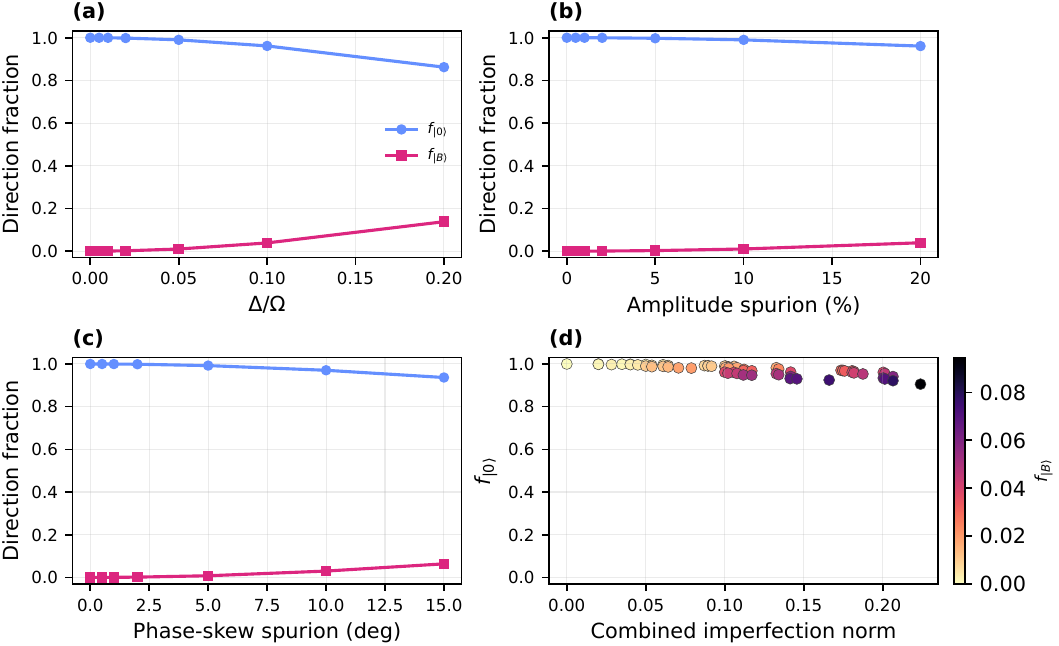}}
  \caption{\label{fig:sector_injection_robustness}%
  Effective-spurion robustness of the $A_2\to |0\rangle$ parity
  filter.  \textbf{(a)} Auxiliary and logical-bright direction
  fractions $f_{|0\rangle}$ and $f_{|B\rangle}$ versus detuning
  $\Delta/\Omega$.  \textbf{(b)} Direction fractions under quadrature
  amplitude imbalance, shown as an amplitude-spurion percentage.
  \textbf{(c)} Direction fractions under quadrature phase skew.
  \textbf{(d)} Combined-spurion scan showing $f_{|0\rangle}$ versus
  the combined imperfection norm, with point color indicating
  $f_{|B\rangle}$.  Across these scans the $A_2$ response degrades
  continuously into the logical-bright direction rather than showing a
  sharp failure of the sector map.}
\end{figure*}

\begin{center}
\begin{tabular}{lcc}
  Imperfection & $f_{|0\rangle}$ & $f_{|B\rangle}$ \\
  \hline
  $\Delta/\Omm=0.10$ & $0.961538$ & $0.038462$ \\
  $a=10\%$ & $0.989968$ & $0.010032$ \\
  $\delta\phi=10^\circ$ & $0.970413$ & $0.029587$ \\
  combined worst point & $0.905407$ & $0.094593$
\end{tabular}
\end{center}
The combined point uses
$\Delta/\Omm=0.10$, $a=10\%$, and
$\delta\phi=10^\circ$.  Even under simultaneous moderate
imperfections in this compact model, the nominal $A_2$ perturbation
remains dominantly auxiliary-directed.  The correct conclusion is
graceful degradation of the sector map, not a hardware-tolerance claim
for a specific transducer.

\subsection*{6. What symmetry does not determine}

The numerical probabilities $p_{\mathrm{era}}$, $p_Z$,
$p_{\mathrm{dep}}$, and $p_{XY}$ are not fixed by representation
theory. They depend on $T_1$, $T_{1\rho}$, surface-spin spectral
density, control imperfections, mechanical mode splitting, quadrature
calibration, and erasure-detection efficiency. In particular, the
ME-CCE/KMC bath in the Regime-A simulation is not assumed to be a
point-group-symmetric environment: it enters as the explicit
$A_2$ perturbation $B_z^{\mathrm{surf}}(t)S_z$.  The dominance of
detected leakage in the extracted channel is consistent with the
$A_2$ parity filter, but the absolute erasure probability is set by
the simulated bath realization and device parameters. A bath with
parametrically enhanced $E$-sector system-operator strength, a
detector with $\eta_{\mathrm{det}}\approx0$, or a device whose
mechanical doublet is not tuned within the RWA bandwidth would not
yield the same biased-erasure channel even though the operator
classification remains valid. Thus symmetry protects the local
circular selection rule and fixes the perturbative hierarchy of
allowed system operators; the simulated device model fixes the rates.

\section{Regime-B resonant HBAR transfer-function map}
\label{app:ghz}

Regime~B translates the Regime-A rotating-frame control Hamiltonian
into a resonant GHz-HBAR setting.  The shared $\Lambda$ Hamiltonian is
unchanged; the resonator implementation sets the attainable $\Omm$,
linewidth, drive bandwidth, and envelope-tracking strategy.  The
membrane mode-topology benchmark of Sec.~IV\,C of the main text and
Appendices~\ref{app:modes} and~\ref{app:topology} operates at
$f_0=\SI{78.8}{\mega\hertz}$, far below
the NV zero-field splitting $D\approx\SI{2.87}{\giga\hertz}$ required
for resonant driving.  Here we present a concrete GHz resonator design
and quantify the achievable Rabi frequency.  The high-strain
$\Omm=\SI{2.22}{\mega\hertz}$ result in the main text establishes the
rotating-frame channel benchmark; the HBAR-realistic prediction is
reported separately below using the smaller Rabi frequencies
accessible in the GHz design.

\paragraph{Consistent Regime-B HBAR design.}
The Regime-B implementation used for Sweep~9 is a
$(100)$ diamond disk with $h_d=\SI{23.4}{\micro\meter}$,
$R=\SI{25}{\micro\meter}$, and a \SI{1}{\micro\meter} AlN
transducer.  This thickness gives an HBAR free spectral range
matching the NV Zeeman splitting at $B_z\simeq\SI{49}{\gauss}$,
so adjacent thickness-shear overtones address the two $\Lambda$
legs.  The overtone frequencies are
\begin{equation}
  f_n = \frac{n\,v_T}{2\,h_d},
  \qquad v_T = \SI{12\,822}{\meter\per\second},
  \label{eq:hbar_fn}
\end{equation}
so $n=10$ and $n=11$ lie near
$D_- \approx \SI{2.733}{\giga\hertz}$ and
$D_+ \approx \SI{3.007}{\giga\hertz}$.  Residual detunings can be
trimmed by RIE thickness control, uniaxial stress, or radius tuning
near an avoided crossing~\cite{chen2019_nanolett}.
AlN/single-crystal-diamond HBAR stacks have been demonstrated as
high-overtone bulk acoustic resonators~\cite{sorokin2013_aln_diamond_hbar}.

\paragraph{Degenerate mode pair for quadrature shear.}
For $[100]$ propagation in cubic diamond, the two transverse shear
branches are degenerate by four-fold symmetry, with
$v_T=\sqrt{C_{44}/\rho}=\SI{12\,822}{\meter\per\second}$
using $C_{44}=\SI{578}{\giga\pascal}$~\cite{mcskimin1972_diamond}
and $\rho=\SI{3515}{\kilo\gram\per\meter\cubed}$.
They generate the orthogonal shear components $\exz$ and $\eyz$.
Electrode loading and fabrication asymmetry can split the doublet,
but the residual splitting is assumed to be tuned below the envelope
bandwidth by the DC spring-softening method of
Appendix~\ref{app:tuning}.

The $n$-th shear overtone has
$\varepsilon_{xz}(z)=\varepsilon_0\cos(n\pi z/h_d)$.
For $n=10$, the surface is an antinode and a shallow
$d_{\mathrm{NV}}\sim\SI{20}{\nano\meter}$ NV samples
$\cos(n\pi d_{\mathrm{NV}}/h_d)>0.9999$, so the strain is effectively
uniform over the implanted layer.  Ensemble operation should place
NVs within an acoustic-antinode window.

Two orthogonal AlN electrodes, or one electrode pair with a
$90^\circ$ rotated pattern, drive the shear polarizations in
quadrature and reproduce the circular strain field of Eq.~(1) of the
main text.

\paragraph{Dual-overtone tuning for $\Lambda$-system control.}
The $\Lambda$-system gate couples both
$|0\rangle\!\leftrightarrow\!|{+1}\rangle$ (at $D_+$) and
$|0\rangle\!\leftrightarrow\!|{-1}\rangle$ (at $D_-$)
simultaneously (Sec.~II\,C of the main text).  Because one
GHz overtone has only $\Delta f_{3\mathrm{dB}}\sim\SI{300}{\kilo\hertz}$,
the two arms are driven by adjacent overtones whose free spectral
range equals the Zeeman splitting:
\begin{equation}
  \frac{v_T}{2h_d} = 2\gamma_e B_z.
  \label{eq:dual_overtone}
\end{equation}
The Regime-B values above give
$\mathrm{FSR}\approx\SI{274}{\mega\hertz}$, placing
$n=10$ and $n=11$ at $D_-$ and $D_+$.  Each overtone carries its own
shear doublet, providing independent $\sigma_-$ and $\sigma_+$
circular drives.

\paragraph{Achievable strain amplitude.}
The AlN layer produces a direct piezoelectric strain
$\varepsilon_{\mathrm{direct}}\simeq
d_{33}V_{\mathrm{RF}}/h_{\mathrm{AlN}}$, where
$d_{33} = \SI{5.5}{\pico\meter\per\volt}$ (AlN) and
$h_{\mathrm{AlN}} = \SI{1}{\micro\meter}$.
The local shear strain relevant for NV driving is not obtained from
this longitudinal AlN strain with unit efficiency.  We therefore
write the target shear amplitude as
\begin{equation}
  \varepsilon_0 \simeq
  \eta_{\rm sh} Q_L
  \frac{d_{33}V_{\mathrm{RF}}}{h_{\mathrm{AlN}}}.
  \label{eq:eps_hbar}
\end{equation}
Here $Q_L$ is the loaded acoustic quality factor and
$\eta_{\rm sh}$ is the effective AlN-to-diamond target shear-strain
transduction factor, incorporating AlN/diamond acoustic loading,
longitudinal-to-shear conversion, electrode participation, finite
mode overlap, and local strain participation at the NV.  The
unit-efficiency expression
($\eta_{\rm sh}=1$) is an upper bound.  The voltage values in
Table~\ref{tab:ghz_scenarios} use $Q_L=10^4$ and
$\eta_{\rm sh}=10^{-3}$, equivalently
$\eta_{\rm sh}Q_L\simeq10$.  This gives
$V_{\mathrm{RF}}\simeq\SI{0.018}{\volt}, \SI{0.18}{\volt},
\SI{0.91}{\volt}$ for
$\varepsilon_0=10^{-6},10^{-5},5\times10^{-5}$, respectively.
Diamond HBAR resonators routinely achieve $Q \gtrsim 10^4$ at GHz
frequencies: MacQuarrie \textit{et~al.}\ demonstrated
mechanically driven NV spin transitions with a diamond resonator
at $Q > 10^4$~\cite{macquarrie2013_prl,macquarrie2015_optica}, and Chen \textit{et~al.}\
achieved $Q \approx 1.2 \times 10^4$ in a semiconfocal diamond
acoustic resonator at $\sim\!\SI{3}{\giga\hertz}$~\cite{chen2019_nanolett},
with single-quantum NV spin driving confirmed at $Q = 10^4$ in
a subsequent experiment~\cite{chen2020_pra}.
These demonstrations establish $Q \geq 10^4$ as an
experimentally validated baseline for GHz diamond nanomechanics.

\paragraph{Cryogenic quality-factor prospects.}
Demonstrated GHz diamond devices already reach
$Q\simeq10^4$~\cite{chen2020_pra}; higher-$Q$ cryogenic phononic
platforms reach $10^6$--$10^{10}$ in related materials and
geometries~\cite{oh2025_diamond_omc,gokhale2020_epihbar,maccabe2020_si_omc}.
Reaching $Q\ge10^5$ in AlN-on-diamond would likely require improved
transducer growth, hybrid phononic confinement, or a separate
optomechanical phonon bus.

Table~\ref{tab:ghz_scenarios} summarizes the single HBAR
implementation assumed in Regime~B, together with the conservative,
moderate, and optimistic strain scenarios used in Sweep~9.

\begin{table*}[t]
\small
\caption{\label{tab:ghz_scenarios}Compact Regime-B GHz-HBAR
implementation table.  The $\Lambda$-leg coupling is
$\Omm=|\htwosix|\varepsilon_0$ with
$|\htwosix|=\SI{2830}{\mega\hertz/strain}$.}
\begin{tabular*}{\textwidth}{@{\extracolsep{\fill}}llll}
\toprule
  Item & Value & Purpose & Comment \\
  \midrule
  Geometry & $h_d=\SI{23.4}{\micro\meter}$, $R=\SI{25}{\micro\meter}$,
  $h_{\mathrm{AlN}}=\SI{1}{\micro\meter}$ &
  Sets HBAR spectrum & $(100)$ diamond disk \\
  Dual overtones & $n=10,11$; $f_{10}\simeq D_-$, $f_{11}\simeq D_+$ &
  $\Lambda$ legs & $B_z\simeq\SI{49}{\gauss}$, FSR $\simeq\SI{274}{\mega\hertz}$ \\
  Mode pair & Two shear polarizations/overtone &
  Circular strain & Residual splitting tuned out \\
  Bandwidth & $f/Q_L\simeq\SI{0.27}{\mega\hertz}$ for $Q_L=10^4$ &
  Envelope tracking & Predistortion or lower $Q_L$ if needed \\
  Strain cases & $10^{-6},10^{-5},5\times10^{-5}$ &
  $\Omm=2.83,28.3,141.5~\mathrm{kHz}$ &
  $V_{\mathrm{RF}}\simeq0.02,0.18,0.91~\mathrm{V}$ \\
  DQ-SATD & $f_{\mathrm{DQ}}\simeq\SI{274}{\mega\hertz}$ &
  CD tone & Separate phase-locked DQ strain drive \\
\bottomrule
\end{tabular*}
\end{table*}

The optimistic scenario ($\varepsilon_0=5\times10^{-5}$,
$V_{\mathrm{RF}}\simeq\SI{0.91}{\volt}$) is conditional on achieving
$\eta_{\rm sh}Q_L\simeq10$, or $\eta_{\rm sh}\simeq10^{-3}$ for
$Q_L=10^4$.  It remains well within diamond's elastic limit
($\varepsilon_\mathrm{yield} \sim 1\%$) and compatible with
cryogenic operation.

\paragraph{Heating and power dissipation.}
At $f = \SI{2.87}{\giga\hertz}$, the RF drive power is
$P = V_\mathrm{RF}^2 / (2 R_\mathrm{mot})$, where the motional
resistance $R_\mathrm{mot} \sim \SI{50}{\ohm}$ for a well-matched
HBAR (typical for AlN-on-diamond composite
resonators~\cite{lakin2003_thinfilm}).  For
$V_{\mathrm{RF}} = \SI{0.91}{\volt}$:
$P \simeq 0.91^2/100 \simeq \SI{8.3}{\milli\watt}$.
With diamond's room-temperature bulk thermal conductivity $\kappa = \SI{2200}{\watt\per\meter\per\kelvin}$,
a one-dimensional steady-state estimate gives
$\Delta T \approx P h_d / (\kappa\,\pi R^2) \ll \SI{1}{\kelvin}$.
This upper bound neglects lateral heat spreading (which further
reduces $\Delta T$) and thermal boundary resistance (which could
locally increase it); a finite-element thermal analysis is
warranted for the cryogenic operating regime, where
the cooling power budget is tighter.

\paragraph{Acoustic bandwidth constraint.}
The resonator is continuously driven at resonance throughout the
gate, so free ring-down does not occur.
The operative constraint is the acoustic modulation bandwidth:
the resonator must track the amplitude envelope of the control
protocol, whose modulation bandwidth is
$\Delta f_\mathrm{env} \approx 2/\tgate$.
Requiring the resonator's $3\,\mathrm{dB}$ bandwidth $f_0/Q$ to
exceed $\Delta f_\mathrm{env}$ yields
\begin{equation}
  Q < \frac{f_0\,\tgate}{2}.
  \label{eq:Qmax_bw}
\end{equation}
For $\tgate \geq \SI{7}{\micro\second}$, $Q = 10^4$ satisfies
this constraint directly.
At shorter gate times, one may reduce $Q$ and compensate with
modestly higher~$V_\mathrm{RF}$, or apply RF pre-distortion
(inverse filtering of the resonator transfer function), a standard
technique in resonator-coupled superconducting-qubit control.
The pre-ring time to reach steady-state strain is
$t_\mathrm{pre} \approx 3Q/(\pi f_0) \approx
\SI{3.3}{\micro\second}$ for $Q = 10^4$.

\paragraph{Gate-time landscape at reduced $\Omm$.}
The SATD counter-diabatic protocol removes the adiabatic-speed
constraint ($\lambda = 1$ cancels leakage exactly), so the gate time
is no longer limited by the adiabatic ratio~$\eta$.
The dominant fidelity limit shifts to $T_{1\rho}$ decoherence:
$F_\mathrm{limit} \approx e^{-\tgate / T_{1\rho}}$.
For $T_{1\rho} = \SI{500}{\micro\second}$, a gate time
$\tgate = \SI{10}{\micro\second}$ yields
$F_\mathrm{limit} \approx 0.980$, while
$\tgate = \SI{5}{\micro\second}$ yields
$F_\mathrm{limit} \approx 0.990$.
The optimal gate time balances the rising decoherence cost against
the diminishing SATD fidelity gain at longer durations.

Sweep~9 maps this trade-off numerically for all three HBAR scenarios.
At $\tgate = \SI{2}{\micro\second}$, the projected HBAR cases give
$\Favg = 99.50$--$99.85\%$; the optimistic design reaches
$\Favg = 99.05\%$ at the first Rabi-refocusing magic time
($\tgate \approx \SI{14}{\micro\second}$).
These Sweep~9 values assume the same ideal rotating-frame control
Hamiltonian used in the benchmark simulations, including the resonant
DQ SATD tone.  HBAR-specific phase noise, heating, and transducer
filtering are bounded separately rather than included as Lindblad
channels; the Floquet validation in Appendix~\ref{app:extended_sweeps_sm},
Sec.~6,
shows that a bare $Q=10^4$ HBAR transfer function is narrower than
the \SI{1.833}{\micro\second} Regime-A benchmark envelope.  The
resonant route therefore requires predistortion, reduced
effective~$Q$, or a slower HBAR-specific pulse, converting the
benchmark into a concrete resonator-design target.
The biased-erasure estimates in the main text therefore remain tied
to the extracted Regime-A channel.  In the HBAR route, residual
envelope filtering, phase skew, and quadrature imbalance are promoted
to the transverse-floor sensitivity parameters used in
Appendix~\ref{app:qec_montecarlo}, Sec.~6, rather than folded into a new
headline erasure-overhead number.
As a control-optimality check, a GRAPE (gradient ascent pulse
engineering) optimization of the piecewise-constant drive amplitudes
at the optimistic Rabi frequency
($\Omm = \SI{141.5}{\kilo\hertz}$, $\tgate = \SI{2}{\micro\second}$)
achieves a noiseless process fidelity of $99.9997\%$, confirming
that the analytical SATD pulse shape is near-optimal and that the
${\sim}0.4\%$ infidelity in Sweep~9 is entirely decoherence-limited.

\paragraph{Noise-channel transferability.}
The three simulation noise channels ($T_1$, $T_{1\rho}$, ME-CCE
surface noise) can be evaluated at any specified~$\Omm$ and do not
depend directly on the acoustic carrier frequency~$f_0$ in the
rotating-frame Hamiltonian.
Implementation-dependent channels (resonator amplitude noise,
drive-induced heating, transducer bandwidth) are not integrated
into the master equation but are bounded by the engineering
estimates above.
As a stress test, enhancing the effective surface-noise coupling
by $5\times$ adds only $\sim\!0.08\%$ to the total error budget
(surface noise contributes $\sim\!0.02\%$ at the high-strain
benchmark), bringing
$\Favg$ from $99.88\%$ to $\sim\!99.80\%$, still well within
the high-fidelity regime.

\section{Complete parameter table}
\label{app:params}

Table~\ref{tab:params} collects all physical and design parameters
used throughout this work.

\begin{table}[!ht]
\caption{\label{tab:params}Complete parameter table.}
\scriptsize
\setlength{\tabcolsep}{2pt}
\resizebox{0.92\columnwidth}{!}{%
\begin{minipage}{1.35\columnwidth}
\begin{ruledtabular}
\begin{tabular}{llrll}
  Parameter & Symbol & Value & Units & Source \\
  \hline
  Zero-field splitting & $D$ & 2.87 & GHz & Ref.~\cite{udvarhelyi2018_prb} \\
  Gyromagnetic ratio & $\gamma_e$ & 2.80 & MHz/G & CODATA \\
  Static magnetic field & $B_z$ & 50 & G & Design \\
  Shear coupling & $\htwosix$ & $-2830$ & MHz/strain & Ref.~\cite{udvarhelyi2018_prb} \\
  In-plane coupling & $\htwofive$ & $-2600$ & MHz/strain & Ref.~\cite{udvarhelyi2018_prb} \\
  Transverse DQ coupling & $\honesix$ & $19\,660$ & MHz/strain & Ref.~\cite{udvarhelyi2018_prb} \\
  Young's modulus & $Y_{[100]}$ & $1050$ & GPa & Ref.~\cite{mcskimin1972_diamond} \\
  Poisson's ratio & $\nu_{[100]}$ & $0.104$ & --- & Ref.~\cite{mcskimin1972_diamond} \\
  Diamond density & $\rho$ & $3515$ & kg/m$^3$ & Ref.~\cite{mcskimin1972_diamond} \\
  Membrane side & $L$ & $10$ & $\mu$m & Design \\
  Membrane thickness & $h_m$ & $200$ & nm & Design \\
  SS spring constant & $\kmech^{\mathrm{SS}}$ & $4310$ & N/m & Computed \\
  Clamped spring constant & $\kmech^{\mathrm{clamp}}$ & $9530$ & N/m & Computed \\
  Fundamental frequency & $f_0$ & $78.8$ & MHz & Computed \\
  Vacuum gap & $d$ & $200$ & nm & Design \\
  Breakdown voltage & $V_{\mathrm{bd}}$ & $35$ & V & Empirical \\
  Max.\ displacement & $x_0$ & $5$ & nm & Linear limit \\
  Benchmark Rabi frequency & $\Omm$ & $2.22$ & MHz & $0.444 \times x_0$ \\
  $T_1$ relaxation & --- & $1.0$ & ms & Ref.~\cite{nagura2026} \\
  $T_{1\rho}$ relaxation & --- & $0.5$ & ms & Enriched \\
  NV depth & $d_{\mathrm{NV}}$ & $20$ & nm & Design \\
  Surface density & $\rho_s$ & $0.004$ & nm$^{-2}$ & Ref.~\cite{myers2014_surfacenoise} \\
  Corr.\ length & $r_c$ & $5$ & nm & Ref.~\cite{nagura2026} \\
  Hopping time (tested) & $\tauc$ & $10^{-8}$--$2\!\times\!10^{-5}$ & s & --- \\
  MC trajectories & --- & $500$ & per point & Convergence \\
  NGQC $\theta_{\max}$ & --- & $\pi$ & rad & Full N$\to$S$\to$N \\
  SATD strength & $\lambda$ & $1.0$ & --- & Exact \\
  Optimal gate time & $\tgate$ & $1.833$ & $\mu$s & Sweep~7 \\
\end{tabular}
\end{ruledtabular}
\end{minipage}%
}
\end{table}

\section{Derivation of the DQ AC Stark Hamiltonian}
\label{app:dqstark}
\label{app:dq_stark}

This appendix derives the effective double-quantum (DQ) AC Stark
Hamiltonian stated in Theorem~2 of the main text from first
principles.  The calculation closely follows the dispersive
(light-shift) formalism of
Refs.~\cite{blais2004_cqed,grimm2000_dipole,cohen_tannoudji1998_api},
adapted to the spin-1 NV Hamiltonian and two-drive gate architecture.

\subsection{DQ operator structure}

In the $\{|{+1}\rangle,|0\rangle,|{-1}\rangle\}$ eigenbasis of $\Sz$
the DQ operators take the matrix form
\begin{equation}
  S_x^2 - S_y^2 =
  \begin{pmatrix}0&0&1\\0&0&0\\1&0&0\end{pmatrix}, \quad
  \{S_x,S_y\} =
  \begin{pmatrix}0&0&-i\\0&0&0\\i&0&0\end{pmatrix},
  \label{eq:DQ_ops}
\end{equation}
so that the raising/lowering combinations are
\begin{equation}
  (S_x^2 - S_y^2) \pm i\{S_x,S_y\}
  = 2|{\mp 1}\rangle\langle{\pm 1}|.
  \label{eq:DQ_ladder}
\end{equation}
Crucially, both operators have \emph{identically zero} matrix
elements involving $|0\rangle$.  Consequently, any second-order
energy shift generated by the DQ channel shifts $|{\pm 1}\rangle$
only relative to each other, never relative to $|0\rangle$.

\subsection{DQ coupling from a circularly polarized strain drive}

The DQ part of the Udvarhelyi spin-strain
Hamiltonian~\cite{udvarhelyi2018_prb} is
\begin{equation}
  \frac{H_{\varepsilon 2}}{h}
  = \frac{\honesix}{2}\bigl[
    \exz\,(S_y^2 - S_x^2)
    + \eyz\,\{S_x,S_y\}
    \bigr].
  \label{eq:Heps2_app}
\end{equation}
For a $\sigma_+$ circularly polarized shear drive at frequency
$\omega$, $\exz = \varepsilon_0\cos\omega t$ and
$\eyz = \varepsilon_0\sin\omega t$.  Inserting
Eq.~\eqref{eq:DQ_ops} and defining
$\Omega_{\mathrm{DQ}} \equiv \honesix\,\varepsilon_0$ gives
\begin{equation}
  H_{DQ}^{(\sigma_+)}
  = -\frac{h\,\Omega_{\mathrm{DQ}}}{2}\bigl[
    e^{i\omega t}|{+1}\rangle\langle{-1}|
    + e^{-i\omega t}|{-1}\rangle\langle{+1}|
    \bigr].
  \label{eq:VDQ_sigp}
\end{equation}
For the opposite helicity ($\sigma_-$; same $\exz$ but
$\eyz \to -\eyz$) the exponential signs swap:
\begin{equation}
  H_{DQ}^{(\sigma_-)}
  = -\frac{h\,\Omega_{\mathrm{DQ}}}{2}\bigl[
    e^{-i\omega t}|{+1}\rangle\langle{-1}|
    + e^{+i\omega t}|{-1}\rangle\langle{+1}|
    \bigr].
  \label{eq:VDQ_sigm}
\end{equation}
This sign asymmetry between helicities is the physical origin of
the $\Sz$ (rather than $\Sz^2$) structure of the effective
Hamiltonian.

\subsection{Two-drive architecture and rotating frame}

The gate protocol uses two strain drives:
\begin{equation}
  \text{Drive 1:}\; \omega_+ = 2\pi(D + \gamma_e B_z),\quad
  g_+ = \tfrac{\Omega_{\mathrm{DQ}}}{2}\cos(\theta/2),
  \label{eq:drive1}
\end{equation}
\begin{equation}
  \text{Drive 2:}\; \omega_- = 2\pi(D - \gamma_e B_z),\quad
  g_- = \tfrac{\Omega_{\mathrm{DQ}}}{2}\sin(\theta/2),
  \label{eq:drive2}
\end{equation}
where $\theta(t)$ is the polar angle of the Bloch-sphere trajectory.
The simulation operates in the doubly-rotating frame defined by
\begin{equation}
  |{+1}\rangle \to e^{-i\omega_+ t}|{+1}\rangle, \quad
  |{-1}\rangle \to e^{-i\omega_- t}|{-1}\rangle, \quad
  |0\rangle \to |0\rangle.
  \label{eq:rot_frame}
\end{equation}
Under this transformation the DQ coupling from Drive~1
($\sigma_+$ at $\omega_+$) oscillates at the residual detuning
\begin{equation}
  \Omega_- \equiv \omega_+ - (\omega_+ - \omega_-)
  = \omega_- = 2\pi(D - \gamma_e B_z),
  \label{eq:det1}
\end{equation}
while the coupling from Drive~2 ($\sigma_-$ at $\omega_-$) oscillates
at
\begin{equation}
  \Omega_+ \equiv \omega_- + (\omega_+ - \omega_-)
  = \omega_+ = 2\pi(D + \gamma_e B_z).
  \label{eq:det2}
\end{equation}
Both detunings are of order $D \sim \SI{2.87}{\giga\hertz}$, three
orders of magnitude larger than the couplings
$g_\pm \sim \SI{}{MHz}$.

\subsection{Second-order energy shifts}

For a monochromatic off-resonant perturbation
$V(t) = Ae^{i\omega t} + A^\dagger e^{-i\omega t}$ acting on states
that are degenerate in the rotating frame, the standard AC Stark
(light-shift) formula~\cite{cohen_tannoudji1998_api,grimm2000_dipole}
gives the second-order energy shift of state $|n\rangle$:
\begin{equation}
  \delta E_n^{(2)}
  = \sum_{m\neq n}\!\left[
    \frac{|\langle m|A|n\rangle|^2}{\hbar\omega}
    + \frac{|\langle m|A^\dagger|n\rangle|^2}{-\hbar\omega}
  \right].
  \label{eq:ACStark_formula}
\end{equation}
This is the dispersive limit of the driven two-level
system~\cite{blais2004_cqed}: the two terms correspond to virtual
absorption and emission with energy denominators $+\hbar\omega$ and
$-\hbar\omega$, respectively.

\paragraph{Drive~1} ($A^{(1)} = -g_+|{+1}\rangle\langle{-1}|$,
oscillating at $\Omega_-$):
\begin{align}
  \delta E_{-1}^{(1)} &= +\frac{g_+^2}{\hbar\Omega_-} > 0, \\
  \delta E_{+1}^{(1)} &= -\frac{g_+^2}{\hbar\Omega_-} < 0, \\
  \delta E_0^{(1)} &= 0.
  \label{eq:shifts_d1}
\end{align}

\paragraph{Drive~2} ($A^{(2)} = -g_-|{-1}\rangle\langle{+1}|$,
oscillating at $\Omega_+$):
\begin{align}
  \delta E_{+1}^{(2)} &= +\frac{g_-^2}{\hbar\Omega_+} > 0, \\
  \delta E_{-1}^{(2)} &= -\frac{g_-^2}{\hbar\Omega_+} < 0, \\
  \delta E_0^{(2)} &= 0.
  \label{eq:shifts_d2}
\end{align}
In both cases $\delta E_0 = 0$ because the DQ operators have no
matrix elements with $|0\rangle$ [Eq.~\eqref{eq:DQ_ops}].

\subsection{Total effective Hamiltonian}

Summing the contributions:
\begin{align}
  \delta E_{+1}
  &= -\frac{g_+^2}{\hbar\Omega_-}
     +\frac{g_-^2}{\hbar\Omega_+}, \notag\\
  \delta E_{-1}
  &= +\frac{g_+^2}{\hbar\Omega_-}
     -\frac{g_-^2}{\hbar\Omega_+}, \notag\\
  \delta E_0 &= 0.
  \label{eq:total_shifts}
\end{align}
The common-mode shift vanishes identically:
$(\delta E_{+1} + \delta E_{-1})/2 = 0$.
The effective Hamiltonian restricted to $\mathcal{Q}$ is therefore
\begin{align}
  H_{\mathrm{eff}}^{DQ}
  &=
  \delta E_{+1}|{+1}\rangle\langle{+1}|
  + \delta E_{-1}|{-1}\rangle\langle{-1}| \nonumber\\
  &=
  \left(\frac{g_-^2}{\Omega_+}
  - \frac{g_+^2}{\Omega_-}\right)\Sz .
  \label{eq:Heff_Sz}
\end{align}
Substituting
\[
  g_+ = \frac{\Omega_{\mathrm{DQ}}}{2}\cos\frac{\theta}{2},
  \qquad
  g_- = \frac{\Omega_{\mathrm{DQ}}}{2}\sin\frac{\theta}{2},
\]
and combining over the common denominator
$D^2 - \gamma_e^2 B_z^2$:
\begin{equation}
  \boxed{
  H_{\mathrm{eff}}^{DQ}
  = -\frac{\Omega_{\mathrm{DQ}}^2}{4D}
    \!\left(\cos\theta
     + \frac{\gamma_e B_z}{D}\right)\Sz
    + \mathcal{O}\!\left(\frac{\gamma_e^2 B_z^2}{D^2}\right),
  }
  \label{eq:Heff_final}
\end{equation}
where we used
$\sin^2(\theta/2) - \cos^2(\theta/2) = -\cos\theta$ and expanded
to leading order in $\gamma_e B_z / D \approx 0.049$.
This reproduces the AC Stark theorem of the main text, specifically
the $\Sz$ Hamiltonian structure in Eq.~\eqref{eq:Hstark_Sz} and the
residual scale in Eq.~\eqref{eq:stark_residual}.
The $\Sz$ (not $\Sz^2$) structure arises because the DQ channel
couples only $|{+1}\rangle \leftrightarrow |{-1}\rangle$ and
cannot produce virtual transitions through $|0\rangle$, yielding
an antisymmetric shift analogous to the dispersive qubit-cavity
pull~\cite{blais2004_cqed}.

\subsection{C3v Stark-scaling substitution for 3C-SiC}
\label{app:sic_stark_scaling}

For any $C_{3v}$ spin-1 platform using the same single-quantum
$\Lambda$-leg mechanism, the strain required to reach the target
mechanical Rabi rate is
\begin{equation}
  \varepsilon_0=\frac{\Omm}{|h_{26}|},
\end{equation}
and the off-resonant DQ coupling generated by the same strain is
\begin{equation}
  \Omega_{\mathrm{DQ}}=|h_{16}|\varepsilon_0
  =
  \left|\frac{h_{16}}{h_{26}}\right|\Omm .
\end{equation}
Thus
\begin{equation}
  \delta_{\mathrm{AC}}
  =
  \frac{1}{4D}
  \left|\frac{h_{16}}{h_{26}}\right|^2\Omm^2 .
\end{equation}
At the fixed target $\Omm=\SI{2.22}{\mega\hertz}$ used in the
Regime-A benchmark, the substitution gives
\begin{center}
\footnotesize
\begin{tabular}{lccccc}
  Platform & $D$ (GHz) & $h_{26}$ & $h_{16}$ &
  $|h_{16}/h_{26}|$ & $\delta_{\mathrm{AC}}$ (kHz) \\
  \hline
  NV & $2.870$ & $2.830$ & $19.660$ & $6.947$ & $20.718$ \\
  3C-SiC & $1.330$ & $1.800$ & $1.350$ & $0.750$ & $0.521$
\end{tabular}
\end{center}
The reduction factor is
\[
  \delta_{\mathrm{AC}}^{\mathrm{NV}}/
  \delta_{\mathrm{AC}}^{\mathrm{3C\mbox{-}SiC}}
  = 39.76 .
\]
This is the Hamiltonian-level input to the matched-environment
3C-SiC benchmark in Appendix~\ref{app:extended_sweeps_sm}, Sec.~7.

\section{Step-by-step experimental protocol}
\label{app:protocol}

The following protocol separates the two NV implementation layers used
in the paper.  Path~A characterizes the membrane mode-pair topology
and calibrates the rotating-frame channel benchmark.  Path~B
translates the same $\Lambda$-control protocol to a resonant GHz-HBAR
implementation, where Rabi rate and envelope tracking are set by the
acoustic transfer function.  Because the rotating-frame dynamics
depend only on $\Omm$ (Sec.~II\,C of the main text), the calibration
and verification steps~(4--7) apply identically to both implementation
paths described below once their attainable $\Omm$ and transfer
function are specified.

\paragraph{Path~A: Membrane mode-topology calibration path.}
This path implements the mode-pair geometry that tests the holonomic
strain topology.  Appendix~\ref{app:topology} gives the thin-plate
strain-topology guide, but the local Rabi rate is not calibrated by
that topology alone; this path must therefore measure the actual Rabi
rate and use the Regime-A value as an effective rotating-frame channel
target.

\begin{enumerate}
  \item \textbf{Membrane fabrication and characterization.}
    Fabricate a simply-supported diamond membrane
    ($L \approx \SI{10}{\micro\meter}$, $h_m \approx \SI{200}{\nano\meter}$)
    via reactive-ion etching (RIE) and undercut release.
    Target EBL alignment $\delta < 0.1\%$ to minimize
    $(1{,}2)/(2{,}1)$ frequency splitting.
    Verify the boundary condition (SS vs.\ clamped) via
    mechanical spectroscopy.

  \item \textbf{Mode frequency measurement and electrostatic tuning.}
    Measure $f_{12}$ and $f_{21}$ using laser interferometry or
    electrical impedance spectroscopy.
    Apply DC voltage to the full-area electrode
    ($V \leq V_{\mathrm{bd}} = \SI{35}{\volt}$) to close any
    residual splitting $\Delta f$ via differential spring softening
    (Appendix~\ref{app:tuning}).

  \item \textbf{NV localization and depth estimation.}
    Confocally locate a single NV center near $(L/4, L/4)$
    where the shear strain components $\exz$, $\eyz$ form a
    circularly rotating pair.
    Estimate the NV depth $d_{\mathrm{NV}}$ via spin-echo
    $T_2$ measurements (shallower NVs exhibit shorter $T_2$
    from surface noise~\cite{myers2014_surfacenoise}).

  \item \textbf{SATD calibration via Rabi characterization.}
    Drive a single mode at calibrated amplitude and measure
    the Rabi frequency~$\Omm$.
    Perform a $\lambda$-scan of the counter-diabatic drive
    amplitude (cf.\ Sweep~6): measure process fidelity at several
    $\lambda$ values bracketing $\lambda = 1$.
    The SATD resonance should be recoverable to $\pm 1\%$
    precision within $\sim\!10$ calibration points
    (Sec.~VI\,B, Sweep~6, of the main text).

  \item \textbf{Magic gate-time identification.}
    With calibrated $\Omm$ and $\lambda = 1$, sweep the gate
    time $\tgate$ and measure $\Favg$ via quantum process
    tomography (QPT).
    Identify the first magic time where bright-state dynamical
    phase refocuses (expected near
    $\tgate \approx 4/\Omm \approx \SI{1.8}{\micro\second}$ for
    $\Omm = \SI{2.22}{\mega\hertz}$).

  \item \textbf{DC Stark compensation.}
    With both modes driven in quadrature, measure the
    equatorial-input process fidelity with and without a DC
    compensation voltage $V_{\mathrm{DC}}$; adjust until
    the $0.2\%$ process-level Stark degradation is nulled.
    The required compensation is deterministic
    ($V_{\mathrm{DC}} \propto \delta_{\mathrm{AC}}$) and
    need not be recalibrated dynamically.

  \item \textbf{Process-tomography verification.}
    At the optimized operating point ($\tgate$, $\lambda$,
    $V_{\mathrm{DC}}$), perform full single-qubit QPT.
    Extract $\Favg$, $F_e$, and $\overline{\mathcal{L}}$.
    Verify that eigenstate fidelity matches process fidelity
    within the Stark-immunity prediction
    (Theorem~2 of the main text).
\end{enumerate}

\paragraph{Path~B: GHz-HBAR resonant implementation route.}
For resonant driving at $f_0 \approx D \approx \SI{2.87}{\giga\hertz}$,
replace steps~(1--3) above with:
\begin{enumerate}
  \item[\textbf{1$'$.}] \textbf{HBAR disk fabrication.}
    Pattern the Regime-B diamond micro-disk
    ($R \approx \SI{25}{\micro\meter}$,
    $t_d \approx \SI{23.4}{\micro\meter}$) capped with an AlN
    piezoelectric transducer for electrical actuation of adjacent
    thickness-shear overtones near $D_\pm$.

  \item[\textbf{2$'$.}] \textbf{Overtone selection and dual-mode
    splitting.}
    Identify the $n=10$ and $n=11$ shear-overtone doublets, and close
    residual polarization splitting electrostatically or via slight
    disk tilt.  Confirm that the Rabi frequency lies within
    $\Omm \approx 2.83$--$\SI{141.5}{\kilo\hertz}$
    (Sweep~9, Table~\ref{tab:ghz_scenarios}).

  \item[\textbf{3$'$.}] \textbf{NV--overtone coupling verification.}
    Drive one overtone and measure $\Omm$ via spin-echo Rabi
    oscillations.  Steps~(4--7) then proceed as above.
\end{enumerate}

\section{AC Stark injection test (Sweep~8)}
\label{sec:starkinjection}

Sweep~8 tests the $\Sz$ structure of the AC Stark shift
(Theorem~2 of the main text)
by amplifying $\delta_{\mathrm{AC}} \to \alpha_S\,\delta_{\mathrm{AC}}$
with $\alpha_S \in \{0, \ldots, 200\}$ while disabling DC
compensation (Table~\ref{tab:starkinjection_sm}).
Because the injected term acts dynamically during the gate rather than
as a static post-gate phase, both legacy and process-level metrics
become sensitive to large uncompensated Stark shifts.
At the physical scale $\alpha_S=1$ the legacy metric remains high
($99.78\%$), but the process fidelity is already phase-sensitive
($99.04\%$).
Larger $\alpha_S$ values produce nonmonotonic coherent wrapping and
leakage, confirming the need for DC compensation/calibration and
ruling out an $\Sz^2$ common-mode mechanism.
With dynamic compensation enabled, the full process fidelity is
recovered.
The dashed 3C-SiC overlay in Fig.~\ref{fig:starkinjection_sm}
is a separate matched-environment rerun with 50 trajectories per
point.  At the physical scale $\alpha_S=1$ it gives
$\Favg=99.8846\%$, statistically indistinguishable from the
no-Stark point $\Favg=99.8830\%$, with leakage $0.476\%$.
Even at $\alpha_S=10$ the 3C-SiC process fidelity remains
$99.848\%$, illustrating the $39.76\times$ reduction of the
bare Stark scale derived in Appendix~\ref{app:dqstark}, Sec.~6.

\begin{table}[htbp]
\caption{\label{tab:starkinjection_sm}NV AC Stark injection test results
at $\tgate = \SI{1.833}{\micro\second}$, composite NGQC + SATD,
DC compensation OFF, 200 trajectories.  $F_{\mathrm{leg}}$ is the
single-state dark eigenstate fidelity.  The matched 3C-SiC overlay
shown in Fig.~\ref{fig:starkinjection_sm} is from the separate
50-trajectory rerun described in the text.}
\begin{ruledtabular}
\begin{tabular}{ccccc}
  $\alpha_S$ & $\delta_{\mathrm{AC}}^{\mathrm{eff}}$ (kHz) &
  $F_{\mathrm{leg}}$ (\%) & $\Favg$ (\%) & Leakage \\
  \hline
  $0$   & $0$     & $99.81$ & $99.88$ & $0.48\%$ \\
  $1$   & $20.7$  & $99.78$ & $99.04$ & $0.50\%$ \\
  $5$   & $103.6$ & $98.37$ & $79.49$ & $0.64\%$ \\
  $10$  & $207$   & $94.47$ & $43.11$ & $0.91\%$ \\
  $20$  & $414$   & $94.67$ & $66.08$ & $1.06\%$ \\
  $50$  & $1036$  & $25.32$ & $51.33$ & $35.25\%$ \\
  $100$ & $2072$  & $58.48$ & $41.68$ & $48.35\%$ \\
  $200$ & $4144$  & $48.95$ & $35.87$ & $25.76\%$ \\
\end{tabular}
\end{ruledtabular}
\end{table}

The table gives the NV point values; Fig.~\ref{fig:starkinjection_sm}
plots the same Sweep~8 data and overlays the matched-environment
3C-SiC neutral-divacancy rerun so that the separation between the
single-state metric, the process metric, leakage, and platform Stark
scale is visible.
The coherent, nonmonotonic response is the important feature: the
Stark term does not behave like a simple monotone penalty, and the
process metric is the diagnostic that exposes the residual phase
sensitivity.

\begin{figure}[!htbp]
  \centering
  \includegraphics[width=\columnwidth]{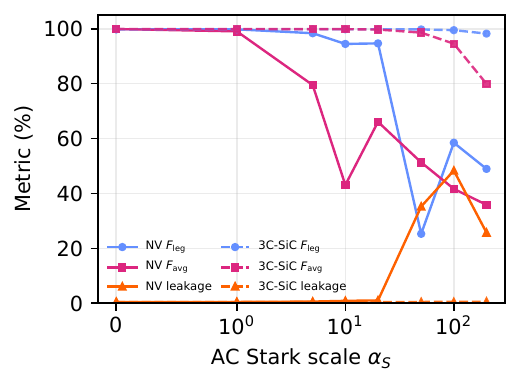}
  \caption{\label{fig:starkinjection_sm}%
  AC Stark injection response for Sweep~8.
  Solid curves show the NV channel benchmark and dashed curves show the
  matched-environment 3C-SiC neutral-divacancy rerun.
  Amplifying the uncompensated dynamic $\Sz$ Stark term separates
  the single-state metric, process metric, and leakage.
  At the physical scale $\alpha_S=1$, the 3C-SiC curve remains
  essentially on the no-Stark baseline, whereas the NV process metric
  is phase-sensitive; larger shifts show coherent wrapping and
  leakage, motivating DC compensation and distinguishing the channel
  from an $\Sz^2$ common-mode mechanism.}
\end{figure}

\section{Device design constraints}
\label{sec:design_sm}
\label{app:design_sm}

\subsection{Boundary conditions and displacement limit}
The simply-supported (SS) membrane has an effective spring constant
$\kmech^{\mathrm{SS}} = \SI{4310}{\newton\per\meter}$
(Sec.~A).  In the original slope-scale estimate this would permit
access to the full
$\Omm = \SI{2.22}{\mega\hertz}$ at
$x_0 = \SI{5}{\nano\meter}$.
Appendix~\ref{app:topology} uses the membrane modes as a
strain-topology guide rather than as a calibrated transverse-shear
prediction; the value is used in the channel simulations as an
effective Rabi-rate target that must be supplied by direct strain/Rabi
calibration or by an implementation such as the HBAR route.
Clamping the edges increases the stiffness by a factor of $2.21$
($\kmech^{\mathrm{clamp}} \approx \SI{9530}{\newton\per\meter}$),
reducing $\Omm$ from \SI{2.22}{\mega\hertz} (SS) to
\SI{1.004}{\mega\hertz} (clamped) at $x_0 = \SI{5}{\nano\meter}$.
With DRAG active, the clamped sweep remains high
($99.43\%$--$99.47\%$) but sits below the simply-supported optimum
($\sim\!99.49\%$) by about $0.02$--$0.06$ percentage points;
simply-supported boundaries remain preferred because higher
$\Omm$ enables greater gate-time compression.

\subsection{Electrostatic degeneracy tuning}
Standard electron-beam lithography introduces dimensional asymmetries
($\delta \sim 1\%$) that break the
$(1{,}2)/(2{,}1)$ modal degeneracy
by $\approx \SI{0.47}{\mega\hertz}$ (Sec.~C).
Degeneracy is restored via electrostatic spring-softening:
a DC bias $V_{\mathrm{DC}} = \SI{21.6}{\volt}$ across the
$d = \SI{200}{\nano\meter}$ vacuum gap selectively softens the stiffer
mode~\cite{barfuss2019_prb}, with a robust $\SI{13.4}{\volt}$ safety
margin below the conservative design voltage
$V_{\mathrm{bd}}=\SI{35}{\volt}$ informed by nanoscale vacuum-gap
breakdown and MEMS field-emission studies~\cite{wang2022_vacuum_breakdown,michalas2012_mems_field_emission}:
\begin{equation}
  \omega(V_{\mathrm{DC}}) = \omega_0
  \sqrt{1 - \frac{\varepsilon_0 A_{\mathrm{eff}}\, V_{\mathrm{DC}}^2}
  {\kmech\, d^3}} \,,
  \label{eq:tuning_sm}
\end{equation}
where $\varepsilon_0$ here denotes the vacuum permittivity.
The full-area electrode is mandatory:
a $50\%$-coverage design requires $\SI{30.5}{\volt}$, impinging on the
design-voltage margin.

\subsection{Voltage budget}
The voltage budget (Table~\ref{tab:voltage_sm}) demonstrates that the
system is \textit{displacement-limited}, not voltage-limited.

\begin{table}[htbp]
\caption{\label{tab:voltage_sm}Voltage budget for the simply-supported
membrane with full-area electrode and $\delta = 1\%$ asymmetry.}
\begin{ruledtabular}
\begin{tabular}{lrl}
  Component & Value & \\
  \hline
  Breakdown limit $V_{\mathrm{bd}}$ & $35.0$ & V \\
  DC tuning $V_{\mathrm{DC}}$ & $21.6$ & V \\
  Safety margin & $3.0$ & V \\
  AC headroom $V_{\mathrm{AC}}^{\max}$ & $10.4$ & V \\
  Required $V_{\mathrm{AC}}$ at $\Omm = \SI{2.22}{\mega\hertz}$
    & $\sim 5$ & mV \\
\end{tabular}
\end{ruledtabular}
\end{table}

\subsection{NV center requirements}
The simulations assume a single NV center at depth
$d_{\mathrm{NV}} = \SI{20}{\nano\meter}$ below the membrane surface,
achievable with standard nitrogen ion
implantation~\cite{pezzagna2010_implant}.
Coherence parameters:
$T_1 = \SI{1}{\milli\second}$~\cite{jarmola2012_t1},
$T_{1\rho} = \SI{500}{\micro\second}$~\cite{bar-gill2013_t1rho}
(isotopically purified ${}^{12}$C diamond with enrichment
${>}99.9\%$~\cite{balasubramanian2009_isotope,ishikawa2012_isotope}),
and $T_2^* \sim \SI{10}{\micro\second}$.

\section{Extended parametric studies}
\label{sec:extended_sweeps_sm}
\label{app:extended_sweeps_sm}

This section presents the full parametric sweeps not included in the
main text.  All simulations use the composite NGQC + SATD protocol
with 200 surface-noise trajectories unless otherwise noted.

\subsection{Fidelity vs.\ mechanical Rabi frequency (Sweep~1)}
Sweep~1 evaluates fidelity versus mechanical Rabi frequency
$\Omm = 0.1$--$\SI{2.5}{\mega\hertz}$ under worst-case noise
($\tauc = \SI{10}{\nano\second}$).
With DRAG active, the fidelity is nearly flat:
$F = 99.40\%$--$99.49\%$, a variation of $0.09\%$
[Fig.~\ref{fig:hardware_sm}(a)], confirming that counter-diabatic
correction decouples the fidelity from the adiabatic parameter.

\subsection{Frequency detuning tolerance (Sweep~3)}
Sweep~3 scans residual frequency detuning
$\Delta f = 0$--$\SI{500}{\kilo\hertz}$
[Fig.~\ref{fig:hardware_sm}(b)]: fidelity is flat to $0.03\%$,
confirming that the DC tuning mechanism does not require extreme
precision.

\subsection{Boundary-condition sensitivity (Sweep~4)}
Sweep~4 repeats Sweep~1 under clamped boundary conditions
($\kmech^{\mathrm{clamp}} = \SI{9530}{\newton\per\meter}$,
$\Omm^{\max} = \SI{1.004}{\mega\hertz}$;
Fig.~\ref{fig:hardware_sm}(d)):
with DRAG active, the clamped sweep remains within
$99.43\%$--$99.47\%$, about $0.02$--$0.06$ percentage points below
the simply-supported optimum.
Sweep~4.1 confirms that without DRAG, the SS membrane shows
oscillatory fidelity ($F = 8\%$--$99\%$) while the clamped
membrane suffers catastrophic monotonic decay from $96\%$ to
$28\%$; counter-diabatic correction is therefore a non-negotiable
prerequisite for clamped geometries
[Fig.~\ref{fig:hardware_sm}(e,\,f)].

\begin{figure*}[t]
  \centering
  \includegraphics[width=\textwidth]{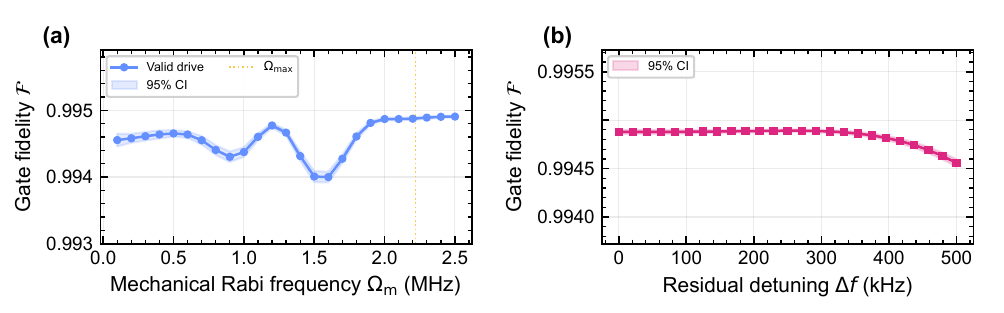}\\[4pt]
  \includegraphics[width=0.48\textwidth]{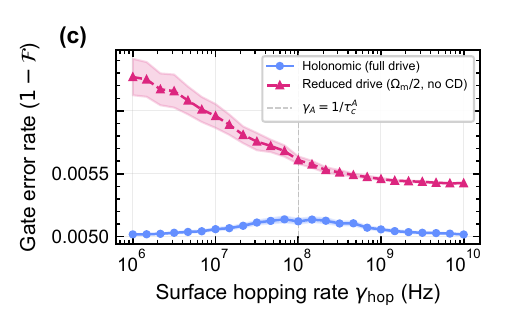}%
  \hfill
  \includegraphics[width=0.48\textwidth]{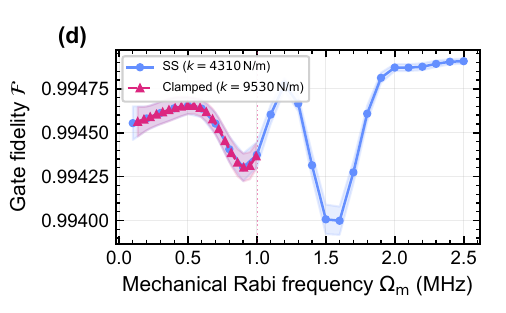}\\[4pt]
  \includegraphics[width=\textwidth]{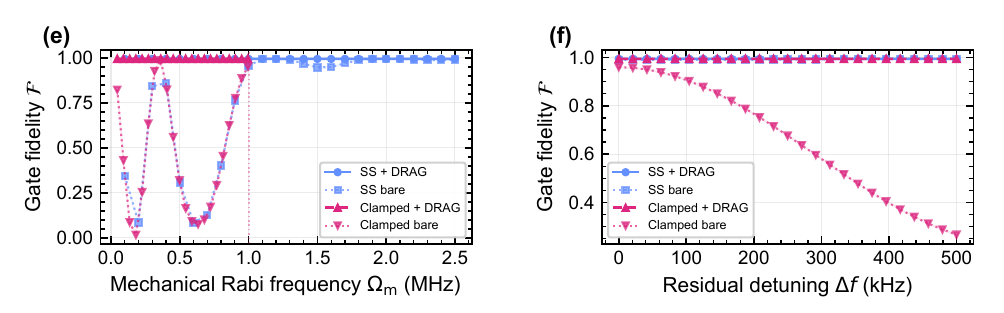}
  \caption{\label{fig:hardware_sm}%
  Hardware and noise tolerance characterization.
  (a,\,b)~Sweep~1, fidelity vs.\ $\Omm$;
  Sweep~3, fidelity vs.\ detuning $\Delta f$.
  (c)~Sweep~2, holonomic (blue) vs.\ Rabi (pink) gate
  error vs.\ surface hopping rate; the holonomic gate forms a
  noise-immune plateau, $7\times$ more stable and consistently
  lower.
  (d)~Sweep~4, SS vs.\ clamped boundary conditions under DRAG.
  (e,\,f)~Sweep~4.1, DRAG vs.\ bare adiabatic for
  headroom and tuning tolerance; without DRAG both membranes show
  oscillatory or monotonic fidelity collapse.
  }
\end{figure*}

\subsection{GHz HBAR extension (Sweep~9)}
Sweep~9 extends the gate-time analysis to GHz bulk acoustic
resonator Rabi frequencies (Appendix~\ref{app:ghz}):
conservative ($\Omm = \SI{2.83}{\kilo\hertz}$),
moderate ($\SI{28.3}{\kilo\hertz}$),
and optimistic ($\SI{141.5}{\kilo\hertz}$).
SATD eliminates the adiabatic-speed constraint regardless
of~$\Omm$: the process fidelity is $\Favg=99.50$--$99.85\%$ across
the three scenarios at $\tgate = \SI{2}{\micro\second}$
(Fig.~\ref{fig:ghz_gatetime_sm}).

\begin{figure}[!htbp]
  \centering
  \includegraphics[width=\columnwidth]{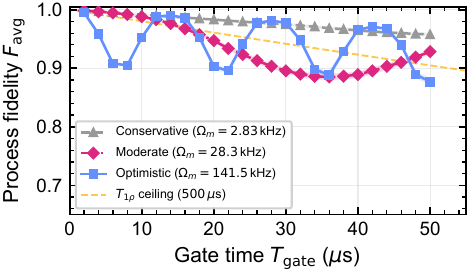}
  \caption{\label{fig:ghz_gatetime_sm}%
  GHz HBAR gate-time compression (Sweep~9):
  $\Favg$ versus $\tgate$ for three
  achievable strain amplitudes (Table~\ref{tab:ghz_scenarios});
  the dashed curve is the $T_{1\rho}$ ceiling.
  SATD eliminates the adiabatic-speed constraint at all $\Omm$.
  }
\end{figure}

\subsection{Quadrature drive robustness (Sweep~10)}
Sweep~10 scans amplitude ratio $r \in [0.90,\,1.10]$ and phase error
$\delta\varphi \in [-10^\circ,\,{+}10^\circ]$ at
$\tgate = \SI{1.833}{\micro\second}$ (Fig.~\ref{fig:ellipticity_sm}).
The gate maintains $\Favg \geq 99.5\%$ over the central plateau
($|r - 1| \lesssim 4\%$, $|\delta\varphi| \lesssim 4^\circ$).
Over the broader $|r-1|\leq 6\%$,
$|\delta\varphi|\leq 6^\circ$ box the minimum sampled value is
$99.32\%$, and the full scan degrades gracefully to $98.0\%$ at
the worst corner.

\begin{figure}[ht]
  \centering
  \includegraphics[width=0.95\columnwidth]{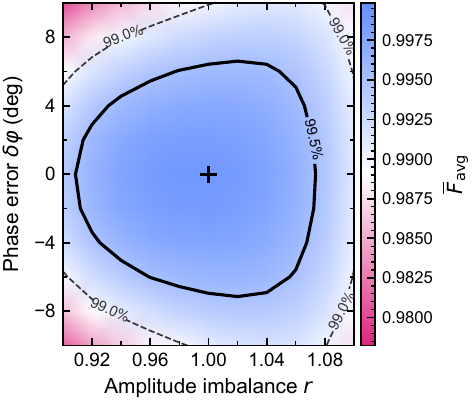}
  \caption{\label{fig:ellipticity_sm}Quadrature drive robustness
  (Sweep~10).  Process-level average gate fidelity $\Favg$ as a
  function of amplitude imbalance $r$ and phase error
  $\delta\varphi$, 200 trajectories per point.
  The solid black contour marks
  $\Favg = 99.5\%$; dashed and dotted contours indicate $99.0\%$
  and $99.9\%$, respectively.
  The gate maintains $\Favg \geq 99.5\%$ over the central plateau
  ($|r - 1| \lesssim 4\%$, $|\delta\varphi| \lesssim 4^\circ$);
  the broader $6\%$--$6^\circ$ box has a minimum sampled value
  of $99.32\%$.}
\end{figure}

\subsection{Floquet multitone validation (Sweep~12)}
\label{sec:floquet_multitone_sm}

To test the rotating-wave and multitone assumptions directly, we added
a Floquet-corrected lab-frame validation beyond the ten core
parametric sweeps.  It is retained in the validation ledger as
Sweep~12 rather than counted among Sweeps~1--10.  The calculation
propagates the Regime-A unitary using the same two $\Lambda$-leg
envelopes, the resonant DQ SATD envelope, the off-resonant DQ Stark
term, and the dynamic DC compensation waveform.  The leading
counter-rotating Bloch--Siegert/Floquet shifts from the
single-quantum tones and the DQ SATD tone are retained.  We also apply
a first-order transducer bandwidth filter, colored phase noise, and
arm-amplitude imbalance.
The reference is the nominal rotating-frame unitary, so the values in
Table~\ref{tab:floquet_multitone_sm} quantify the additional error
from effects omitted in the RWA simulation.

\begin{table*}[t]
\caption{\label{tab:floquet_multitone_sm}Floquet/lab-frame multitone
validation at $\tgate=\SI{1.833}{\micro\second}$ and
$\Omm=\SI{2.22}{\mega\hertz}$.  $F_{\rm avg}$ is computed relative to
the nominal rotating-frame unitary on the computational subspace.}
\begin{ruledtabular}
\begin{tabular}{lcc}
  Case & $F_{\rm avg}$ vs.\ RWA & Added leakage \\
  \hline
  Counter-rotating terms only & $0.999976$ & $3.7\times10^{-6}$ \\
  \SI{200}{\mega\hertz} BW, $0.2^\circ$ phase noise, $0.2\%$ imbalance
    & $0.999899$ & $6.0\times10^{-5}$ \\
  \SI{50}{\mega\hertz} BW, $1^\circ$ phase noise, $1\%$ imbalance
    & $0.998247$ & $1.4\times10^{-3}$ \\
  Unpredistorted $Q=10^4$ HBAR BW ($\SI{0.287}{\mega\hertz}$)
    & $0.414075$ & $4.2\times10^{-1}$ \\
\end{tabular}
\end{ruledtabular}
\end{table*}

The pure counter-rotating correction is therefore small
($1-F_{\rm avg}=2.4\times10^{-5}$), and the high-strain benchmark
remains within $1.0\times10^{-4}$ of the RWA unitary when the control
envelopes are delivered with \SI{200}{\mega\hertz} bandwidth and modest
phase/amplitude errors.  The bandwidth scan gives
$1-F_{\rm avg}=4.2\times10^{-5}$, $1.7\times10^{-4}$, and
$8.0\times10^{-4}$ for ideal controls filtered at 200, 100, and
\SI{50}{\mega\hertz}, respectively.  Conversely, a bare
$Q=10^4$ GHz-HBAR linewidth cannot track the \SI{1.833}{\micro\second}
benchmark envelope without predistortion or a slower waveform.  This is
why the HBAR Regime-B projection is separated from the high-strain
benchmark and why HBAR-specific envelope tracking is treated as an
implementation requirement rather than a group-theoretic consequence.
For the decoder-facing erasure analysis, this same residual is treated
as an explicit contribution to the imposed $p_{XY}$ floor.  Thus the
HBAR check tests transferability of the control package; it does not
replace the Regime-A open-system channel extraction.

\subsection{Matched-environment 3C-SiC C3v benchmark}
\label{sec:sic_regime_a_sm}
\label{app:sic_regime_a_sm}

The 3C-SiC calculation is a controlled platform-design check, not a
measured-device forecast.  We keep the Regime-A SATD echo-lune
trajectory, surface-bath traces, $T_1$, $T_{1\rho}$, erasure
bookkeeping, and trajectory seeds fixed at the NV channel-benchmark values,
and change only the $C_{3v}$ spin-strain parameters entering the DQ
Stark scale.  This isolates the effect of smaller
$|h_{16}/h_{26}|$.
Figure~\ref{fig:sic_regime_a} summarizes the matched-environment
comparison.

\begin{figure*}[t]
  \centering
  \includegraphics[width=0.86\textwidth]{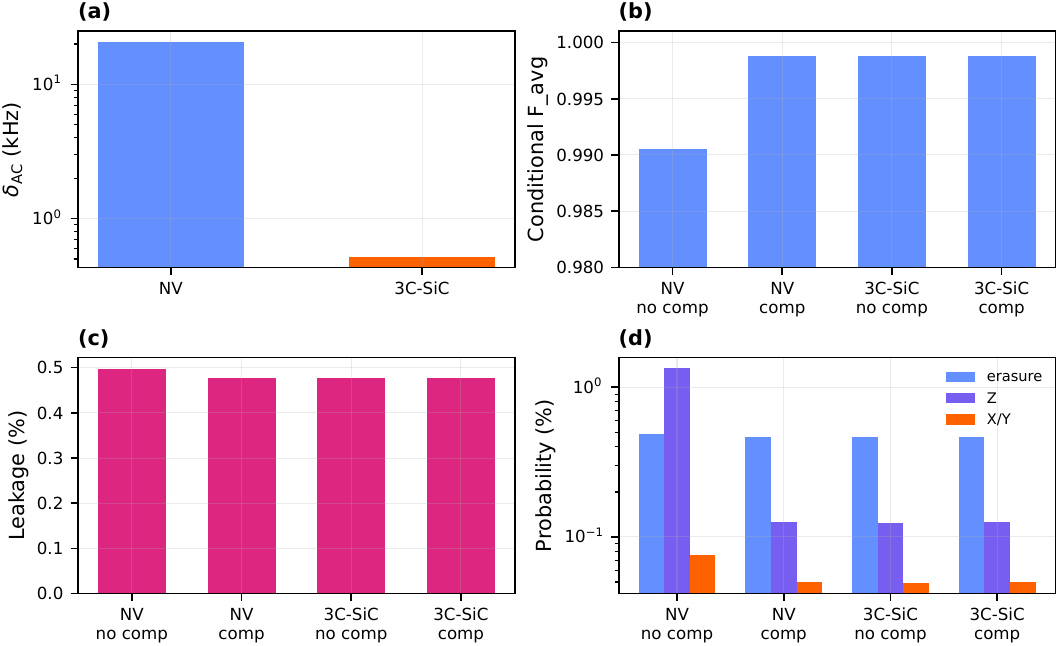}
  \caption{\label{fig:sic_regime_a}%
  Matched-environment 3C-SiC $C_{3v}$ comparison.  (a)~At fixed
  $\Omm=\SI{2.22}{\mega\hertz}$, the assumed 3C-SiC divacancy
  parameters reduce the parasitic DQ Stark scale by $39.76\times$.
  (b--d)~Using matched surface-bath traces and relaxation settings,
  uncompensated 3C-SiC lies on the compensated baseline, whereas
  uncompensated NV shows a visible Stark penalty.}
\end{figure*}

The platform inputs are:
\begin{center}
\scriptsize
\begin{tabular}{lcccc}
  Platform & $|h_{16}/h_{26}|$ & $\epsilon_{\rm SQ}$ &
  $\epsilon_{\rm DQ\mbox{-}CD}$ & $\delta_{\mathrm{AC}}$ \\
  \hline
  NV reference & $6.947$ & $7.845\times10^{-4}$ &
  $8.718\times10^{-5}$ & $20.718$ kHz \\
  3C-SiC assumed & $0.750$ & $1.233\times10^{-3}$ &
  $1.270\times10^{-3}$ & $0.521$ kHz
\end{tabular}
\end{center}

The full open-system rows use 96 matched stochastic traces and a
1600-step midpoint propagator:
\begin{table*}[t]
\caption{\label{tab:sic_regime_a_sm}Matched-environment
3C-SiC matched-environment benchmark.  $F_{\mathrm{cond}}$ is conditioned
on global survival; $F_{\mathrm{eff}}$ counts leakage as loss.}
\begin{ruledtabular}
\begin{tabular}{lcccccc}
  Scenario & $F_{\mathrm{cond}}$ & $F_{\mathrm{eff}}$ & Leakage &
  $p_{\mathrm{era}}$ & $p_Z$ & $p_{XY}$ \\
  \hline
  NV no comp & $0.990568$ & $0.985641$ & $0.4974\%$ &
  $0.4850\%$ & $1.3394\%$ & $0.0753\%$ \\
  NV comp & $0.998826$ & $0.994063$ & $0.4769\%$ &
  $0.4649\%$ & $0.1261\%$ & $0.0500\%$ \\
  3C-SiC no comp & $0.998846$ & $0.994079$ & $0.4773\%$ &
  $0.4654\%$ & $0.1237\%$ & $0.0493\%$ \\
  3C-SiC comp & $0.998826$ & $0.994063$ & $0.4769\%$ &
  $0.4649\%$ & $0.1261\%$ & $0.0500\%$
\end{tabular}
\end{ruledtabular}
\end{table*}

The differential comparison makes the Stark effect explicit:
\begin{center}
\footnotesize
\begin{tabular}{lcc}
  Comparison & $\Delta F_{\mathrm{cond}}$ & $\Delta p_Z$ \\
  \hline
  NV no-comp $-$ NV comp & $-8.258\times10^{-3}$ & $+1.2133$ pp \\
  3C no-comp $-$ 3C comp & $+2.046\times10^{-5}$ & $-0.0024$ pp \\
  3C no-comp $-$ NV no-comp & $+8.278\times10^{-3}$ & $-1.2158$ pp
\end{tabular}
\end{center}
The coherent isolation rows, with surface bath and Lindblad collapse
turned off, show the same pattern:
\begin{center}
\footnotesize
\begin{tabular}{lccc}
  Scenario & $F_{\mathrm{cond}}$ & $F_{\mathrm{eff}}$ & Leakage \\
  \hline
  NV Stark & $0.99151421$ & $0.98752089$ & $0.4027\%$ \\
  3C-SiC Stark & $0.99997888$ & $0.99615165$ & $0.3827\%$ \\
  3C-SiC no Stark & $0.99996251$ & $0.99614006$ & $0.3823\%$
\end{tabular}
\end{center}

The correct conclusion is a platform-design rule: $C_{3v}$ hosts with
small $|h_{16}/h_{26}|$ can preserve the same $\Lambda$-sector
symmetry-to-channel mechanism while suppressing the parasitic DQ Stark
channel before active compensation.  The calculation does not prove a
deployable 3C-SiC architecture; it still requires platform-specific
surface noise, relaxation, strain limits, mechanical mode shapes,
optomechanical coupling, and auxiliary-state detection.
The $p_Z$ and $p_{XY}$ values in Table~\ref{tab:sic_regime_a_sm} are
conditional Pauli-twirl diagnostics of this compact projected lossy
superoperator and should not be substituted for the final Regime-A
code-capacity channel rates used in Appendix~\ref{app:qec_montecarlo}.

\subsection{Error reduction summary}
Figure~\ref{fig:error_reduction_sm} compares the baseline DRAG
protocol with the composite NGQC + SATD protocol across the full
gate-time range.
The high-strain benchmark composite protocol achieves
$\Favg = 99.88\%$ at $\tgate = \SI{1.833}{\micro\second}$, a
substantial error reduction over the baseline ceiling, confirming that
the three mechanisms
(singularity neutralization, exact counter-diabatic driving, and
gate-time compression) act in concert.

\begin{figure}[ht]
  \centering
  \includegraphics[width=\columnwidth]{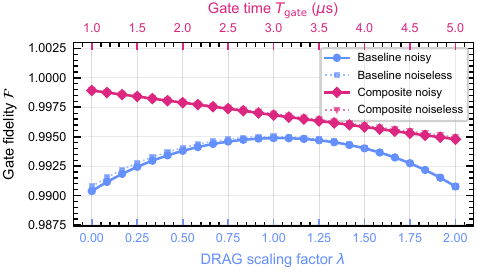}
  \caption{\label{fig:error_reduction_sm}Error reduction summary.
  Fidelity comparison of the baseline DRAG protocol (blue) and the
  composite NGQC + SATD protocol (pink) as a function of
  DRAG scaling factor $\lambda$ (bottom axis) and gate time
  $\tgate$ (top axis).
  The high-strain benchmark composite achieves
  $\Favg = 99.88\%$ at $\tgate = \SI{1.833}{\micro\second}$, a
  substantial error reduction over the baseline ceiling.}
\end{figure}

\section{Monte Carlo QEC simulation methods and full data}
\label{app:qec_montecarlo}

This section details the Monte Carlo QEC simulations summarized in
Sec.~VI\,G of the main text.  Each data qubit independently
experiences four error channels per round, with physical error rates
drawn from the Regime-A channel extraction in Sec.~VI\,F and the error
budget in Table~VII of the main text.  All simulations use $5000$
independent trials per data point with minimum-weight perfect matching
(MWPM) decoding via PyMatching~2~\cite{higgott2023_pymatching}.
The microscopic sector-to-channel origin of the extracted noise model
is verified independently in Appendix~\ref{app:noise_sectors} by
sector-injection diagnostics; this appendix uses the extracted
Regime-A channel as input and does not recompute microscopic dynamics.

\subsection{Noise model}

The Regime-A/SATD noise model is:
\begin{center}
\footnotesize
\setlength{\tabcolsep}{2pt}
\renewcommand{\arraystretch}{1.08}
\newcommand{\qectabcell}[2]{\parbox[t]{#1\columnwidth}{#2}}
\begin{tabular}{@{}l@{\hspace{0.5em}}l@{\hspace{0.5em}}l@{}}
  \qectabcell{0.20}{Channel} &
  \qectabcell{0.25}{Probability} &
  \qectabcell{0.47}{Origin} \\
  \hline
  \qectabcell{0.20}{Erasure} &
  \qectabcell{0.25}{$p_{\mathrm{era}} = 0.47\%$} &
  \qectabcell{0.47}{Detected $|0\rangle$ leakage:\\
  noise-mediated CD mismatch + $T_1$} \\
  \qectabcell{0.20}{$Z$ dephasing} &
  \qectabcell{0.25}{$p_Z = 0.168\%$} &
  \qectabcell{0.47}{$T_{1\rho}$ + $A_2$-sector residual} \\
  \qectabcell{0.20}{Depolarizing} &
  \qectabcell{0.25}{$p_{\mathrm{dep}} = 0.012\%$} &
  \qectabcell{0.47}{Undetected leakage} \\
  \qectabcell{0.20}{$X/Y$ flip} &
  \qectabcell{0.25}{numerical floor in the nominal extraction} &
  \qectabcell{0.47}{Echo-suppressed $E$-sector} \\
\end{tabular}
\end{center}
All rates are uniformly scaled by a common factor
$s \in \{1,\,2,\,5,\,10,\,20,\,50,\,75,\,100\}$ to probe
behavior near and above threshold.
Erasure events are injected by replacing the data qubit with a
random Pauli ($X$, $Y$, or $Z$ with equal probability) and
setting the decoder edge weight to zero (effective error
probability $0.5$), corresponding to the known-location
erasure model~\cite{stace2010_pra}.

\textit{Scope and scheduled stress diagnostic.}---The code-capacity
simulations isolate the response of CSS and XZZX decoders to the
extracted Regime-A biased-erasure channel.  Because the nominal
transverse component is a model-extracted numerical floor rather than
a symmetry constant, we also promote finite $p_{XY}$ floors to an
explicit decoder stress axis.
To test the first circuit layer without introducing a hardware-specific
optical/mechanical layout, we run a scheduled two-sector
detector-model diagnostic.  This diagnostic includes repeated syndrome
rounds, erasure-aware weights, measurement/reset faults, explicit
transverse $X/Y$ faults, finite erasure-detection efficiency, leakage
persistence, delayed flags, and local crosstalk.
The nominal code-capacity point gives the $d=9$ / $64\%$ model-channel
proxy; finite transverse floors $p_{XY}=10^{-6}$--$10^{-4}$ move the
estimate to the $d=11$ / $46.2\%$ envelope; $p_{XY}=10^{-3}$ moves the
code-capacity estimate to $d=13$ / $24.9\%$ and the scheduled
detector-model proxy to $d=11$; combined missed-erasure and
finite-$p_{XY}$ stress remains in the $d=11/d=13$ envelope; and
$2\times$ local crosstalk remains above target even at $d=15$.
Thus the scheduled diagnostic supports the biased-erasure mechanism
and identifies crosstalk as the primary hardware-specific target for a
future calibrated detector-error model.
MWPM is not the optimal decoder for biased noise; bias-tailored
decoders (e.g.\ union-find weighted, tensor-network) would further
improve the XZZX advantage, so the numerical distances reported here
should be read as model-channel and stress-diagnostic estimates rather
than full architecture-level thresholds.

\begin{table*}[t]
\caption{\label{tab:qec_scheduled_stress}
Compact scheduled two-sector XZZX stress diagnostic for the extracted
Regime-A biased-erasure channel.  The table summarizes the first
circuit-layer validation envelope used in the main text.}
\begin{ruledtabular}
\begin{tabular}{lll}
Diagnostic & Result & Interpretation \\
\hline
Code-capacity nominal XZZX &
$d=9$, $64\%$ saving &
Nominal model-channel estimate \\
Scheduled nominal, flags &
$p_L(d=9)=1.00\times10^{-3}$ &
First circuit layer preserves nominal proxy \\
$p_{XY}=10^{-3}$ &
Code-capacity $d=13$; scheduled proxy $d=11$ &
High transverse floor sets stress envelope \\
$\eta_{\rm det}=0.9$ plus finite $p_{XY}$ &
$d=11/d=13$ regime &
Missed erasures reduce margin \\
$2\times$ crosstalk &
$p_L(d=15)=2.6\times10^{-3}$ &
Unresolved hardware-specific risk \\
\end{tabular}
\end{ruledtabular}
\end{table*}

\subsection{Code constructions}

\paragraph{CSS toric code.}
Standard $d \times d$ periodic lattice with $X$-stabilisers on
faces and $Z$-stabilisers on vertices, encoding one logical
qubit in $2d^2$ data qubits.
Distances tested: $d \in \{3,\,5,\,7,\,9,\,11\}$.

\paragraph{XZZX toric code.}
Same lattice with Hadamard rotations applied to vertical edges,
converting each $XXXX$ face stabiliser to $XZZX$ and each
$ZZZZ$ vertex stabiliser to
$ZXXZ$~\cite{bonilla_ataides2021_xzzx}.
Distances tested: $d \in \{3,\,5,\,7,\,9,\,11\}$.

\paragraph{XZZX rectangular planar code.}
Open-boundary rotated planar code with column distance $d_c$
(protecting the dominant $Z$-error direction) and row distance
$d_r$ (protecting the rare $X$-error direction).
Tested configurations: $d_c \in \{7,\,9,\,11\}$,
$d_r \in \{3,\,5,\,7,\,9\}$, at scales
$s \in \{1,\,10,\,50,\,100\}$.
Face-type-aware Hadamard assignments ensure every column of the
check matrix has exactly $2$ (boundary) or $4$ (bulk)
nonzero entries.

Logical error is detected by computing the parity of the
recovery operator along a nontrivial homology cycle (toric)
or logical representative (planar).

\subsection{Full data: toric codes}

Table~\ref{tab:qec_css_toric} gives CSS toric code results.
$p_L$ increases with $d$ for $s \geq 5$, confirming
above-threshold operation.

\begin{table*}[t]
\caption{\label{tab:qec_css_toric}CSS toric code logical error
rate $p_L$ (failures / 5000 trials) as a function of code
distance $d$ and noise scale factor $s$.}
\begin{ruledtabular}
\begin{tabular}{lcccccccc}
  $d$ & $s{=}1$ & $s{=}2$ & $s{=}5$ & $s{=}10$ & $s{=}20$
  & $s{=}50$ & $s{=}75$ & $s{=}100$ \\
  \hline
  $3$  & $4.0\!\times\!10^{-4}$ & $8.0\!\times\!10^{-4}$
       & $4.4\!\times\!10^{-3}$ & $1.42\!\times\!10^{-2}$
       & $7.38\!\times\!10^{-2}$ & $0.370$ & $0.617$ & $0.778$ \\
  $5$  & $0$      & $8.0\!\times\!10^{-4}$
       & $6.6\!\times\!10^{-3}$ & $1.48\!\times\!10^{-2}$
       & $8.52\!\times\!10^{-2}$ & $0.476$ & $0.699$ & $0.836$ \\
  $7$  & $0$      & $6.0\!\times\!10^{-4}$
       & $6.6\!\times\!10^{-3}$ & $2.62\!\times\!10^{-2}$
       & $0.116$ & $0.564$ & $0.736$ & $0.840$ \\
  $9$  & $0$      & $1.4\!\times\!10^{-3}$
       & $5.6\!\times\!10^{-3}$ & $3.24\!\times\!10^{-2}$
       & $0.143$ & $0.624$ & $0.747$ & $0.843$ \\
  $11$ & $6.0\!\times\!10^{-4}$ & $2.2\!\times\!10^{-3}$
       & $1.06\!\times\!10^{-2}$ & $4.06\!\times\!10^{-2}$
       & $0.177$ & $0.671$ & $0.739$ & $0.841$ \\
\end{tabular}
\end{ruledtabular}
\end{table*}

Table~\ref{tab:qec_xzzx_toric} gives XZZX toric code results.
$p_L$ decreases monotonically with $d$ at every tested scale
including $s = 100$, confirming deeply sub-threshold operation.

\begin{table*}[t]
\caption{\label{tab:qec_xzzx_toric}XZZX toric code logical error
rate $p_L$ (failures / 5000 trials) as a function of code
distance $d$ and noise scale factor $s$.}
\begin{ruledtabular}
\begin{tabular}{lcccccccc}
  $d$ & $s{=}1$ & $s{=}2$ & $s{=}5$ & $s{=}10$ & $s{=}20$
  & $s{=}50$ & $s{=}75$ & $s{=}100$ \\
  \hline
  $3$  & $2.0\!\times\!10^{-4}$ & $0$
       & $2.8\!\times\!10^{-3}$ & $6.8\!\times\!10^{-3}$
       & $2.62\!\times\!10^{-2}$ & $0.149$ & $0.303$ & $0.471$ \\
  $5$  & $0$      & $0$
       & $4.0\!\times\!10^{-4}$ & $2.0\!\times\!10^{-4}$
       & $3.0\!\times\!10^{-3}$ & $4.98\!\times\!10^{-2}$
       & $0.146$ & $0.300$ \\
  $7$  & $0$      & $0$      & $0$      & $0$
       & $4.0\!\times\!10^{-4}$ & $2.42\!\times\!10^{-2}$
       & $7.48\!\times\!10^{-2}$ & $0.194$ \\
  $9$  & $0$      & $0$      & $0$      & $0$
       & $0$      & $7.0\!\times\!10^{-3}$
       & $4.40\!\times\!10^{-2}$ & $0.137$ \\
  $11$ & $0$      & $0$      & $0$      & $0$
       & $0$      & $3.4\!\times\!10^{-3}$
       & $2.34\!\times\!10^{-2}$ & $9.32\!\times\!10^{-2}$ \\
\end{tabular}
\end{ruledtabular}
\end{table*}

The CSS-to-XZZX advantage ratio at $d = 11$ grows from a
factor of $\sim\!200$ at $s = 50$ ($p_L^{\mathrm{CSS}} /
p_L^{\mathrm{XZZX}} = 0.671 / 0.0034 = 197$) to $9.0$ at
$s = 100$.
The advantage increases with $d$ at every fixed scale,
confirming that the two codes lie on opposite sides of a
threshold boundary.

\subsection{Full data: rectangular XZZX planar codes}

Table~\ref{tab:qec_rect} gives rectangular XZZX planar code
results.
Every asymmetric code ($d_r < d_c$) achieves
\emph{zero} logical failures in 5000 trials at all tested
scales including $s = 100$.
Only the square codes ($7 \times 7$ and $9 \times 9$)
exhibit nonzero failures, and only at $s \geq 50$.
The Wilson score $95\%$ confidence-level upper bound for
a zero-failure result is
$p_L < 7.7 \times 10^{-4}$.
Thus the rectangular-code data support only a finite-statistics
code-capacity bound at the $10^{-4}$ level; they are not extrapolated
to $p_L=10^{-10}$ and are not used for the headline $64\%$ square-code
overhead estimate.

\begin{table*}[t]
\caption{\label{tab:qec_rect}Rectangular XZZX planar code
logical error rate $p_L$ (failures / 5000 trials).
All asymmetric configurations ($d_r < d_c$) show zero
failures at all scales; only the square codes ($d_r = d_c$,
bold) have nonzero failures at high scale.}
\begin{ruledtabular}
\begin{tabular}{lccccc}
  $d_r \!\times\! d_c$ & Qubits & $s{=}1$ & $s{=}10$
  & $s{=}50$ & $s{=}100$ \\
  \hline
  $3 \times 7$   & $21$  & $0$ & $0$ & $0$ & $0$ \\
  $5 \times 7$   & $35$  & $0$ & $0$ & $0$ & $0$ \\
  $\mathbf{7 \times 7}$ & $\mathbf{49}$  & $0$ & $0$
  & $\mathbf{7.8\!\times\!10^{-3}}$ & $\mathbf{8.84\!\times\!10^{-2}}$ \\
  $3 \times 9$   & $27$  & $0$ & $0$ & $0$ & $0$ \\
  $5 \times 9$   & $45$  & $0$ & $0$ & $0$ & $0$ \\
  $7 \times 9$   & $63$  & $0$ & $0$ & $0$ & $0$ \\
  $\mathbf{9 \times 9}$ & $\mathbf{81}$  & $0$ & $0$
  & $\mathbf{4.6\!\times\!10^{-3}}$ & $\mathbf{6.70\!\times\!10^{-2}}$ \\
  $3 \times 11$  & $33$  & $0$ & $0$ & $0$ & $0$ \\
  $5 \times 11$  & $55$  & $0$ & $0$ & $0$ & $0$ \\
  $7 \times 11$  & $77$  & $0$ & $0$ & $0$ & $0$ \\
  $9 \times 11$  & $99$  & $0$ & $0$ & $0$ & $0$ \\
\end{tabular}
\end{ruledtabular}
\end{table*}
\clearpage

\subsection{Conservative erasure-only baseline}

As a conservative comparison, we also analyze the same leakage
conversion mechanism without exploiting the strong $Z$ bias of the
Regime-A channel.  All non-erasure faults are grouped into an
isotropic depolarizing residual, so the only structured resource
retained by the decoder is known-location erasure.
The resulting erasure-only baseline is shown in
Fig.~\ref{fig:erasure_baseline}.

\begin{figure}[ht]
  \centering
  \includegraphics[width=\columnwidth]{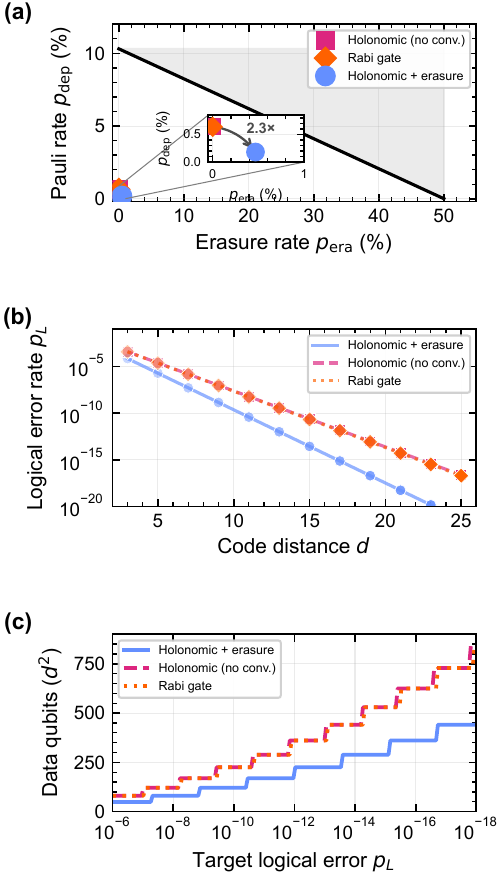}
  \caption{\label{fig:erasure_baseline}%
  Conservative erasure-only analysis (isotropic depolarizing
  non-erasure residual, no Z-bias exploitation).
  \textbf{(a)}~Code-capacity threshold boundary in the
  $p_{\mathrm{era}}$--$p_{\mathrm{dep}}$ plane.
  Blue circle: holonomic gate with erasure conversion
  ($p_{\mathrm{eff}} = 0.28\%$); pink square: without
  conversion; orange diamond: Rabi gate.
  \textbf{(b)}~$p_L$ versus code distance $d$
  for the conservative erasure-only model of Sec.~VI\,F.
  \textbf{(c)}~Physical qubit overhead versus target $p_L$.
  This baseline is superseded by the bias-aware analysis
  of Sec.~VI\,G and Fig.~6 of the
  main text.}
\end{figure}

\subsection{Transverse-floor validation envelope}
\label{app:qec_sensitivity}

The square-code overhead quoted in the main text is a nominal
code-capacity estimate for the extracted Regime-A channel.  Because
the nominal $p_{XY}$ component is a model-extracted numerical floor
rather than a group-theoretic constant, we propagate explicit $X/Y$
floors and reduced erasure-detection efficiency through the same
fit-extrapolated distance model used for Table~XI of the main text.
This diagnostic is not a circuit-level simulation; it asks how the
nominal $d=9$ XZZX estimate moves when transverse floors or
missed-erasure channels are added.

The extrapolation is calibrated to reproduce the two main-text
code-capacity anchors: the nominal XZZX square-code point
($d=9$ at $p_L=10^{-10}$) and the conservative erasure-only CSS
baseline ($d=11$).  A finite $p_{XY}$ floor interpolates the effective
scaling from the strongly biased XZZX limit toward the erasure-only
baseline, while missed erasures are treated as unheralded residual
faults, consistent with surface-code analyses of imperfect erasure
checks~\cite{chang2024_imperfect_erasure}.  Figure~\ref{fig:qec_sensitivity} shows the resulting
required square-code distance and overhead saving relative to the
Rabi CSS baseline ($d=15$).

At the nominal detection efficiency
$\eta_{\rm det}=97.5\%$, the fit-extrapolated XZZX distance is
$d=9$ for the nominal extracted floor.  Adding floors of
$p_{XY}=10^{-6}$--$10^{-4}$ increases the estimate to $d=11$,
giving a $46.2\%$ saving relative to the $d=15$ Rabi/CSS baseline,
while $p_{XY}=10^{-3}$ gives $d=13$ and a $24.9\%$ saving.
Reducing $\eta_{\rm det}$ from $97.5\%$ to $95\%$ similarly moves the
nominal-floor point from $d=9$ to $d=11$.  Thus the XZZX advantage
degrades continuously under finite transverse floors; the headline
$64\%$ reduction should be read as a code-capacity result for the
nominal extracted channel.

\begin{figure}[ht]
  \centering
  \includegraphics[width=\columnwidth]{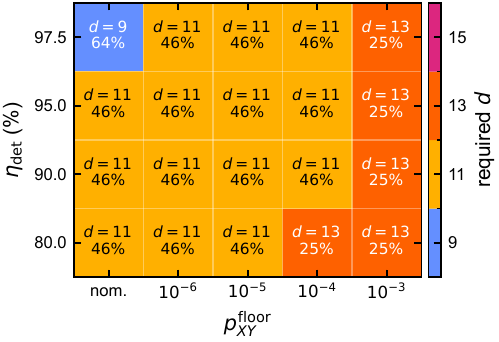}
  \caption{\label{fig:qec_sensitivity}%
  Transverse-floor validation envelope for the square-code XZZX
  overhead estimate.  Each cell gives the required fit-extrapolated
  XZZX distance for $p_L=10^{-10}$ and, underneath, the corresponding
  data-qubit saving relative to the Rabi CSS baseline with $d=15$.
  The axes scan the imposed transverse floor $p_{XY}^{\rm floor}$ and
  erasure-detection efficiency $\eta_{\rm det}$.
  The plot uses the same code-capacity scaling model as
  the main-text QEC table and is intended as a model-channel
  sensitivity test, not a circuit-level threshold calculation.}
\end{figure}

\subsection{Perturbative physical-error-to-\texorpdfstring{$p_{XY}$}{pXY} map}
\label{app:qec_physical_pxy_map}

To connect the decoder stress axis to calibration knobs, we construct
a compact perturbative transverse-floor map.  Each physical
imperfection is represented by a normalized small parameter
$\epsilon$ and converted to an effective transverse probability
$p_{XY}=p_{XY}^{\rm nom}+c\epsilon^2$.  This captures the expected
``first-order transverse amplitude, second-order probability''
scaling and defines calibration targets for the decoder envelope.  It
is a perturbative map, not a microscopic qutrit/open-system channel
extraction.

The map includes residual $E$-sector strain noise, quadrature
amplitude imbalance, quadrature phase error, transverse magnetic-field
misalignment, transverse hyperfine components, SATD amplitude and
phase errors, and HBAR envelope-filtering residuals.
Figure~\ref{fig:qec_physical_pxy_map} shows the resulting effective
$p_{XY}$ values.  For example, a $3\%$ quadrature amplitude imbalance,
a $3^\circ$ quadrature phase error, a $1\%$ SATD amplitude error, or
a normalized HBAR filtering residual of $10^{-2}$ reaches the
$10^{-4}$ transverse-floor envelope; stronger residuals approach the
$10^{-3}$ high-floor stress regime.

\begin{figure}[ht]
  \centering
  \includegraphics[width=\columnwidth]{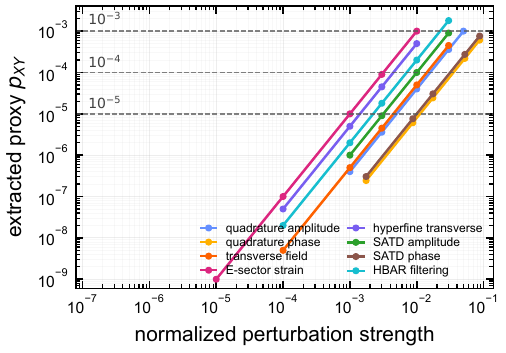}
  \caption{\label{fig:qec_physical_pxy_map}%
  Perturbative physical-error-to-$p_{XY}$ map.
  Representative transverse-error mechanisms are converted into an
  effective decoder stress floor using
  $p_{XY}=p_{XY}^{\rm nom}+c\epsilon^2$.
  Horizontal guide lines mark the $10^{-5}$, $10^{-4}$, and
  $10^{-3}$ transverse-floor envelopes used in the decoder analysis.
  The plot should be read as a calibration-target map rather than a
  microscopic channel extraction.}
\end{figure}

\FloatBarrier

\section{Single-shot bright-state compiler details}
\label{app:single_shot}

This appendix collects the details that are intentionally kept out of
the main text.  They support the Regime-D universal-gate validation and
are not used in the SATD Regime-A XZZX overhead estimate.

\subsection{Holonomy derivation}

Use the cycle-frequency Hamiltonian $K=H/h$ on
$\mathcal H_\Lambda=\mathrm{span}\{\ket{0_L},\ket{1_L},\ket a\}$:
\begin{align}
K(t)=&\;
\Delta(t)\ket a\bra a
\nonumber\\
&+\frac{1}{2}
\left[
\Omega_0(t)\ket a\bra{0_L}
+\Omega_1(t)\ket a\bra{1_L}
+{\rm h.c.}
\right].
\end{align}
The proportional-control command is
\begin{align}
\Omega_0(t)&=\Omega(t)\cos\alpha\cos\frac{\vartheta}{2},\\
\Omega_1(t)&=\Omega(t)\cos\alpha e^{-i\phi}
              \sin\frac{\vartheta}{2},\\
\Delta(t)&=\Omega(t)\sin\alpha .
\end{align}
With
\begin{align}
\ket{b_{\mathbf n}}&=
\cos\frac{\vartheta}{2}\ket{0_L}
+e^{i\phi}\sin\frac{\vartheta}{2}\ket{1_L},\\
\ket{d_{\mathbf n}}&=
\sin\frac{\vartheta}{2}\ket{0_L}
-e^{i\phi}\cos\frac{\vartheta}{2}\ket{1_L},
\end{align}
the Hamiltonian factorizes as
\begin{equation}
K(t)=\Omega(t)M_{\alpha,\mathbf n},
\end{equation}
where
\begin{equation}
M_{\alpha,\mathbf n}
=\sin\alpha\ket a\bra a
+\frac{\cos\alpha}{2}
\left(\ket a\bra{b_{\mathbf n}}+
\ket{b_{\mathbf n}}\bra a\right).
\end{equation}
In the ordered basis
$(\ket{b_{\mathbf n}},\ket a)$,
\begin{equation}
M_\alpha=
\begin{pmatrix}
0 & \cos\alpha/2\\
\cos\alpha/2 & \sin\alpha
\end{pmatrix},
\qquad
\lambda_\pm=\frac{\sin\alpha\pm1}{2}.
\end{equation}
The dark state is exactly decoupled and
$[K(t),K(t')]=0$.  When
\begin{equation}
\int_0^T\Omega(t)\,dt=1,
\end{equation}
the bright/auxiliary block evolves as
$\exp[-i2\pi M_\alpha]=e^{-i\pi(1+\sin\alpha)}I$.
The induced logical gate is therefore
\begin{equation}
U_L(\mathbf n,\gamma)=
\ket{d_{\mathbf n}}\bra{d_{\mathbf n}}
+e^{-i\gamma}\ket{b_{\mathbf n}}\bra{b_{\mathbf n}},
\qquad
\gamma=\pi(1+\sin\alpha),
\end{equation}
or, up to a global phase,
$U_L\doteq\exp[-i\gamma\,\mathbf n\cdot\boldsymbol\sigma/2]$.

\subsection{Non-Abelian diagnostic}

The commutator follows directly from the Pauli representation:
\begin{equation}
[U(\mathbf n_1,\gamma_1),U(\mathbf n_2,\gamma_2)]
=
-2i\sin\frac{\gamma_1}{2}\sin\frac{\gamma_2}{2}
(\mathbf n_1\times\mathbf n_2)\cdot\boldsymbol\sigma .
\end{equation}
Thus nonparallel nontrivial pulses do not commute.  The numerical
diagnostic applies $X_{\pi/2}Z_{\pi/2}$ and
$Z_{\pi/2}X_{\pi/2}$ to the same test states.  Each ordered sequence
has unit fidelity to its own target in the ideal compiler, while the
two composed unitaries have average gate fidelity $0.5$ relative to
one another, Frobenius commutator norm $\sqrt2$, and operator
commutator norm $1$.

\subsection{Transfer matrix and detuning channel}

The exact theorem requires the physical controls to preserve a fixed
direction in command space,
\begin{equation}
\mathbf u(t)=(\Omega_0,\Omega_1,\Delta)^T=\Omega(t)\mathbf u_0 .
\end{equation}
We model the acoustic and detuning actuator by a baseband transfer
matrix,
\begin{equation}
\mathbf u_{\rm act}(\omega)=G(\omega)\mathbf u_{\rm cmd}(\omega).
\end{equation}
Common-envelope filtering is harmless when
\begin{equation}
G(\omega)\mathbf u_0\simeq g(\omega)\mathbf u_0
\end{equation}
over the pulse bandwidth.  Relative delay, arm imbalance, phase skew,
reflections, or detuning-channel mismatch rotate $\mathbf u_0$ and
therefore produce gate errors.  The compact design criterion used in
the Regime-D validation is $B\tgate\gtrsim10$, where $B$ is the usable
calibrated waveform bandwidth.  The baseline detuning implementation
is IQ-programmed rotating-frame detuning synchronized to the two
$\Lambda$-leg envelopes; Stark or DC detuning can be substituted only
after being included as the third measured waveform channel.

\subsection{Scope relative to the SATD stack}

The SATD echo-lune stack and the single-shot compiler answer different
questions.  SATD uses the composite lune, phase echo, and resonant
lower-doublet counterdiabatic actuator to engineer the strongly
biased-erasure Regime-A channel.  Single-shot uses one cyclic
bright-state pulse with two $\Lambda$ legs plus scalar detuning to
generate universal $\mathrm{SU}(2)$ gates.  Its residual channel is
moderately biased and gate-dependent, so it is not used for
architecture-level fault-tolerance claims.

\subsection{SiV single-shot bright-state benchmark (Regime E)}
\label{app:siv_single_shot}

Regime~E applies the same bright-state compiler to the SiV orbital
$\Lambda$ manifold.
The benchmark uses only the two orbit-strain $\Lambda$ legs and a
synchronized scalar detuning; no lower-doublet SATD actuator is
assumed.
The orbital $T_1$ model sends auxiliary-state population equally into
the two logical states, producing in-subspace depolarization rather
than erasure-convertible leakage.
Orbital dephasing, charge noise, and actuator transfer-function
errors are not included.
The results should therefore be read as orbital-$T_1$-limited
universal-control benchmarks, not as a SiV biased-erasure or XZZX
channel claim.

Table~\ref{tab:siv_single_shot_sensitivity_sm} gives the full
sensitivity table generated by
\texttt{code/siv\_single\_shot.py}.
The primary millikelvin point is the
$\Omega_{\mathrm{peak}}=\SI{300}{\mega\hertz}$ row used in the main
text.
The \SI{100}{\mega\hertz} and \SI{500}{\mega\hertz} rows test the
expected speed dependence of the orbital-$T_1$ error, while the
\SI{4}{\kelvin} row uses the measured
$T_{1,\mathrm{orb}}=\SI{40}{\nano\second}$ temperature-stress value.

\begin{table*}[t!]
\caption{\label{tab:siv_single_shot_sensitivity_sm}
Supplemental Regime~E SiV single-shot sensitivity table.
All rows use the bright-state compiler with two orbit-strain
$\Lambda$ legs plus synchronized scalar detuning and no lower-doublet
SATD actuator.}
\begin{ruledtabular}
\begin{tabular}{llccccc}
Scenario & Gate & $\Omega_{\mathrm{peak}}$ & $\tgate$ &
$F_{\mathrm{eff}}$ & $F_{\mathrm{cond}}$ & Leakage \\
\colrule
mK, primary & $Z_\pi$ & \SI{300}{\mega\hertz} &
\SI{6.667}{\nano\second} & $99.5221\%$ & $99.7488\%$ &
$2.275\times10^{-3}$ \\
mK, primary & $X_\pi$ & \SI{300}{\mega\hertz} &
\SI{6.667}{\nano\second} & $99.5221\%$ & $99.7488\%$ &
$2.275\times10^{-3}$ \\
mK, primary & $X_{\pi/2}$ & \SI{300}{\mega\hertz} &
\SI{6.667}{\nano\second} & $99.6699\%$ & $99.7979\%$ &
$1.283\times10^{-3}$ \\
mK, primary & Generic $0.73\pi$ & \SI{300}{\mega\hertz} &
\SI{6.667}{\nano\second} & $99.5671\%$ & $99.7621\%$ &
$1.957\times10^{-3}$ \\
mK, slow & $Z_\pi$ & \SI{100}{\mega\hertz} &
\SI{20.000}{\nano\second} & $98.5894\%$ & $99.2362\%$ &
$6.540\times10^{-3}$ \\
mK, slow & $X_\pi$ & \SI{100}{\mega\hertz} &
\SI{20.000}{\nano\second} & $98.5894\%$ & $99.2362\%$ &
$6.540\times10^{-3}$ \\
mK, slow & $X_{\pi/2}$ & \SI{100}{\mega\hertz} &
\SI{20.000}{\nano\second} & $99.0231\%$ & $99.3906\%$ &
$3.707\times10^{-3}$ \\
mK, slow & Generic $0.73\pi$ & \SI{100}{\mega\hertz} &
\SI{20.000}{\nano\second} & $98.7214\%$ & $99.2788\%$ &
$5.631\times10^{-3}$ \\
mK, fast & $Z_\pi$ & \SI{500}{\mega\hertz} &
\SI{4.000}{\nano\second} & $99.7123\%$ & $99.8497\%$ &
$1.377\times10^{-3}$ \\
mK, fast & $X_\pi$ & \SI{500}{\mega\hertz} &
\SI{4.000}{\nano\second} & $99.7123\%$ & $99.8497\%$ &
$1.377\times10^{-3}$ \\
mK, fast & $X_{\pi/2}$ & \SI{500}{\mega\hertz} &
\SI{4.000}{\nano\second} & $99.8014\%$ & $99.8788\%$ &
$7.757\times10^{-4}$ \\
mK, fast & Generic $0.73\pi$ & \SI{500}{\mega\hertz} &
\SI{4.000}{\nano\second} & $99.7395\%$ & $99.8576\%$ &
$1.184\times10^{-3}$ \\
\SI{4}{\kelvin} stress & $Z_\pi$ & \SI{300}{\mega\hertz} &
\SI{6.667}{\nano\second} & $98.3299\%$ & $99.0883\%$ &
$7.685\times10^{-3}$ \\
\SI{4}{\kelvin} stress & $X_\pi$ & \SI{300}{\mega\hertz} &
\SI{6.667}{\nano\second} & $98.3299\%$ & $99.0883\%$ &
$7.685\times10^{-3}$ \\
\SI{4}{\kelvin} stress & $X_{\pi/2}$ & \SI{300}{\mega\hertz} &
\SI{6.667}{\nano\second} & $98.8425\%$ & $99.2743\%$ &
$4.363\times10^{-3}$ \\
\SI{4}{\kelvin} stress & Generic $0.73\pi$ &
\SI{300}{\mega\hertz} & \SI{6.667}{\nano\second} &
$98.4859\%$ & $99.1397\%$ & $6.619\times10^{-3}$ \\
\end{tabular}
\end{ruledtabular}
\end{table*}

\section{Bright-projector bus validation and phase-cycled biased-erasure closure}
\label{app:2q_projector_extension}

This appendix records the Regime-F architecture diagnostic.  It tests
whether the same symmetry-generated bright-state projector used in the
single-qubit compilers can also define an effective phonon-bus
entangler whose extracted channel remains phase-biased and compatible
with heralded leakage.  Regime~F is an effective $\Lambda$-level
validation of the projector-force model, not a microscopic
hardware-level two-qubit theorem.

\subsection{Projector-force identity}

For two defects coupled to a shared mechanical mode, consider
\begin{equation}
K_I(t) =
\left(f_1P_{b1}+f_2P_{b2}\right)
\left(ae^{-i2\pi\delta t}+a^\dagger e^{i2\pi\delta t}\right),
\end{equation}
where $K=H/h$ and $P_{bi}=|b_i\rangle\langle b_i|$ is the logical bright-state
projector.  Since $P_{b1}$ and $P_{b2}$ commute, the Magnus expansion
closes exactly.  At $T=1/\delta$ the displacement vanishes and
\begin{equation}
U(T)=\exp\left[
i\frac{2\pi}{\delta^2}(f_1P_{b1}+f_2P_{b2})^2
\right].
\end{equation}
After removing single-qubit phases, the nonlocal part is
\begin{equation}
U_{\rm ent}\sim
\exp\left[
i\frac{4\pi f_1f_2}{\delta^2}P_{b1}P_{b2}
\right].
\end{equation}
For $f_1=f_2=f$, the CZ-class condition is $f/\delta=1/2$, giving
$T_{\rm CZ}=1/(2f)$ in ordinary-frequency units.  Choosing
$P_b=(I+Z)/2$ gives
\begin{equation}
P_{b1}P_{b2}=\frac{1}{4}(I+Z_1+Z_2+Z_1Z_2),
\end{equation}
so the nonlocal part is locally equivalent to CZ after single-qubit
$Z$ corrections.

\subsection{\texorpdfstring{$\Lambda$}{Lambda}-level CZ extraction}

The validation embeds the two logical qubits in a two-qutrit
$\{|0\rangle,|1\rangle,|a\rangle\}^{\otimes2}$ space and couples the
logical bright projectors to a truncated oscillator.  The model
includes bus damping, logical dephasing, auxiliary-state relaxation,
and detected auxiliary leakage.  The main Regime-F point uses
$f=\SI{2.0}{\mega\hertz}$, $\delta=\SI{4.0}{\mega\hertz}$,
$Q=10^7$, $T=\SI{0.25}{\micro\second}$, $n_{\rm Fock}=16$,
$T_1=\SI{1}{\milli\second}$, $T_{1\rho}=\SI{0.5}{\milli\second}$,
and detection efficiency $\eta_{\rm det}=0.975$.

\begin{table}[t]
\caption{\label{tab:regime_f_channel_sm}Regime-F $\Lambda$-level
CZ-class channel extraction.  The Pauli components are the
target-frame conditional Pauli-twirl diagnostics after heralded
auxiliary leakage is separated.}
\begin{ruledtabular}
\begin{tabular}{lcc}
Quantity & Value & Percent \\
\colrule
Conditional $\Favg$ & $0.9980961930$ & $99.8096\%$ \\
Mean logical survival & $0.9995001250$ & $99.9500\%$ \\
$p_{\mathrm{era}}$ & $4.873782{\times}10^{-4}$ & $0.04874\%$ \\
$p_Z$ & $2.377468{\times}10^{-3}$ & $0.23775\%$ \\
$p_{ZZ}$ & $2.290797{\times}10^{-6}$ & $0.000229\%$ \\
$p_{XY}$ & $0$ & $0$ \\
Missed leakage & $1.249688{\times}10^{-5}$ & $0.00125\%$ \\
Bus nonvacuum mean & $9.458929{\times}10^{-5}$ & $0.00946\%$ \\
\end{tabular}
\end{ruledtabular}
\end{table}

The extracted channel is strongly phase-biased: in the tested
projector-force model the measurable non-erasure faults are $Z$ and
$ZZ$, while $XY$ terms vanish within numerical precision.  It is,
however, not erasure-dominated at this operating point, since
$p_Z/p_{\mathrm{era}}=4.88$.

\begin{table}[t]
\tiny
\caption{\label{tab:regime_f_cross_checks_sm}Regime-F cross-checks.
The nearly lossless row checks oscillator truncation; the $Q=10^6$
row is a bus-damping stress test.}
\begin{ruledtabular}
\begin{tabular}{lccccc}
Case & $\Favg$ & $p_{\mathrm{era}}$ & $p_Z$ & $p_{ZZ}$ & $p_{XY}$ \\
\colrule
Lossless & $0.99999784$ & $1.89{\times}10^{-16}$ &
$3.60{\times}10^{-6}$ & $0$ & $3.47{\times}10^{-17}$ \\
$Q=10^6$ & $0.98501275$ & $4.87{\times}10^{-4}$ &
$1.85{\times}10^{-2}$ & $2.57{\times}10^{-4}$ &
$2.78{\times}10^{-17}$ \\
\end{tabular}
\end{ruledtabular}
\end{table}

\subsection{Phase-cycled E-sector diagnostic}

The two-lune echo used for the SATD channel cancels the leading
$E$-sector geometric term.  Regime~F also tests the roots-of-unity
generalization
\begin{equation}
\phi_j=\phi_0+\frac{2\pi j}{N},\qquad j=0,\ldots,N-1 .
\end{equation}
Low azimuthal harmonics cancel according to
$\sum_j e^{im\phi_j}=0$ unless $m$ is divisible by $N$.  The numerical
open-lune diagnostic gives the log--log slopes in
Table~\ref{tab:phase_cycle_sm}.

\begin{table}[t]
\caption{\label{tab:phase_cycle_sm}E-sector phase-cycle slopes.  The
$E_{Sx}$, $N=4$ row is below the numerical fitting floor over the
sampled perturbation range.}
\begin{ruledtabular}
\begin{tabular}{lcccc}
Sector & $N=1$ & $N=2$ & $N=3$ & $N=4$ \\
\colrule
$E_{Sx}$ & $0.995$ & $2.995$ & $4.939$ & below floor \\
$E_{Sy}$ & $0.988$ & $2.994$ & $5.050$ & $6.986$ \\
\end{tabular}
\end{ruledtabular}
\end{table}

Thus the two-lune result reproduces cubic suppression, while the
three- and four-lune cycles suppress at least as strongly as the
simple $N+1$ heuristic in this diagnostic.  A general analytic law
requires an explicit harmonic inventory for the chosen path family.

\subsection{Repeated-syndrome XZZX proxy}

Finally, we ran a scheduled repeated-syndrome XZZX proxy using the
Regime-A single-qubit channel and the extracted Regime-F two-qubit
channel.  The proxy includes noisy repeated syndrome rounds,
measurement faults, known-location erasure weights, and scheduled
correlated $ZZ$ pair faults.  It is circuit-level only in this limited
sense; it is not a full Stim gate-by-gate detector-error-model circuit.

At scale one the physical rates used by the proxy are
$p_Z^{\rm data}=2.908734\times10^{-3}$,
$p_X^{\rm data}=4.0\times10^{-5}$,
$p_{\mathrm{era}}^{\mathrm{data}}=4.943689\times10^{-3}$,
$p_{ZZ}^{\rm pair}=2.290797\times10^{-6}$, and
$p_{\rm meas}=5.0\times10^{-4}$.  With $400$ trials per point, the
base-scale runs show no logical failures for distances $3$, $5$, and
$7$; the Wilson upper bound is $9.5\times10^{-3}$ for each zero-count
row.  Stress points at scale $30$ give $p_L=0.2725$, $0.2125$, and
$0.1750$ for distances $3$, $5$, and $7$, respectively.  At scale
$50$ the proxy is near the high-noise failure regime
($p_L\simeq0.5$).

\subsection{Force-budget reminder}

A simple off-resonant $\Lambda$ elimination gives the desired force but
also exposes the main feasibility constraint.  With
\begin{equation}
\frac{H}{\hbar}
=\Delta |a\rangle\langle a|
+\left[\frac{\Omega_c}{2}+g_0 a e^{-i\delta t}\right]
|a\rangle\langle b|+\mathrm{h.c.},
\end{equation}
adiabatic elimination of $|a\rangle$ gives
\begin{equation}
\frac{H_{\rm eff}}{\hbar}
=-\frac{|\Omega_c|^2}{4\Delta}P_b
-\frac{|g_0|^2}{\Delta}a^\dagger aP_b
-\left[
\frac{\Omega_c^\ast g_0}{2\Delta}ae^{-i\delta t}
+\mathrm{h.c.}\right]P_b .
\end{equation}
Thus $f=\Omega_c g_0/(2\Delta)$.  In the Rabi-rate convention used in
the main text, $g_0=\Omega_{\rm zpf}/2$, so
\begin{equation}
f=\frac{\Omega_c\Omega_{\rm zpf}}{4\Delta}
=\sqrt{p_a}\frac{\Omega_{\rm zpf}}{2},\qquad
p_a\simeq\left(\frac{\Omega_c}{2\Delta}\right)^2 .
\end{equation}

\begin{table}[t]
\tiny
\caption{\label{tab:2q_force_budget_sm}Zero-point-strain feasibility
estimates for the dispersive projector-force route using the SiV strain
susceptibility scale.}
\begin{ruledtabular}
\begin{tabular}{lcccc}
Geometry & $V_{\rm eff}$ & $\epsilon_{\rm zpf}$ &
$\Omega_{\rm zpf}$ & $F_{\max}$ \\
\colrule
HBAR disk & $3.14{\times}10^{-15}$ & $9.36{\times}10^{-11}$ & 0.122 MHz & 0.006 MHz \\
Ext.\ nanobeam & $2.50{\times}10^{-19}$ & $1.05{\times}10^{-8}$ & 13.6 MHz & 0.682 MHz \\
Conf.\ nanobeam & $3.00{\times}10^{-20}$ & $3.03{\times}10^{-8}$ & 39.4 MHz & 1.97 MHz \\
Aggr.\ nanobeam & $3.00{\times}10^{-21}$ & $9.58{\times}10^{-8}$ & 125 MHz & 6.23 MHz \\
\end{tabular}
\end{ruledtabular}
\end{table}

Here $F_{\max}$ assumes $p_a=1\%$ and $V_{\rm eff}$ is in
${\rm m}^3$.  Thus the fast $F=10$ MHz effective-model point is an
aspirational limit for extremely confined SiV nanophononics, larger
tolerated auxiliary admixture, or a different coupling mechanism.
HBAR-scale zero-point strain is too weak for this dispersive
projector-force route.

\subsection{Scope and future work}

Regime~F closes the architecture at the effective projector-force
level and defines the next validation targets.  The immediate future
work is to derive the projector force from a microscopic actuator and
then validate auxiliary leakage, off-resonant carrier excitation,
actuator-transfer distortion, bus damping, thermal occupation,
crosstalk, and regenerated transverse error channels.  The scheduled
two-sector diagnostic in Sec.~O tests the first detector-model layer
for the Regime-A biased-erasure channel; a calibrated
architecture-level detector-error model with microscopic actuator
timing remains the next step before making a device-level threshold
claim.

\FloatBarrier

\section*{Code and reproducibility statement}

The public reproduction package for this work is available at
\url{https://github.com/E-zClap/PhononQ}.  The repository contains the
curated manuscript source, figure PDFs, Python simulation modules,
control-pulse CSV, saved JSON sweep data in \texttt{m5/data/}, and
compact validation outputs in \texttt{code/build/}.  These files cover
the data products used for the main-text figures, supplemental tables,
biased-erasure channel extraction, QEC Monte Carlo diagnostics,
sector-injection tests, single-shot checks, SiC comparison, lab-frame
multitone validation, and the prospective two-qubit bus model.

The recommended lightweight reproduction path is to regenerate figures
from the saved data.  From the repository's \texttt{code/} directory,
the relevant commands are \texttt{python replot\_figures.py},
\texttt{python qec\_full\_simulation.py --plot-only}, and
\texttt{python qec\_channel\_sensitivity.py}.  Full simulation
re-execution is also supported through commands such as
\texttt{python main.py --quick}, selected \texttt{python main.py --sweep
1} jobs, and \texttt{python qec\_full\_simulation.py --toric}; these
jobs are slower and small Monte Carlo differences may occur across
library versions, hardware, or random seeds.  The repository file
\texttt{REPRODUCIBILITY.md} lists the principal replotting and
re-execution commands.

\clearpage
\twocolumngrid


\end{document}